\documentclass[aps,pra,twocolumn,groupedaddress,nofootinbib,notitlepage,showpacs,floatfix]{revtex4-1}
\usepackage{amsfonts}

\usepackage{bm,amssymb,amsmath,amsthm,natbib,mathtools,multirow,epsf,latexsym,setspace,array}

\usepackage[normalem]{ulem}
\usepackage{graphicx}
\usepackage{mathrsfs}
\usepackage{wrapfig}
\usepackage{boxedminipage}
\usepackage{epsfig}
\usepackage{setspace}
\usepackage{subfigure}
\usepackage{hyperref}
\usepackage{dsfont}
\usepackage{braket}
\usepackage[pdftex]{color}

\newcommand{\bes} {\begin{subequations}}
\newcommand{\ees} {\end{subequations}}
\newcommand{\bea} {\begin{eqnarray}}
\newcommand{\eea} {\end{eqnarray}}
\newcommand{\beq} {\begin{equation}}
\newcommand{\eeq} {\end{equation}}
\newcommand{\ba} {\begin{align}}
\newcommand{\ea} {\end{align}}

\def\>{\rangle}
\def\<{\langle}

\newcommand{\ignore}[1]{}
\newtheorem{NUDD}{NUDD Theorem}
\newtheorem{mylemma}{Lemma}
\newtheorem{mycorollary}{Corollary}
\newtheorem{mydef}{Definition}

\begin{document}

\title{Universality proof and analysis of generalized nested Uhrig dynamical decoupling}
\author{Wan-Jung Kuo}
\affiliation{Department of Physics and Astronomy, Center for Quantum Information Science \&
Technology, University of Southern California, Los Angeles, California
90089, USA}
\author{Gerardo Andres Paz-Silva}
\affiliation{Center for Quantum Information Science \& Technology, University of Southern
California, Los Angeles, California 90089, USA}
\author{Gregory Quiroz}
\affiliation{Department of Physics and Astronomy, Center for Quantum Information Science \&
Technology, University of Southern California, Los Angeles, California
90089, USA}
\author{Daniel A. Lidar}
\affiliation{Departments of Electrical Engineering, Chemistry, and Physics, and Center
for Quantum Information Science \& Technology, University of Southern
California, Los Angeles, California 90089, USA}

\begin{abstract}

Nested Uhrig dynamical decoupling (NUDD) is a highly efficient quantum error suppression scheme that builds on optimized single axis UDD sequences. We prove the universality of NUDD and analyze its suppression of different error types in the setting of generalized control pulses. We present an explicit lower bound for the decoupling order of each error type, which we relate to the sequence orders of the nested UDD layers. We find that the error suppression capabilities of NUDD are strongly dependent on the {parities} and relative magnitudes of all nested UDD sequence orders. This allows us to predict the optimal arrangement of sequence orders. We test and confirm our analysis using numerical simulations.
\end{abstract}

\maketitle

\section{Introduction}

One of the main obstacles in building a quantum computer is the inevitable coupling between a quantum system and its environment, or bath, which typically results in decoherence and leads to computational errors \cite{Breuer:book,NielsenChuang:book,Ladd:10}.  Adapted from nuclear magnetic resonance (NMR) refocusing techniques \cite{Hahn:50,Haeberlen:book}, dynamical decoupling (DD) \cite{Yang-DD-review} is a powerful open-loop technique that can be used to suppress decoherence by applying a sequence of short and strong pulses purely on the system to mitigate unwanted system-bath interactions.

Early schemes such as periodic DD (PDD) \cite{ViolaLloyd:98,Ban:98,Duan:98e,Zanardi:99,ViolaKnillLloyd:99}, using equidistant $\pi$-pulses, have been shown to suppress general system-bath interactions up to first order in time-dependent perturbation theory, with respect to the total sequence duration $T$, thereby achieving first-order decoupling. Concatenated DD (CDD) \cite{KhodjastehLidar:05,KhodjastehLidar:07}, which recursively embeds PDD into itself, was the first explicit scheme capable of achieving $N^{\rm th}$-order decoupling for general single-qubit decoherence, i.e., complete error suppression to order $T^{N}$, and has been amply tested in recent experimental studies \cite{Peng:11,AlvarezAjoyPengSuter:10,2010arXiv1011.1903T,2010arXiv1011.6417W,Barthel:10}.  However, the number of pulses required by CDD grows exponentially with decoupling order: $4^N$ pulses are required to attain $N$th order error suppression, which is not efficient enough to implement scalable quantum computing when $N$ is large.  For scalable quantum computing to be possible, it is desirable to design a DD sequence which is both accurate (high decoupling order) and efficient (small number of pulses).

For the single-qubit pure dephasing spin-boson model,  an optimal scheme called Uhrig DD (UDD) \cite{Uhrig:07,Uhrig:08}, achieves $N^{\rm th}$ order decoupling with the smallest possible number $N$ of ideal $\pi$-pulses, applied at non-equidistant pulse timings, 
\begin{equation}
t_{j}=T\sin ^{2}\frac{j\pi }{2(N+1)},\quad j\in\{1,\dots,\bar{N}\}
\label{UDDtiming}
\end{equation}
where $\bar{N}=N$ or $N+1$ depending on whether $N$ is even or odd. $N$ is also referred to as the sequence order.
The applicability of UDD extends beyond the spin-boson model to models with pure dephasing interactions, with generic bounded baths \cite{YangLiu:08,UL:10}, rendering  it model-independent, or ``universal". Moreover, it has also been proven to be applicable to analytically time-dependent Hamiltonians \cite{PasiniUhrig:10}. Extensive numerical and experimental studies of UDD performance can be found in { \cite{DharGroverRoy:06,LeeWitzelDasSarma:08,Cywinskietal:08,Biercuk2:09,Biercuk:09,Du:09,AlvarezAjoyPengSuter:10,Elizabeth:09,Barthel:10,WangDobrovitski:11,Almog1Davidson:11,YoungWhaley:11}.} Rigorous performance bounds for UDD were found in \cite{UL:10}, and for QDD and NUDD in \cite{Xia:11}.

While UDD is applicable to single- and two-axis interactions (e.g., pure dephasing and/or pure bit-flip), it cannot overcome general, three-axis single-qubit decoherence.
Quadratic DD (QDD) \cite{WestFongLidar:10} and its generalization Nested UDD (NUDD) \cite{WangLiu:11}, {which were proposed to tackle} general single and multi-qubit decoherence respectively, exploit the decoupling efficiency of UDD by nesting two- and multi-layer UDD, where sequence orders of different UDD layers can be different in order to address the dominant sources of error more efficiently. In \cite{WangLiu:11}, NUDD (including QDD) with even sequence orders on the inner levels was verified analytically to achieve arbitrary order decoupling with only a polynomial increase {in} the number of pulses over UDD, an exponential improvement over CDD or CUDD \cite{Uhrig:09} which combines orthogonal single-axis CDD and UDD sequences. A universality proof of NUDD with arbitrary sequence orders was given in \cite{JiangImambekov:11}. The same result, for QDD with arbitrary sequence orders, was proved independently in our previous work \cite{WanLidar:11}, using a different method. Our universality proof for QDD \cite{WanLidar:11} went beyond quantifying the overall performance of QDD, in that we provided a thorough analytical study of the suppression abilities of the inner- and outer-layer UDD sequences of QDD for each  single-axis error, and obtained the decoupling order of each single-axis error. We found the decoupling order dependence on the parities and relative magnitudes of the inner and outer sequence orders. However, the analysis for the performance of NUDD with arbitrary sequence orders on each individual error is still missing. A further restriction concerns the type of control pulses, which in general QDD or NUDD proofs \cite{JiangImambekov:11, WanLidar:11} {have so far been} limited to elements of SU(2) or tensor products thereof. Numerical studies of QDD and NUDD performance are given in \cite{QuirozLidar:11,WangDobrovitski:11,KHYMLEEKANG:06,PanXiGong:11,Mukhtar2:10}

In this paper, we give a rigorous and compact proof for the universality and performance of NUDD (including QDD) with arbitrary sequence orders. The set of control {pulse} types is generalized to the mutually  orthogonal operation set (MOOS) defined in \cite{WangLiu:11}. The concept of error types is also generalized, in the sense that errors are classified according to the types of control pulses chosen.  Most importantly, we obtain the explicit decoupling order formula,  a function of the given error type, and the parities and magnitudes of all the sequence orders of NUDD.  This formula shows explicitly how each UDD layer contributes to {the suppression of} a given type of error. An immediate consequence is  that the overall suppression order of NUDD with {a} MOOS  as the control pulse set is the minimum among all sequence orders of NUDD. Moreover, our analysis identifies the condition under which the suppression ability of a given UDD layer is being hindered, or rendered totally ineffective, or enhanced by other UDD layers with odd sequence orders.  One can thus design an NUDD scheme such that the full power of each UDD layer is fully exploited. Note that our proof also shows that the performance of NUDD schemes with generalized control pulse types and arbitrary sequence orders  is universal, as in previous UDD-like schemes with Pauli group elements as control pulses \cite{YangLiu:08,UL:10,WangLiu:11,JiangImambekov:11,WanLidar:11}. In other words, the performance of general NUDD sequence remains the same for multi-qubit or multi-level systems coupling to arbitrary bounded environments. We present numerical simulations in support of our analytical results for the decoupling order of each error type, for a four-layer NUDD scheme applied to a two-qubit system.

The structure of this paper is as follows.
In Sec.~\ref{sec:formulation}, a general NUDD scheme with MOOS as the control pulse set  is formulated. In  Sec.~\ref{sec: NUDD performance}, the results of the performance of NUDD are presented.   Specifically, we present the NUDD Theorem (theorem \ref{thm:NUDD}), which gives the decoupling order formula for each error type. Corollary \ref{col:NUDDoverall}, along with its proof, gives the overall performance of NUDD. The complete proof of the NUDD Theorem is presented in Sec.~\ref{sec:NUDDproof}.  Numerical results for {a 4-layer} NUDD scheme on 2 qubits system are demonstrated in Sec.~\ref{sec:numerical result}, in support of our analysis. We conclude in Sec.~\ref{sec:conclusion}. The appendixes provides additional  technical details.


\section{NUDD formulation}
\label{sec:formulation}

\subsection{The noise model}
We assume a completely general noise Hamiltonian $H$ acting on the joint system-bath Hilbert space, the only assumption being that $\|H\| < \infty$. We allow for arbitrary interactions between the system and the bath, as well as between different parts of the system or between different parts of the bath. We use
\begin{equation}
\|A\|\equiv \sup_{\ket{\psi}}\frac{\sqrt{\bra{\psi}A^{\dagger}A\ket{\psi}}}{\braket{\psi|\psi}},
\end{equation}
to denote the sup-operator norm of any operator $A$, i.e., the largest singular value of $A$, or the largest eigenvalue (in absolute value) if $A$ is Hermitian.

\subsection{NUDD pulse timing}
A general $\ell$-layer NUDD scheme with a sequence order set $\{N_{1},N_{2},\dots, N_{\ell}\}$ is constructed by concatenating $\ell$ levels of UDD sequences, where $N_{i}$ is the sequence order of the UDD$_{N_{i}}$ sequence at the $i^{\textrm{th}}$ level of NUDD.   The sequence orders of different UDD layers can assume different values in order to address the dominant sources of error in any particular implementation more efficiently. The control {pulse} operator set $\{\Omega_{1},\Omega_{2},\dots, \Omega_{\ell}\}$,  where the subscript $i$  is the layer index, is chosen to be the mutually orthogonal operation set (MOOS) defined in \cite{WangLiu:11}, which consists of independent, mutually commuting or anticommuting system operators, each of which is both unitary and Hermitian. Obviously, each $\Omega_{i}$ is required not to commute with the total Hamiltonian, for otherwise it would not have any effect on the noise. Note that the MOOS elements $\Omega_{i}$  are not restricted to be single-qubit system operators such as a single-qubit Pauli matrix.

The normalized $\ell$-layer NUDD pulse timing  $\eta_{\,j_{\ell},j_{\ell-1},\dots,j_{1}}$ is defined as the actual NUDD timing, divided by total evolution time $T$, where  $j_{i}\in\{1,\dots,N_{i}+1\}$ is called the $i^{\textrm{th}}$ layer UDD pulse timing index. With fixed $\{j_{k}\}_{k=i+1}^{\ell}$ and $\{j_{k}=N_{k}+1\}_{k=1}^{i-1}$, $\eta_{\,j_{\ell},j_{\ell-1},\dots,j_{1}}$ with $j_{i}$ running from $1$ to $N_{i}+1$   constitute one cycle of UDD$_{N_{i}}$, of total duration $s_{j_{\ell}}s_{j_{\ell-1}}\cdots s_{j_{i+1}}$, i.e.,
\begin{align}
&\eta_{\,j_{\ell},\cdots,j_{i+1},j_{i}}=\eta_{\,j_{\ell},\dots,j_{i+1},0}\notag \\ 
&\qquad + s_{j_{\ell}}s_{j_{\ell-1}}\cdots s_{j_{i+1}}\sin ^{2}\frac{j_{i}\pi }{2(N_{i}+1)}, 
\label{ith UDDtiming}
\end{align}
where 
\bes
\begin{eqnarray}
\eta_{\,j_{\ell},\dots,j_{i}}& \equiv & \eta_{\,j_{\ell},\dots,j_{i+1},j_{i},N_{i-1}+1,\dots,N_{1}+1} \\
& \equiv & \eta_{\,j_{\ell},\dots,j_{i+1},j_{i}+1,0,\dots,0}
\end{eqnarray}
\ees
with $\eta_{j_{\ell}=0}\equiv 0$, 
and 
\begin{equation}  
s_{j_{k}}=\sin \frac{\pi }{2(N_{k}+1)}\sin \frac{(2j_{k}-1)\pi }{2(N_{k}+1)} \label{s}\\
\end{equation}
is the $j_{k}^{\textrm{th}}$ pulse interval of the normalized UDD$_{N_{k}}$ sequence.

Accordingly, the $i^{\textrm{th}}$ level $\Omega_{i}$ pulses are applied at the timings $\eta_{\,j_{\ell},\dots,j_{i}}$ with $j_{i}=1,2,\dots, \overline{N}_{i}$, where $\overline{N}_{i}=N_{i}$ if $N_{i}$ even while $\overline{N}_{i}=N_{i}+1$ if $N_{i}$ odd, and 
$\{j_{k}\in \{1,\dots, N_{k}+1\}\}_{k=i+1}^{\ell}$. The additional pulse applied at the end of the sequence when $N_{i}$ is odd, is required in order to make the total number of $\Omega_{i}$ pulses even, so that the overall effect of the $\Omega_{i}$ pulses at the final time $T$ will be to leave the qubit state unchanged \cite{WestFongLidar:10,WanLidar:11}.


\subsection{$(r_{1},r_{2},\dots, r_{\ell})$-type error}

Each control operator  $\Omega_{i}$ can divide the total Hamiltonian $H$ into two parts: one {that} commutes with  $\Omega_{i}$ and another {that} anticommutes with $\Omega_{i}$. Hence, by the procedure provided in Appendix \ref{app: step H},  the total Hamiltonian can accordingly be divided into $2^{\ell}$ independent pieces by the given MOOS $\{\Omega_{1},\Omega_{2},\dots, \Omega_{\ell}\}$, as 
\begin{equation}
H=\sum_{\{r_{i}=0,1\}_{i=1}^{\ell}}H_{(r_{1},r_{2},\dots, r_{\ell})}
\label{Hr}
\end{equation} 
where 
\begin{equation}
 \begin{cases}
        \;    r_{i}=1              &  \Rightarrow  \{\Omega_{i}\,,H_{(r_{1},r_{2},\dots, r_{\ell})}\} =0 \\
         \;  r_{i}=0               & \Rightarrow   [\Omega_{i}\,,H_{(r_{1},r_{2},\dots, r_{\ell})}] =0
 \end{cases}
\label{c,anti-c}
\end{equation}
with $i=1, 2,\dots \ell$ 
and
\begin{equation}
\sum_{\{r_{i}=0,1\}_{i=1}^{\ell}}\equiv \sum_{r_{1}=0,1}\sum_{r_{2}=0,1}\dots\sum_{r_{\ell}=0,1}.
\label{r}
\end{equation}

$H_{(r_{1},r_{2},\dots, r_{\ell})}$ is classified as an $(r_{1},r_{2},\dots, r_{\ell})$-type error, a definition which includes all the operators that have the same commuting or anticommuting  relation ~\eqref{c,anti-c} as $H_{(r_{1},r_{2},\dots, r_{\ell})}$ with respect to a given  MOOS $\{\Omega_{1},\Omega_{2},\dots, \Omega_{\ell}\}$. In particular,  the $\vec{0}_{\ell}\equiv (0,0,\dots,0)$-type error, which commutes with all control pulses and hence is not suppressed by the NUDD sequence, is called a trivial error. All other error-types are non-trivial.

For example, suppose the first and second pulse types used in 2-layer NUDD  (namely, QDD) are $Z$ and $X$-type pulses, respectively. And the target quantum system we consider is a single qubit subjected to general decoherence, which can always be modeled as 
\begin{equation}
H=J_{0}I\otimes B_{0}+J_{x}\sigma _{x}\otimes B_{x}+J_{y}\sigma _{y}\otimes
B_{y}+J_{z}\sigma _{z}\otimes B_{z},
\label{H1qubit}
\end{equation}%
where $B_{\lambda}$, $\lambda \in \{0,X,Y,Z\}$, are arbitrary bath-operators, the Pauli matrices, $\sigma _{\lambda}$, $\lambda \in \{X,Y,Z\}$, are the unwanted errors acting on the system qubit, and $J_{\lambda}$, $\lambda \in \{0,X,Y,Z\}$, are bounded qubit-bath coupling coefficients.  Then the MOOS  $\{\sigma _{z}, \sigma _{x}\}$ divides the Hamiltonian \eqref{H1qubit} into four pieces: one trivial error $H_{(0,0)}=J_{0}I\otimes B_{0}$, and three non-trivial errors, $H_{(1,0)}=J_{x}\sigma_{x}\otimes B_{x}$,   $H_{(0,1)}=J_{z}\sigma_{z}\otimes B_{z}$, and $H_{(1,1)}=J_{y}\sigma_{y}\otimes B_{y}$.

\subsection{$(r_{1},r_{2},\dots, r_{\ell})$-type error modulation function}

Due to the discreteness of the pulse timings,  it is easier to perform the analysis in the toggling frame, i.e., the frame that rotates with the control pulses.  Up to $\pm 1$ factors, the normalized ($T=1$) control evolution operator is 
\begin{equation}
U_{c}(\eta)=\Omega_{\ell}^{j_{\ell}-1}\dots\Omega_{2}^{j_{2}-1} \Omega_{1}^{j_{1}-1}
\label{Uc} 
\end{equation}
when $\eta\in[\eta_{\,j_{\ell},j_{\ell-1},\dots,j_{1}-1},\eta_{\,j_{\ell},j_{\ell-1},\dots,j_{1}})$ and $U_{c}(1)=I$, the identity operator.

Note that the commuting and anticommuting relation, Eq.~\eqref{c,anti-c}, can be reformulated equivalently as follows, 
\begin{equation}
\Omega_{i}\,H_{\{r_{1},r_{2},\dots, r_{\ell}\}}\,\Omega_{i}=(-1)^{r_{i}} H_{(r_{1},r_{2},\dots, r_{\ell})},
\label{c,anti-c 2}
\end{equation}
for $i$ from $1$ to $\ell$,  by using  the unitary Hermitian property $\Omega_{i}^{2}=I$.

Hence, with Eqs.~\eqref{Uc} and ~\eqref{c,anti-c 2}, we obtain  the Hamiltonian in the toggling frame,
\begin{eqnarray}
\tilde{H}(\eta)&=&U^{\dagger}_{c}(\eta)H U_{c}(\eta)\notag\\
&=&\sum_{\{r_{i}=0,1\}_{i=1}^{\ell}} (f_{1})^{r_{1}}\dots (f_{\ell})^{r_{\ell}}H_{(r_{1},r_{2},\dots, r_{\ell})}
\label{H}
\end{eqnarray}
where 
\begin{equation}
f_{i}(\eta)=(-1)^{j_{i}-1} \qquad  \eta \in [\eta_{\,j_{\ell},\dots,j_{i}-1}, \eta_{\,j_{\ell},\dots,j_{i}}) 
\label{f1}
\end{equation}
is the normalized $i^{\textrm{th}}$-layer modulation function for an $\ell$-layer NUDD sequence, which switches sign only when the $i^{\textrm{th}}$ layer UDD$_{N_{i}}$ pulse index, $j_{i}$, changes. The factor $\prod_{i=1}^{\ell}f_{i}(\eta)^{r_{i}}$, i.e., the coefficient of the $(r_{1},r_{2},\dots, r_{\ell})$-type error $H_{(r_{1},r_{2},\dots, r_{\ell})}$ in Eq.~\eqref{H},  is called an $(r_{1},r_{2},\dots, r_{\ell})$-type error modulation function.

Note that since 
\begin{equation}
\prod_{p=1}^{n}f_{i}(\eta)^{r_{i}^{\,(p)}}=f_{i}(\eta)^{\sum_{p=1}^{n}r_{i}^{\,(p)}}=f_{i}(\eta)^{\oplus_{p=1}^{n}r_{i}^{\,(p)}}
\label{Z2}
\end{equation}
where $\oplus$ is the binary addition defined as ordinary integer addition followed by the modulo 2 operation,
$\{f_{i}^{r_{i}=0}, f_{i}^{r_{i}=1}\}$ forms a $Z_{2}$ group under ordinary multiplication. Likewise, for  $(r_{1},r_{2},\dots, r_{\ell})$-type error modulation functions, we have
\begin{equation}
\prod_{p=1}^{n}\prod_{i=1}^{\ell}(\,f_{i}(\theta)\,)^{r_{i}^{\,(p)}}=\prod_{i=1}^{\ell}(\,f_{i}(\theta)\,)^{\oplus_{p=1}^{n}r_{i}^{\,(p)}},
\label{productf}
\end{equation} 
which indicates that the set of the $(r_{1},r_{2},\dots, r_{\ell})$-type error modulation functions   constitutes a $Z_{2}^{\otimes \ell}$ group.

\subsection{Dyson expansion of the evolution operator in the toggling frame}

With the shorthand notations, 
\begin{equation}
\vec{r}_{\ell}\equiv (r_{1},r_{2},\dots, r_{\ell})\in\{0,1\}^{\otimes \ell}
\end{equation}
and 
\begin{equation}
\sum_{\vec{r}_{\ell}}\equiv\sum_{\{r_{i}=0,1\}_{i=1}^{\ell}}
\end{equation}
where the subscript $\ell$ of the vector $\vec{r}_{\ell}$ indicates that $\vec{r}_{\ell}$ has $\ell$ components, the Dyson series expansion of  the evolution propagator in the toggling frame, 
\begin{equation}
\widetilde{U}(T)=\widehat{T}\exp [-iT\int_{0}^{1}\,\widetilde{H}(\eta)\,d\eta],
\end{equation}
reads
\begin{equation}
\widetilde{U}(T)=\sum_{n=0}^{\infty}\sum_{\{\vec{r}_{\ell}^{\,(p)}\}_{p=1}^{n}}(-iT)^{n}\prod_{p=1}^{n}H_{\vec{r}_{\ell}^{\,(p)}}F_{\oplus_{p=1}^{n}\vec{r}_{\ell}^{\,(p)}}
\label{U}
\end{equation}
where the operators in this expression are
\begin{equation}
H^{(n)}_{\vec{r}_{\ell}}\equiv \prod_{p=1}^{n}H_{\vec{r}_{\ell}^{\,(p)}}= H_{\vec{r}_{\ell}^{(n)}}H_{\vec{r}_{\ell}^{(n-1)}}\dots H_{\vec{r}_{\ell}^{(1)}}
\label{eq:Hrln}
\end{equation} 
(note the ordering), and the scalars are
\begin{equation}
F^{(n)}_{\vec{r}_{\ell}}\equiv F_{\oplus_{p=1}^{n}\vec{r}_{\ell}^{\,(p)}}\equiv \prod_{p=1}^{n}\int_{0}^{\eta^{(p+1)}}\prod_{i=1}^{\ell}f_{i}(\eta^{\,(p)})^{r_{i}^{\,(p)}} d\eta^{\,(p)},
\label{F}
\end{equation}
with $\eta^{(n+1)}=1$, and where $r_{i}^{\,(p)}$ is the $i^{\textrm{th}}$ component of $\vec{r}_{\ell}^{\,(p)}$. Our task will be to find the conditions under which these $F_{\oplus_{p=1}^{n}\vec{r}_{\ell}^{\,(p)}}$ coefficients vanish, which will dictate the decoupling orders of the NUDD sequence.

Since
\begin{eqnarray}
&&\Omega_{i}H^{(n)}_{\vec{r}_{\ell}}\Omega_{i}=\prod_{p=1}^{n}\Omega_{i}H_{\vec{r}_{\ell}^{\,(p)}}\Omega_{i}=\prod_{p=1}^{n}(-1)^{r_{i}^{(p)}}H_{\vec{r}_{\ell}^{\,(p)}}\notag\\
&&=(-1)^{\oplus_{p^{\prime}=1}^{n}r_{i}^{(p^{\prime})}}\prod_{p=1}^{n}H_{\vec{r}_{\ell}^{\,(p)}}
\end{eqnarray}
where the first equality is obtained by inserting $\Omega_{i}^{2}=I$ between any adjacent $H_{\vec{r}_{\ell}^{\,(p)}}$ and the second equality is obtained by using Eq.~\eqref{c,anti-c 2}, $\prod_{p=1}^{n}H_{\vec{r}_{\ell}^{\,(p)}}$ either commutes ($\oplus_{p=1}^{n}r_{i}^{(p)}=0$) or anticommutes ($\oplus_{p=1}^{n}r_{i}^{(p)}=1$) with each control pulses operator $\Omega_{i}$. In other words, $\prod_{p=1}^{n}H_{\vec{r}_{\ell}^{\,(p)}}$ belongs to an $\vec{r}_{\ell}$-type error whose component $r_{i}=\oplus_{p=1}^{n}r_{i}^{(p)}$, and can be denoted as either $H_{\oplus_{p=1}^{n}\vec{r}_{\ell}^{\,(p)}=\vec{r}_{\ell}}$, or $H^{(n)}_{\vec{r}_{\ell}}$ [as in Eq.~\eqref{eq:Hrln}] if we only care about the resulting error type $\vec{r}_{\ell}$. Therefore, the set of  $2^{\ell}$  $(r_{1},r_{2},\dots, r_{\ell})$-type errors $H^{(n)}_{\vec{r}_{\ell}}$ form a $Z_{2}^{\otimes \ell}$ group under multiplication.  

Note the key role played by
\beq
\vec{r}_{\ell} = (r_{1},r_{2},\dots, r_{\ell}),\qquad r_{i}=\oplus_{p=1}^{n}r_{i}^{(p)}.
\eeq
There are $2^\ell$ such vectors, and they completely classify all the summands in the Dyson series \eqref{U}.

\subsection{NUDD coefficients}
\label{NUDDcoefficient}

From Eq.~\eqref{F}, $F_{\oplus_{p=1}^{n}\vec{r}_{\ell}^{\,(p)}}$ is a normalized $n$-nested integral (total duration $T=1$), and the $\vec{r}_{\ell}^{\,(p)}$ in its subscript indicates that  the $p^{\textrm{th}}$ integrand of $F^{(n)}_{\oplus_{p=1}^{n}\vec{r}_{\ell}^{\,(p)}}$ is the 
 $\vec{r}_{\ell}^{\,(p)}=(r_{1}^{\,(p)},r_{2}^{\,(p)},\dots, r_{\ell}^{\,(p)})$-type error modulation function $\prod_{i=1}^{\ell}f_{i}(\eta^{\,(p)})^{r_{i}^{\,(p)}}$. According to $F_{\oplus_{p=1}^{n}\vec{r}_{\ell}^{\,(p)}}$'s associated error type $H_{\oplus_{p=1}^{n}\vec{r}_{\ell}^{\,(p)}=\vec{r}_{\ell}}$ (or $H^{(n)}_{\vec{r}_{\ell}}$ for short), we can also denote $F_{\oplus_{p=1}^{n}\vec{r}_{\ell}^{\,(p)}}$ as $F_{\oplus_{p=1}^{n}\vec{r}_{\ell}^{\,(p)}=\vec{r}_{\ell}}$ or $F^{(n)}_{\vec{r}_{\ell}}$ [Eq.~\eqref{F}],
and name it the $n^{\textrm{th}}$ order normalized  $\ell$-layer NUDD $\vec{r}_{\ell}$-type error coefficient. If no confusion can arise we will sometimes call $F_{\oplus_{p=1}^{n}\vec{r}_{\ell}^{(p)}=\vec{r}_{\ell}}$ an ``NUDD coefficient"  for short.

Let us consider a couple of examples.
From Eq.~\eqref{F}, the form of the $1$-layer NUDD $(\vec{r}_{1}=1)$-type error coefficients 
\begin{equation}
F_{\oplus_{p=1}^{n}\vec{r}_{1}^{(p)}=\vec{r}_{1}}=\prod_{p=1}^{n}\int_{0}^{\eta^{(p+1)}}f_{1}(\eta^{(p)})^{r_{1}^{(p)}} d\eta^{(p)}
\end{equation}
is exactly the same as the UDD coefficients appearing in \cite{UL:10} for the pure dephasing model where the dephasing term $\sigma_{z}\otimes B_{z}$ is our  $(\vec{r}_{1}=1)$-type  error and the Pauli matrix $\sigma_{x}$ is our control pulse operator.

Moreover,  the form of the $2$-layer NUDD $\vec{r}_{2}=(r_{1},r_{2})$-type error coefficients 
\begin{equation} 
F_{\oplus_{p=1}^{n}\vec{r}_{2}^{\,(p)}=\vec{r}_{2}}=\prod_{p=1}^{n}\int_{0}^{\eta^{(p+1)}}f_{1}(\eta^{(p)})^{r_{1}^{(p)}} f_{2}(\eta^{(p)})^{r_{2}^{\,(p)}} d\eta^{(p)}
\end{equation} 
is exactly the same as the  QDD coefficients defined in our earlier work \cite{WanLidar:11}, where $\sigma_{x}\otimes B_{x}$, $\sigma_{z}\otimes B_{z}$, and $\sigma_{y}\otimes B_{y}$ are our $(1,0)$-, $(0,1)$- and $(1,1)$-type errors, respectively, while the Pauli matrices $\sigma_{z}$ and $\sigma_{x}$ are our first-layer and second-layer control pulses, respectively.

Note that from Eq.~\eqref{F}, the NUDD coefficients are actually the same no matter what the control pulse operators are, as long as they are independent and constitute a MOOS. Therefore, the proofs for the performance of UDD and QDD sequence in \cite{UL:10} and \cite{WanLidar:11}, which used single qubit Pauli matrices as control pulses, apply directly to $1$-layer NUDD and $2$-layer NUDD schemes with more general control pulses.


\section{Performance of the NUDD sequence}
\label{sec: NUDD performance}

\subsection{The decoupling order of each error type }
\label{sec: decoupling order of each error}

As we observed, the summands in the Dyson series expansion of the evolution propagator Eq.~\eqref{U} are each classified as one of $2^{\ell}$ types of errors. Therefore, for a given $\vec{r}_{\ell}$-type error,  if all of its first $\check{N}_{\vec{r}_{\ell}}$ NUDD coefficients vanish, then we say that  the $\ell$-layer NUDD sequence eliminates the $\vec{r}_{\ell}$-type error to  order $\check{N}_{\vec{r}_{\ell}}$, i.e., $\check{N}_{\vec{r}_{\ell}}$ is the decoupling order of  the $\vec{r}_{\ell}$-type error.   

Let us define 
\beq
[x]_{2} \equiv (x \!\!\!\!\mod 2)
\eeq
and
\bes
\bea
p_\oplus(a,b)&\equiv&\bigoplus_{k=a}^{b}r_{k}[N_{k}]_{2}\ {\textrm{if }b\geq a,\textrm{otherwise }0}\\
p_+(a,b)&\equiv&\sum_{k=a}^{b}r_{k}[N_{k}]_{2}\ {\textrm{if }b\geq a,\textrm{otherwise }0}
\eea
\ees
The value of $[N_{k}]_{2}$ indicates the parity of the sequence order $N_{k}$ of the $k^{\textrm{th}}$ UDD layer, i.e., $[N_{k}]_{2}=0$ or $1$ when $N_{k}$ is even or odd, respectively.
Hence $p_\oplus(1,i-1)\in\{0,1\}$
gives the parity of the total number of UDD layers,  each of which has odd sequence order and control pulses that anticommute with the $\vec{r}_{\ell}$-type error, in the first $i-1$ levels of NUDD. Likewise, $p_+(i+1,\ell)$ counts the total number of  those UDD layers after the $i^{\textrm{th}}$ UDD layer, with odd  sequence orders and control pulses anticommuting with the $\vec{r}_{\ell}$-type error. With this in mind, we shall prove the following theorem in Sec.~\ref{sec:NUDDproof}:

\begin{NUDD}
An $\ell$-layer  NUDD scheme with a given sequence order set $\{N_{1},N_{2},\dots, N_{\ell}\}$  eliminates $\vec{r}_{\ell}=(r_{1},r_{2},\dots, r_{\ell})$-type errors  to order $\check{N}_{\vec{r}_{\ell}}$, i.e. 
\begin{equation}
F^{(n)}_{\vec{r}_{\ell}}=0,\quad \forall\, n\leq \check{N}_{\vec{r}_{\ell}},
\label{Fl}
\end{equation}
with the decoupling order of the $\vec{r}_{\ell}$-type error being
\begin{equation}
\check{N}_{\vec{r}_{\ell}}= \max_{i\in\{1,\dots,\ell\}}[r_{i}(p_\oplus(1,i-1)\oplus 1)\widetilde{N}_{i}+p_+(i+1,\ell)]
\label{Nrl}
\end{equation}
where
\begin{equation}
\widetilde{N}_{i}=\begin{cases}
                      N_{i}                    &   \textrm{when } i\leq o_{1}\\
                   \min[ N^{k'<i}_{o_{\min}}+1,N_{i}] &   \textrm{when } i> o_{1},
                  \end{cases}
\label{Ni}
\end{equation}
and where the $o_{1}$th layer is the first UDD layer with odd sequence order denoted as $N_{o_{1}}$, and
\begin{equation}
{N^{k'<i}_{o_{\min}}\equiv\min\{N_{k'}\ |\ [N_{k'}]_{2}=1\}.}
\label{Nioddmin}
\end{equation}
\label{thm:NUDD}
\end{NUDD}

Let us proceed to explain this theorem before embarking on its proof. First, note that $N^{k'<i}_{o_{\min}}$ is simply the minimum value among all the odd sequence orders of the first $i-1$ UDD layers of NUDD.
As we shall see, $\widetilde{N}_{i}$ is the suppression order of the $i^{\textrm{th}}$ UDD layer, in contrast with the sequence order $N_i$ of the same layer.
By Eq.~\eqref{Ni},
$\widetilde{N}_{i}\leq N_i$.
When $\widetilde{N}_{i}<N_i$ occurs, we have $\widetilde{N}_{i}=N^{k'<i}_{o_{\min}}+1$, which suggests that the suppression order of the $i^{\textrm{th}}$ UDD layer is partially hindered by the UDD layer with the smallest odd sequence order, which is nested inside the $i^{\textrm{th}}$ UDD layer.

The coefficient in front of the suppression order $\widetilde{N}_{i}$ in Eq.~\eqref{Nrl} is $1$ if and only if both $r_{i}=1$ and $ p_\oplus(1,i-1)=0$ are satisfied, and vanishes otherwise. Accordingly, there are two requirements for the $i^{\textrm{th}}$ UDD layer to be effective on a given $\vec{r}_{\ell}$-type error. First, the error must anticommute with the control pulses of this UDD layer ($r_{i}=1$). Second, the error must anticommute with a total even number of odd order UDD layers among the first $(i-1)$ layers [$p_\oplus(1,i-1)=0$]. 

In contrast, if an error  anticommutes with a total odd number of odd order UDD layers among the first $(i-1)$ layers [$p_\oplus(1,i-1)=1$],  then the $i^{\textrm{th}}$-layer UDD sequence is totally ineffective in suppressing this error type, even though this error also anticommutes with the control pulses of the $i^{\textrm{th}}$ UDD layer. 

Note that for the trivial error type, the $\vec{0}_{\ell}$-type error,   we have $\check{N}_{\vec{0}_{\ell}}=0$, which by Eq.~\eqref{Fl}  implies $F^{(n)}_{\vec{0}_{\ell}}=0$  for $n \leq 0$. This should be interpreted as saying that the vanishing of the $\vec{0}_{\ell}$-type errors cannot be concluded from Theorem~\ref{thm:NUDD}.

In the following, we discuss the decoupling order formula Eq.~\eqref{Nrl} for some particular examples of NUDD schemes.


\subsubsection{1-layer NUDD (UDD)}

Since 1-layer NUDD  has $\ell=1$, for the decoupling order of $(\vec{r}_{1}=1)$-type error, Eq.~\eqref{Nrl} gives $\check{N}_{\vec{r}_{1}=1}=\widetilde{N}_{1}$. By definition of $\widetilde{N}_{1}$ [Eq.~\eqref{Ni}], the suppression order of the first UDD layer $\widetilde{N}_{1}$ is always equal to its sequence order $N_{1}$, i.e. $\widetilde{N}_{1}=N_{1}$,  irrespective of $N_{1}$'s parity. Therefore, we have 
\begin{equation}
\check{N}_{\vec{r}_{1}=1}=N_{1} 
\end{equation}
which shows that the non-trivial error is eliminated up to the UDD's sequence order, in agreement with \cite{YangLiu:08,UL:10}.

\subsubsection{2-layer NUDD (QDD)}

For 2-layer NUDD  with a sequence order set $\{N_1,N_2\}$, the decoupling order formula Eq.~\eqref{Nrl} simplifies to
\begin{equation}
\check{N}_{\vec{r}_{2}}= \max[r_{1}\widetilde{N}_{1}+r_{2}[N_{2}]_{2}, r_{2}\,(\,r_{1}[N_{1}]_{2}\,\oplus 1\,  )\widetilde{N}_{2}\,].
\label{qddformula}
\end{equation}
Specifically, we have
\begin{eqnarray}
&&\check{N}_{(1,0)}= \widetilde{N}_{1}=N_{1}\notag\\
&&\check{N}_{(0,1)}= \max[[N_{2}]_{2}, (\,0 \oplus 1\,  )\widetilde{N}_{2}\,]=\widetilde{N}_{2}\notag\\
&&\check{N}_{(1,1)}= \max[\widetilde{N}_{1}+[N_{2}]_{2}, (\,[N_{1}]_{2}\,\oplus 1\,  )\widetilde{N}_{2}\,].
\label{qddorder}
\end{eqnarray}
If we identify the $(1,0)$-type, $(0,1)$-type, and $(1,1)$-type errors to be $\sigma_{x}$, $\sigma_{z}$, and $\sigma_{y}$ errors, respectively, Eq.~\eqref{qddorder} agrees  exactly with the results obtained in our earlier QDD work \cite{WanLidar:11}.

\subsubsection{NUDD with all even sequence orders}

For $\ell$-layer NUDD with  $[N_{i}]_{2}=0$  $\forall\, i\leq \ell$,   the suppression order of each UDD layer is equal to its corresponding sequence order $N_{i}$, $\widetilde{N}_{i}=N_{i}$ [Eq.~\eqref{Ni}]. Therefore, the decoupling order formula Eq.~\eqref{Nrl} for this type of NUDD schemes reduces to
\begin{equation}
\check{N}_{\vec{r}_{\ell}}=\max_{i\in\{1,\dots,\ell\}}[\,r_{i}\widetilde{N}_{i}]=\max_{i\in\{1,\dots,\ell\}}[\,r_{i}{N}_{i}],
\label{Nalleven}
\end{equation}
which indicates that the suppression ability of the UDD sequence in each layer is unaffected by the other nested UDD sequences, i.e., successive UDD layers do not interfere with one another. 

\subsubsection{NUDD with all even sequence orders except the outer-most UDD layer}

For $\ell$-layer NUDD with  $[N_{i}]_{2}=0$ $\forall\, i<\ell$, and $[N_{\ell}]_{2}=1$,   we again have $\widetilde{N}_{i}=N_{i}$ [Eq.~\eqref{Ni}]. Therefore, the decoupling order formula  Eq.~\eqref{Nrl} reads
\begin{equation}
\check{N}_{\vec{r}_{\ell}}=\max[\{\,r_{i}N_{i}\oplus r_{\ell}\}_{i=1}^{\ell-1},r_{\ell}N_{\ell}].
\label{Nreven}
\end{equation}
From Eq.~\eqref{Nreven}, we can see that if a given error anticommutes with the outer-most UDD layer, namely, $r_{\ell}=1$, then the outer-most layer boosts the suppression abilities of all inner UDD layers by one additional order. We call this the outer-odd-UDD effect.


\subsubsection{NUDD with all odd sequence orders}

For any NUDD scheme with all odd sequence orders, in general, $\widetilde{N}_{i}=N^{k'<i}_{o_{\min}}+1<N_{i}$ could occur, which suggests that the suppression ability of the $i^{\rm th}$ UDD layer is hindered by one of the odd order UDD layers inside the $i^{\textrm{th}}$ layer.

Nevertheless, for the case of $\ell$-layer NUDD with a $[N_{i}]_{2}=1$ $\forall\,i$ \emph{and} $N_{1}>N_{2}>\dots>N_{\ell}$, $\widetilde{N}_{i}=N_{i}$ [Eq.~\eqref{Ni}] is guaranteed . Moreover, it is easy to show that 
\begin{equation}
N_{i}+\sum_{k=i+1}^{\ell}r_{k}>N_{j}+\sum_{k=j+1}^{\ell}r_{k}
\label{Ni>Nj}
\end{equation}
for any $i<j$. Accordingly, due to Eq.~\eqref{Ni>Nj}, it turns out that  the maximum  value in Eq.~\eqref{Nrl} occurs at the inner-most UDD layer that the error anticommutes with. Suppose the first  non-zero component of an $\vec{r}_{\ell}$-type error is $r_{M}=1$. Then it follows that the $M^{\rm th}$ UDD layer has the maximum suppression on this error, i.e. 
\begin{equation}
\check{N}_{\vec{r}_{\ell}}=N_{M}+\sum_{k=M+1}^{\ell}r_{k}.
\label{formula:odd2}
\end{equation}
The second term of the above equation implies that the outer-odd-UDD suppression effects {generated} by  UDD layers with odd sequence orders outside the $M^{\rm th}$ layer all add up.  In other words, the outer-odd-UDD suppression effect is cumulative. 


\subsection{The overall performance of NUDD scheme}
\label{sec: overall performance}

The overall performance of an $\ell$-layer  NUDD scheme is quantified by the minimum over the decoupling orders of all error types, i.e., by
\beq
\check{N}_{\min} \equiv \min_{\{\vec{r}_{\ell}\}}[\check{N}_{\vec{r}_{\ell}}].
\label{eq:Nmin}
\eeq
We call $\check{N}_{\min}$ the overall decoupling order. The following corollary of the NUDD Theorem states the relationship between the overall decoupling order $\check{N}_{\min}$ and given sequence orders $\{N_i\}$:
\begin{mycorollary}
The overall decoupling order $\check{N}_{\min}$ of an $\ell$-layer NUDD scheme with a given sequence order set $\{N_{1},N_{2},\dots, N_{\ell}\}$ is 
\begin{equation}
\check{N}_{\min}=\min_{i\in\{1,\dots,\ell\}}[N_{i}]
\label{Nmin}
\end{equation}
\label{col:NUDDoverall}
\end{mycorollary}

\begin{proof}[Proof of Corollary \ref{col:NUDDoverall}]
For a given error type $\vec{r}_{\ell}\neq \vec{0}_{\ell}$, suppose the inner-most non-zero component is $r_{i}=1$ 
{with $i\geq 2$, i.e.,}
$r_{k<i}=0$, 
so that $p_\oplus(1,i-1)=0$. {By definition for $i=1$ we also have $p_\oplus(1,i-1)=0$}. Due to  $r_{i}(p_\oplus(1,i-1)\oplus 1)=1$ and  $\check{N}_{\vec{r}_{\ell}}$ taking the maximum value over the set displayed in Eq.~\eqref{Nrl}, it follows that
\begin{equation}
\check{N}_{\vec{r}_{\ell}} \geq r_{i}(p_\oplus(1,i-1)\oplus 1)\widetilde{N}_{i}=\widetilde{N}_{i}
\label{decoupling>suppression}
\end{equation} 
Among $\vec{r}_{\ell}\neq \vec{0}_{\ell}$ error types, there are error types which anti-commute with only one control pulse type denoted as  $\vec{e}_{i}$ with $r_{i}=1$ and all $r_{j\neq i}=0$. According to Eq.~\eqref{Nrl}, the decoupling order of {the} $\vec{e}_{i}$-type error is 
\begin{equation}
\check{N}_{\vec{e}_{i}}=\widetilde{N}_{i}.
\label{Nei} 
\end{equation}
Owing to Eq. \eqref{decoupling>suppression} and  \eqref{Nei}, $\check{N}_{\min}$, the minimum among decoupling orders of all non-trivial error types [Eq.~\eqref{eq:Nmin}], occurs among the decoupling orders of all $\vec{e}_{i}$-type errors, i.e., 
\begin{equation}
\check{N}_{\min}\equiv\min_{\{\vec{r}_{\ell}\}}[\check{N}_{\vec{r}_{\ell}}]=\min_{i\in\{1,\dots,\ell\}}[\check{N}_{\vec{e}_{i}}]=\min_{i\in\{1,\dots,\ell\}}[\widetilde{N}_{i}]
\end{equation}
Suppose that among the $\ell$ layers, the $o_{1}^{\rm th}$, $o_{2}^{\rm th}$, $\dots$, and $o_{b}^{\rm th}$ UDD layers where $0<o_{1}<o_{2}<\dots<o_{b}\leq \ell$ are the layers with odd sequence orders.
By the definition of $\widetilde{N}_{i}$ [Eq.~\eqref{Ni}], for $i\leq o_{1}$, $\widetilde{N}_{i}=N_{i}$, while for $o_{j-1}<k\leq o_{j}$ with $j\in \{{2},\dots, {b}\}$, $\widetilde{N}_{k}=\min[N_{o_{1}}+1,N_{o_{2}}+1, \dots, N_{o_{j-1}}+1, N_{k}]$.
Then it follows that 
\begin{equation}
\min_{i\in\{1,\dots,\ell\}}[\widetilde{N}_{i}]=\min_{i\in\{1,\dots,\ell\}}[{N}_{i}]
\label{eq:NtildeN}
\end{equation}
which proves Eq.~\eqref{Nmin}.
\end{proof}



\section{Proof of the NUDD Theorem}
\label{sec:NUDDproof}

\subsection{Synopsis of the proof}
\label{sec: Synopsis proof}

Our proof is by induction. To establish the base case, we recall that (as already discussed in Sec.~\ref{NUDDcoefficient}) the proofs of the vanishing of UDD coefficients in \cite{UL:10} and  QDD coefficients in our earlier work \cite{WanLidar:11}, which used single qubit Pauli matrices as control pulses, are also valid for the 1-layer and the 2-layer NUDD schemes with a more general set of control pulses (MOOS). From \cite{UL:10} and \cite{WanLidar:11}, the decoupling orders of each error type for UDD and QDD match exactly with Eq.~\eqref{Nrl}.  Accordingly, the NUDD Theorem holds for the 1-layer and the 2-layer NUDD sequence.

Suppose that Theorem \ref{thm:NUDD} holds for $(\ell-1)$-layer NUDD with an arbitrary integer $\ell\geq 2$, i.e. 
\begin{equation}
F^{(n)}_{\vec{r}_{\ell-1}}=0,\quad \forall\, n \leq \check{N}_{\vec{r}_{\ell-1}}.
\label{F l-1}
\end{equation}
Then the remaining task is to show that Theorem \ref{thm:NUDD}---in particular Eqs.~\eqref{Fl} and ~\eqref{Nrl}---also holds for $\ell$-layer NUDD. The procedure is the following.

$F^{(n)}_{\vec{r}_{\ell}}$ can be re-expressed in two different forms, building on two different methods adapted from our earlier QDD proof \cite{WanLidar:11}: the outer-most-layer interval decomposition (``Method 1"), and the nested integral analysis with certain function types (``Method 2"). Each method allows us to show the vanishing of $F^{(n)}_{\vec{r}_{\ell}}$ for some error types, and when put together the two methods complete the proof.

In Sec.~\ref{sec: Fdecomposition}, using Method 1, the {$\ell$-layer} NUDD coefficient $F^{(n)}_{\vec{r}_{\ell}}$ is expressed in terms of the {$(\ell-1)$-layer} NUDD coefficients. Then, in Sec.~\ref{sec:inner performance}, we show that it is the vanishing of the {$(\ell-1)$-layer}  NUDD coefficients that makes the first $\check{N}_{\vec{r}_{\ell-1}}$ orders of the $\ell$-layer NUDD coefficient $F^{(n)}_{\vec{r}_{\ell}}$ vanish, i.e.,
\begin{equation}
F^{(n)}_{\vec{r}_{\ell}=(\vec{r}_{\ell-1}, r_{\ell})}=0,\quad \forall\, n \leq \check{N}_{\vec{r}_{\ell-1}}
\label{Finner1}
\end{equation} 
due to the inductive assumption Eq.~\eqref{F l-1}.

Further, using the outer-layer interval decomposition form of $F^{(n)}_{\vec{r}_{\ell}}$, we  show in Sec.~\ref{sec:oddeffect} that for the $(\vec{r}_{\ell-1},1)$-type error, which anticommutes with the control pulses of the $\ell^{\textrm{th}}$ UDD layer, if the $\ell^{\textrm{th}}$ sequence order is odd then due to the anti-symmetry of the $\ell^{\textrm{th}}$ UDD layer, the first $\ell-1$ UDD layers in fact suppress the error by one more order, namely,  
\begin{equation}
F^{(\check{N}_{\vec{r}_{\ell-1}}+1)}_{(\vec{r}_{\ell-1},1)}=0.
\label{F N+1}
\end{equation}

Combining Eq.~\eqref{F N+1} with Eq.~\eqref{Finner1}, Method 1 gives rise to
\begin{equation}
F^{(n)}_{\vec{r}_{\ell}}=0,\quad \forall\, n \leq \check{N}_{\vec{r}_{\ell-1}}+r_{\ell}[N_{\ell}]_{2}.
\label{Finner2}
\end{equation} 

Let us now define two special cases, by parity:
\bea
\vec{v}_{i}^{a}&\equiv& \vec{r}_{i} \textrm{ with  }p_\oplus(1,i)=a,\quad a\in\{0,1\}
\label{def:errortype}
\eea
The decoupling order of the $(\vec{v}_{\ell-1}^{\,0},1)$-type error which by Eq.~\eqref{def:errortype} includes the $(\vec{0}_{\ell-1},1)$-type error, is derived independently by Method 2 in subsections~\ref{sec: Fourier}-\ref{sec:outermost error}.

The second part of the NUDD proof  is summarized as follows. In Sec.~\ref{sec: Fourier}, we apply a piecewise linear change of variables to $F^{(n)}_{\vec{r}_{\ell}}$, such that the Fourier expansions of all its integrands, $\prod_{i=1}^{\ell}f_{i}(\eta)^{r_{i}}$, belong to certain specific function types. In particular, for NUDD with all even inner sequence orders, there are only two function types, $\{c^{N_{\ell},0}_{\rm even},\zeta^{N_{\ell},0}_{\rm odd}\}$, while for NUDD with at least one odd sequence order, there are four types, $\{c^{N_{\ell},0}_{\rm even},\zeta^{N_{\ell},0}_{\rm odd},\zeta^{N_{\ell},0}_{\rm even},c^{N_{\ell},0}_{\rm odd}\}$. These function types are defined as follows:
\begin{mydef} 
Let $k,q \in \mathbb{Z}$ with $|q|\leq n$ and let $d_{ki}\in \mathbb{R}$ be arbitrary. 
\begin{align*}
&\textrm{name}\,\, &&\quad  \textrm{function type} \\
&c^{N_{\ell},n}_{\rm odd} \quad\,\,&&\quad \sum_{k}d_{k1}\cos[(2k+1)(N_{\ell}+1)\theta + q \theta ] \\
&c^{N_{\ell},n}_{\rm even}\quad\,\,\,&&\quad  \sum_{k}d_{k2}\cos[2k(N_{\ell}+1)\theta + q \theta ]\\
&\zeta^{N_{\ell},n}_{\rm even}\quad\,\,&&\quad   \sum_{k}d_{k3}\sin[2k(N_{\ell}+1)\theta + q \theta] \\
&\zeta^{N_{\ell},n}_{\rm odd} \quad&&\quad   \sum_{k}d_{k4}\sin[(2k+1)(N_{\ell}+1)\theta + q \theta]
\end{align*}
Also, let $\diamond$ denote the product-to-sum trigonometric function operation.
\label{def:type}
\end{mydef}

By utilizing the group properties of  $\{c^{N_{\ell},0}_{\rm even},\zeta^{N_{\ell},0}_{\rm odd}\}$ (or $\{c^{N_{\ell},0}_{\rm even},\zeta^{N_{\ell},0}_{\rm odd},\zeta^{N_{\ell},0}_{\rm even},c^{N_{\ell},0}_{\rm odd}\}$), after integrating out the first $n-1$ integrals of $F^{(n)}_{\vec{r}_{\ell}}$, we obtain the second expression for $F^{(n)}_{\vec{r}_{\ell}}$,
\begin{eqnarray}
F^{(n)}_{(\vec{v}_{\ell-1}^{\,0},1)}&=& \int_{0}^{\pi}\, c^{N_{\ell},n}_{\rm odd}\, d\theta \notag\\
&=&\sum_{k}d_{k}\sin[(2k+1)(N_{\ell}+1)\theta + q \theta]|_{0}^{\pi} \notag\\
&=&0.
\label{outer error}
\end{eqnarray}
for the $(\vec{v}_{\ell-1}^{\,0},1)$-type error,  where $n \leq N_{\ell}$ for NUDD with all even inner sequence orders, and $n \leq \min[ N^{k'<\ell}_{o_{\min}}+1,N_{\ell}]$ for  NUDD with at least one odd inner sequence order. {The vanishing orders of the other types of errors cannot be deduced from this method and we denote this fact by: $F^{(n)}_{\vec{r}_{\ell}\notin(\vec{v}_{\ell-1}^{\,0},1)}=0,\quad \forall\, n \leq 0$.} In conclusion, Method 2 yields     
\begin{equation}
F^{(n)}_{\vec{r}_{\ell}}=0,\quad \forall\, n \leq r_{\ell}[( p_\oplus(1,\ell-1)\oplus 1]\widetilde{N}_{\ell}.
\label{Fouter1}
\end{equation}

Combining  Eq.~\eqref{Finner2} with Eq.~\eqref{Fouter1}, we have 
\begin{align}
\label{Nrl2}
&F^{(n)}_{\vec{r}_{\ell}}=0\\
&\quad \forall n\leq
\max[\check{N}_{\vec{r}_{\ell-1}}+r_{\ell}[N_{\ell}]_{2}, r_{\ell}((p_\oplus(1,\ell-1)\oplus 1)\widetilde{N}_{\ell}].\notag
\end{align}
Equation~\eqref{Nrl2} is equivalent to the decoupling order formula Eq.~\eqref{Nrl} given in Theorem \ref{thm:NUDD}, as can be seen by substituting the explicit expression for $\check{N}_{\vec{r}_{\ell-1}}$, i.e., Eq.~\eqref{Nrl} for $(\ell-1)$-layer NUDD, into Eq.~\eqref{Nrl2}. Therefore, Theorem~\ref{thm:NUDD} also holds for the $\ell$-layer NUDD scheme with arbitrary sequence orders. The induction method also implies that Theorem~\ref{thm:NUDD}  holds for any number of nested layers of NUDD.


\subsection{The outer-layer interval decomposition form of $F^{(n)}_{\vec{r}_{\ell}}$}
\label{sec: Fdecomposition}

It is expected that the $(\vec{r}_{\ell-1}\neq\vec{0}_{\ell-1},r_{\ell})$-type errors, which anticommute with one or more of the first $(\ell-1)$  inner-layer control operators, are suppressed mainly by the inner $(\ell-1)$-layer NUDD sequences. In order to extract the action of the inner $(\ell-1)$-layer NUDD scheme and factor out  the action of the outer-most ($\ell^{\textrm{th}}$-layer) UDD sequence, we employ the outer-layer interval decomposition method, which splits each integral of $F^{(n)}_{\vec{r}_{\ell}}$ [Eq.~\eqref{F}] into a sum of sub-integrals over the $\ell^{\textrm{th}}$-layer UDD$_{N_{\ell}}$ pulse intervals $s_{j_{\ell}}$ [Eq.~\eqref{s}].

The derivation of $F^{(n)}_{\vec{r}_{\ell}}$ given in Appendix \ref{app: outer decomposition} shows that each  $\ell$-layer NUDD coefficient can be expressed in terms of  the ($\ell-1$)-layer NUDD coefficients as follows,
\begin{eqnarray}
&&F^{(n)}_{\vec{r}_{\ell}=(\vec{r}_{\ell-1},r_{\ell})}\equiv F_{\oplus_{p=1}^{n}(\vec{r}_{\ell-1}^{\,(p)},r_{\ell}^{\,(p)})}
=\sum_{m=1}^{n}\sum_{\{\sum_{a=1}^{m}n_{a}=n\}}\notag\\
&&(\prod_{a=1}^{m}F^{(n_{a})}_{\vec{r}_{\ell-1}^{\,\<a\>}})
\,(\prod_{a=1}^{m}\sum_{j_{\ell}^{(a)}=a}^{j_{\ell}^{(a+1)}-1}(-1)^{(j_{\ell}^{(a)}-1)r_{\ell}^{\,\<a\>}}\,s_{j_{\ell}^{(a)}}^{\,n_{a}})
\label{Fdecomposition}
\end{eqnarray}
where $j_{\ell}^{(m+1)}-1\equiv N_{\ell}+1$, and for given $m$, $\sum_{\{\sum_{a=1}^{m}n_{a}=n\}}$  sums over all possible combinations of $\{n_{a}\}_{a=1}^{m}$ under the condition 
$\sum_{a=1}^{m}n_{a}=n$,
with $1\leq n_{a}\leq n$. From the derivation of Eq.~\eqref{Fdecomposition} in Appendix \ref{app: outer decomposition}, one can see that each segment with a specific configuration $\{m, \{n_{a}\}_{a=1}^{m}\}$ defines a unique way to separate the $n$ vectors $\{\vec{r}_{\ell}^{\,(p)}\}_{p=1}^{n}$ into $m$ clusters. The $a^{\textrm{th}}$ cluster contains $n_{a}$  vectors, $\{\vec{r}_{\ell}^{\,(p)}\}_{p=\sum_{k=1}^{a-1}n_{k}+1}^{\sum_{k=1}^{a-1}n_{k}+n_{a}}$  with $n_{0}\equiv0$, and $\vec{r}_{\ell}^{\,\<a\>}$ is the resulting vector of the $a^{\textrm{th}}$ cluster, i.e., 
\begin{equation}
\vec{r}_{\ell}^{\,\<a\>}\equiv \oplus_{p=\sum_{k=1}^{a-1}n_{k}+1}^{\sum_{k=1}^{a-1}n_{k}+n_{a}}\,\vec{r}_{\ell}^{\,(p)},
\end{equation}
which implies $\oplus_{a=1}^{m}\vec{r}_{\ell}^{\,\<a\>}=\vec{r}_{\ell}$, i.e.,
\bes
\begin{eqnarray}
&& \oplus_{a=1}^{m}\vec{r}_{\ell-1}^{\,\<a\>}=\vec{r}_{\ell-1}\label{r l-1}\\
&& \oplus_{a=1}^{m}r_{\ell}^{\,\<a\>}=r_{\ell} \label{r l}.
\end{eqnarray} 
\ees
Moreover, the $a^{\textrm{th}}$ cluster is associated with $F^{(n_{a})}_{\vec{r}_{\ell-1}^{\,\<a\>}}$, an $n_{a}^{\textrm{th}}$ order   ($\ell-1$)-layer $\vec{r}_{\ell}^{\,\<a\>}$-type error  NUDD coefficient,   and with $\sum_{j_{\ell}^{(a)}=a}^{j_{\ell}^{(a+1)}-1}(-1)^{(j_{\ell}^{(a)}-1)r_{\ell}^{\,\<a\>}}\,s_{j_{\ell}^{(a)}}^{\,n_{a}}$ which contains the outermost ($\ell^{\textrm{th}}$) UDD layer's information.

\subsection{The performance of the inner  ($\ell-1$)-layer NUDD}
\label{sec:inner performance}

We shall show that it is the vanishing of the inner $(\ell-1)$-layer  NUDD coefficients,  $F^{(n_{a})}_{\vec{r}_{\ell-1}^{\,\<a\>}}=0,\, \forall\, n_{a} \leq \check{N}_{\vec{r}_{\ell-1}^{\,\<a\>}}$, that make the first $\check{N}_{\vec{r}_{\ell-1}}$ orders of $F^{(n)}_{\vec{r}_{\ell}=(\vec{r}_{\ell-1},r_{\ell})}$ vanish.

First, we have the following lemma, with the proof given in Appendix \ref{app: decomplemma}:
\begin{mylemma}
$\check{N}_{\vec{r}_{\ell-1}} \leq \sum_{a=1}^{m}\check{N}_{\vec{r}_{\ell-1}^{\,\<a\>}}$, where  $\vec{r}_{\ell-1}=\oplus_{a=1}^{m}\vec{r}_{\ell-1}^{\,\<a\>}$.
\label{lem:innersum N}
\end{mylemma}
With Lemma \ref{lem:innersum N}, for $F^{(n)}_{\vec{r}_{\ell}=(\vec{r}_{\ell-1},r_{\ell})}$ with $n \leq \check{N}_{\vec{r}_{\ell-1}}$, each $\prod_{a=1}^{m}F^{(n_{a})}_{\vec{r}_{\ell-1}^{\,\<a\>}}$ with $\oplus_{a=1}^{m}\vec{r}_{\ell-1}^{\,\<a\>}=\vec{r}_{\ell-1}$ in Eq.~\eqref{Fdecomposition} satisfies
\begin{equation}
\sum_{a=1}^{m}n_{a} =n \leq \check{N}_{\vec{r}_{\ell-1}} \leq \sum_{a=1}^{m}\check{N}_{\vec{r}_{\ell-1}^{\,\<a\>}}.
\label{n<N}
\end{equation}  
From Eq.~\eqref{n<N}, it follows that there must exist at least one $a$ from $1$ to $m$ such that $n_{a}\leq \check{N}_{\vec{r}_{\ell-1}^{\,\<a\>}}$, because if $n_{a} > \check{N}_{\vec{r}_{\ell-1}^{\,\<a\>}}$ for all $a$, it leads to 
\begin{equation}
\sum_{a=1}^{m}n_{a}> \sum_{a=1}^{m}\check{N}_{\vec{r}_{\ell-1}^{\,\<a\>}}
\label{counterargument}
\end{equation}
which contradicts the assumption Eq.~\eqref{n<N}. Accordingly,  each $\prod_{a=1}^{m}F^{(n_{a})}_{\vec{r}_{\ell-1}^{\,\<a\>}}$ term of  $F^{(n)}_{\vec{r}_{\ell}=(\vec{r}_{\ell-1},r_{\ell})}$ with $n \leq \check{N}_{\vec{r}_{\ell-1}}$ contains at least one $F^{(n_{a})}_{\vec{r}_{\ell-1}^{\,\<a\>}}$ with $n_{a} \leq \check{N}_{\vec{r}_{\ell-1}^{\,\<a\>}}$, which turns out to be zero due to the assumption Eq.~\eqref{F l-1} for $(\ell-1)$ NUDD coefficients. Therefore, the first step  $F^{(n)}_{\vec{r}_{\ell}=(\vec{r}_{\ell-1},r_{\ell})}=0$ for all $n \leq \check{N}_{\vec{r}_{\ell-1}}$ Eq.~\eqref{Finner1} is proven.

From the proof,  one can see that the effect of the outer $\ell$-layer UDD sequence which is entirely contained in the outer part $\prod_{a=1}^{m}\sum_{j_{\ell}^{(a)}=a}^{j_{\ell}^{(a+1)}-1}(-1)^{(j_{\ell}^{(a)}-1)r_{\ell}^{\,\<a\>}}\,s_{j_{\ell}^{(a)}}^{\,n_{a}}$ in Eq.~\eqref{Fdecomposition} does not interfere with the elimination ability of the inner $(\ell-1)$-layer NUDD sequence.

Note that for the $\vec{r}_{\ell}=(\vec{0}_{\ell-1},1)$-type error, $N_{\vec{0}_{\ell-1}}=0$ by  Eq.~\eqref{Nrl}  gives rise to $F^{(n)}_{(\vec{0}_{\ell-1},1)}=0$ for all $n \leq N_{\vec{0}_{\ell-1}}=0$  [Eq.~\eqref{Finner1}]. However, this does not mean that  $F^{(n)}_{(\vec{0}_{\ell-1},1)}\neq 0$ in our context, but rather that  the vanishing of $F^{(n)}_{(\vec{0}_{\ell-1},1)}$ cannot be deduced by the just-mentioned method, which only considers the contribution of the first $(\ell-1)$ UDD layers. This makes sense because the $\vec{r}_{\ell}=(\vec{0}_{\ell-1},1)$-type error, which commutes with all the control pulses of the inner $(\ell-1)$ UDD layers, is supposed to be suppressed by the $\ell^{\rm \,th}$-layer UDD sequence.


\subsection{One more order suppression of  the $(\vec{r}_{\ell-1},1)$-type error}
\label{sec:oddeffect}

In the previous section, we proved that $F^{(n)}_{\vec{r}_{\ell}=(\vec{r}_{\ell-1},r_{\ell})}=0$ for all $n \leq \check{N}_{\vec{r}_{\ell-1}}$ [Eq.~\eqref{Finner1}]. Now we shall show that for $(\vec{r}_{\ell-1},1)$-type error, there is an additional order suppression, i.e., $F^{(\check{N}_{\vec{r}_{\ell-1}}+1)}_{(\vec{r}_{\ell-1},1)}=0$ [Eq.~\eqref{F N+1}], if $N_{\ell}$ is odd. 

First, we need the following lemma:
\begin{mylemma}
For $\sum_{a=1}^{m}n_{a}=n=\check{N}_{\vec{r}_{\ell-1}}+1$ with $m \geq 2$, there must exist at least one $a$ such that $n_{a}\leq \check{N}_{\vec{r}_{\ell-1}^{\,\<a\>}}$
\label{lem:N+1}
\end{mylemma}

\begin{proof}[Proof of Lemma \ref{lem:N+1}]
With Lemma \ref{lem:innersum N},  we have 
\begin{equation}
\sum_{a=1}^{m}n_{a}\leq \sum_{a=1}^{m}\check{N}_{\vec{r}_{\ell-1}^{\,\<a\>}}+1.
\label{sum na}
\end{equation}
For $m \geq 2$,  pick an arbitrary $n_{a'}$ among $\{n_{a}\}_{a=1}^{m}$. If $n_{a'}\leq \check{N}_{\vec{r}_{\ell-1}^{\<a'\>}}$, then Lemma \ref{lem:N+1} is true. On the other hand, if  $n_{a'}\geq \check{N}_{\vec{r}_{\ell-1}^{\<a'\>}}+1$,  subtracting $n_{a'}$ from  Eq.~\eqref{sum na} leads to 
\begin{equation}
\sum_{a\neq a'}n_{a} \leq \sum_{a=1}^{m}\check{N}_{\vec{r}_{\ell-1}^{\,\<a\>}}+1-n_{a}\leq \sum_{a\neq a'}\check{N}_{\vec{r}_{\ell-1}^{\,\<a\>}}.
\label{n-na}
\end{equation}
Accordingly, by the same counterargument {mentioned} in Sec.~\ref{sec:inner performance}, there must exist at least one $a''\neq a'$ such that $n_{a''}\leq \check{N}_{\vec{r}_{\ell-1}^{\,<a''>}}$, for otherwise we would obtain $\sum_{a\neq a'}n_{a}\geq \sum_{a\neq a'}\check{N}_{\vec{r}_{\ell-1}^{\,\<a\>}}$ which contradicts Eq.~\eqref{n-na}.
\end{proof}

With Lemma \ref{lem:N+1}, it follows that for the outer-layer decomposition form of $F^{(\check{N}_{\vec{r}_{\ell-1}}+1)}_{(\vec{r}_{\ell-1},1)}$ Eq.~\eqref{Fdecomposition}, in each  $\prod_{a=1}^{m}F^{(n_{a})}_{\vec{r}_{\ell-1}^{\,\<a\>}}$  with $m\geq 2$ and $\sum_{a=1}^{m}n_{a}=n=\check{N}_{\vec{r}_{\ell-1}}+1$, there must exist at least one $F^{(n_{a})}_{\vec{r}_{\ell-1}^{\,\<a\>}}$ with $n_{a} \leq \check{N}_{\vec{r}_{\ell-1}^{\,\<a\>}}$. Due to the induction assumption Eq.~\eqref{F l-1},  all $\prod_{a=1}^{m}F^{(n_{a})}_{\vec{r}_{\ell-1}^{\,\<a\>}}$  with $m\geq 2$ in Eq.~\eqref{Fdecomposition} vanish.
Therefore, only  the $\{m=1, n_{a}=n\}$ term, $F^{(\check{N}_{\vec{r}_{\ell-1}}+1)}_{\vec{r}_{\ell-1}}$,  remains in the outer-layer decomposition form of $F^{(\check{N}_{\vec{r}_{\ell-1}}+1)}_{\vec{r}_{\ell}}$. In particular, for the $(\vec{r}_{\ell-1},1)$-type error, we have
\begin{equation}
F^{(\check{N}_{\vec{r}_{\ell-1}}+1)}_{(\vec{r}_{\ell-1},1)}=F^{(\check{N}_{\vec{r}_{\ell-1}}+1)}_{\vec{r}_{\ell-1}}\,\sum_{j_{\ell}=1}^{N_{\ell}+1}(-1)^{(j_{\ell}-1)}\,s_{j_{\ell}}^{\,n}
\end{equation}
Despite $F^{(\check{N}_{\vec{r}_{\ell-1}}+1)}_{\vec{r}_{\ell-1}}$,  the outer part $\sum_{j_{\ell}=1}^{N_{\ell}+1}(-1)^{(j_{\ell}-1)}\,s_{j_{\ell}}^{\,n}$ turns out to be zero  due to the cancellation of opposite signs between $j_{\ell}$ and $N_{\ell} + 2-j_{\ell}$ when $N_{\ell}$ is odd, but otherwise equal terms $s_{j_{\ell}}=s_{N_{\ell}+2-j_{\ell}}$ (time-symmetric property of UDD) in the sum.  Therefore,  $F^{(\check{N}_{\vec{r}_{\ell-1}}+1)}_{(\vec{r}_{\ell-1},1)}=0$ and Eq.~\eqref{F N+1} is proven. Combining Eq.~\eqref{F N+1} with Eq.~\eqref{Finner1}, which we proved in the previous section, leads to $F^{(n)}_{\vec{r}_{\ell}}=0$ for $n \leq \check{N}_{\vec{r}_{\ell-1}}+r_{\ell}[N_{\ell}]_{2}$, i.e., Eq.~\eqref{Finner2}.

From the proof, one can see that  the anti-symmetry of  the $\ell^{\textrm{th}}$ outer-most UDD layer with odd sequence order can help  the inner $(\ell-1)$ NUDD sequences to suppress the $(\vec{r}_{\ell-1},1)$-type error by one more order.


\subsection{Fourier expansion after linear change of variables}
\label{sec: Fourier}

We shall complete the proof of Theorem \ref{thm:NUDD} by another approach which analyzes the nested integral  with certain Fourier function types. First apply a piecewise linear transformation, 
\begin{equation}
\theta =\frac{\pi }{N_{\ell}+1}(\frac{\eta-\eta_{j_{\ell}-1}}{s_{j_{\ell}}})+\frac{(j_{\ell}-1)\pi}{N_{\ell}+1} \quad \eta \in \lbrack \eta_{j_{\ell}-1},\eta_{j_{\ell}})
\label{linear}
\end{equation}
with $j_{\ell}\in\{1,\dots,N_{\ell}\}$, to the $n$-fold nested integral $F^{(n)}_{\vec{r}_{\ell}}$. With $d\eta=G_{1}(\theta )\,d\theta$, $F^{(n)}_{\vec{r}_{\ell}}=F_{\oplus_{p=1}^{n}\vec{r}_{\ell}^{\,(p)}}$ Eq.~\eqref{F} is reexppressed as
\begin{equation}
F^{(n)}_{\vec{r}_{\ell}}=\prod_{p=1}^{n}\int_{0}^{\theta^{(p+1)}} G_{1}(\theta^{\,(p)})\,\prod_{i=1}^{\ell}f_{i}(\theta^{\,(p)})^{r_{i}^{\,(p)}}  d\theta^{\,(p)}
\label{Ftheta}
\end{equation}
with 
\begin{eqnarray}
\theta^{(n+1)}&=&\pi,\\
f_{i}(\theta)&=&(-1)^{j_{i}-1} \qquad  \theta \in [\theta_{\,j_{\ell},\dots,j_{i}-1}, \theta_{\,j_{\ell},\dots,j_{i}}), \label{f2}\\
G_{1}(\theta)&=&\frac{N_{\ell}+1}{\pi}s_{j_{\ell}}\qquad \theta \in \lbrack \theta _{j_{\ell}-1},\theta _{j_{\ell}})
\end{eqnarray}
where $\theta_{\,j_{\ell},\dots,j_{i}}$ is the new pulse timing.

The advantage of using the piecewise linear transformation Eq.~\eqref{linear} is that all the modulation functions become periodic functions. Consequently, as explained in Appendix \ref{app:Fourier}, it turns out {that} the Fourier expansion of each {function appearing} in $F_{\oplus_{p=1}^{n}\vec{r}_{\ell}^{\,(p)}}$ belongs  to one of the function types from Definition \ref{def:type},    
\begin{eqnarray}
\Psi(f_{\ell}(\theta ))&=&\zeta^{N_{\ell},0}_{\rm odd}, \label{psifl}\\
\Psi(f_{i<\ell}(\theta ))&=& \begin{dcases}
                c^{N_{\ell},0}_{\rm even}\textrm{\quad  when $N_{i}$ even}  \label{psifi} \\
               \zeta^{N_{\ell},0}_{\rm even}\textrm{\quad  when $N_{i}$ odd}
               \end{dcases}\\
\Psi(G_{1}(\theta))&=& \zeta^{N_{\ell},1}_{\rm even}, \label{psiG}
\end{eqnarray}
where $\Psi$ maps a function to the function type of its Fourier expansion up to unimportant coefficients.

Note that the sets $\{c^{N_{\ell},0}_{\rm even},\zeta^{N_{\ell},0}_{\rm odd}\}$ and $\{c^{N_{\ell},0}_{\rm even},\zeta^{N_{\ell},0}_{\rm odd},\zeta^{N_{\ell},0}_{\rm even},c^{N_{\ell},0}_{\rm odd}\}$ constitute  $Z_{2}$ and $Z_{2}\times Z_{2}$ groups with $c^{N_{\ell},0}_{\rm even}$ as identity, respectively, under the binary operation $\diamond$ (the product-to-sum trigonometric formula). Then one can obtain the function types of the $\vec{r}_{\ell}$-type error modulation functions by employing the group algebra of $Z_{2}$ and $Z_{2}\times Z_{2}$ to
\begin{equation}
\Psi(\,\prod_{i=1}^{\ell}f_{i}(\theta )^{r_{i}}\,)=\diamond_{i=1}^{\ell}\Psi(\,f_{i}(\theta )^{r_{i}}\,).
\end{equation}
where $\Psi(f_{i}(\theta )^{0}=1)=c^{N_{\ell},0}_{\rm even}$ for all $i$. For example, for the $\vec{v}_{\ell-1}^{\,1}$-type error [see Eq.~\eqref{def:errortype}], its {$\diamond_{i=1}^{\ell-1}\Psi(\,f_{i}(\theta )^{r_{i}}\,)$ }would result in $ \zeta^{N_{\ell},0}_{\rm even}$, because the condition $\oplus_{k=1}^{\ell-1}r_{k}[N_{k}]_{2}=1$, that  the components of  $\vec{v}_{\ell-1}^{\,1}$ satisfy, implies that there is a  total odd number of $r_{i}=1$ associated with $[N_{i}]_{2}=1$ such that $\Psi(\,f_{i}(\theta )^{r_{i}}\,)=\zeta^{N_{\ell},0}_{\rm even}$. It turns out that $\Psi(\,\prod_{i=1}^{\ell}f_{i}(\theta )^{r_{i}}\,)$ is determined by only two values, $\oplus_{k=1}^{\ell-1}r_{k}[N_{k}]_{2}$ and $r_{\ell}$, (see Table \ref{table:modulation type}).
\begin{table}[htbp]
\begin{tabular}{|ccc|c|}
\hline
 \multicolumn{3}{|c|}{error type $\vec{r}_{\ell}$} & $\Psi(\,\prod_{i=1}^{\ell}f_{i}(\theta )^{r_{i}}\,)$ \\  \hline 
$(1)$  $(\vec{v}_{\ell-1}^{\,0},0)$:& $\oplus_{k=1}^{\ell-1}r_{k}[N_{k}]_{2}=0$, & $r_{\ell}=0$ & \,$c^{N_{\ell},0}_{\rm even}$ \\  
$(2)$ $(\vec{v}_{\ell-1}^{\,0},1)$:  & $\oplus_{k=1}^{\ell-1}r_{k}[N_{k}]_{2}=0$, & $r_{\ell}=1$ & \,$\zeta^{N_{\ell},0}_{\rm odd}$\\  
$(3)$ $(\vec{v}_{\ell-1}^{\,1},0)$: & $\oplus_{k=1}^{\ell-1}r_{k}[N_{k}]_{2}=1$, & $r_{\ell}=0$ & \, $\zeta^{N_{\ell},0}_{\rm even}$\\   
$(4)$ $(\vec{v}_{\ell-1}^{\,1},1)$:  & $\oplus_{k=1}^{\ell-1}r_{k}[N_{k}]_{2}=1$, & $r_{\ell}=1$ & \,$c^{N_{\ell},0}_{\rm odd}$ \\  \hline 
\end{tabular}
\caption{The first column classifies $2^{\ell}$ $\vec{r}_{\ell}$-type errors into four groups  by two values, $\oplus_{k=1}^{\ell-1}r_{k}[N_{k}]_{2}$ and $r_{\ell}$, where $\vec{v}_{\ell-1}^{\,0}$ and $\vec{v}_{\ell-1}^{\,1}$ are defined in Eq.~\eqref{def:errortype}. The second column shows the function types of the Fourier expansion of the $\vec{r}_{\ell}$-type error modulation functions. }
\label{table:modulation type}
\centering
\end{table}

Note that for the $\ell$-layer NUDD with $[N_{i}]_{2}=0\,\forall\,i\leq \ell-1$, $\oplus_{k=1}^{\ell-1}r_{k}[N_{k}]_{2}$ is zero for all $2^{\ell}$ $\vec{r}_{\ell}$-type errors.  Therefore, $\Psi(\,\prod_{i=1}^{\ell}f_{i}(\theta )^{r_{i}}\,)$ for all error types in this case {is}  either $\zeta^{N_{\ell},0}_{\rm odd}$  or $c^{N_{\ell},0}_{\rm even}$, the first two rows in Table \ref{table:modulation type}. On the other hand,  for the $\ell$-layer NUDD with at least one UDD layer with odd sequence order, there are four functions types as shown in Table \ref{table:modulation type}.

The following expression of $F^{(n)}_{\vec{r}_{\ell}}$ focuses on the function types of the integrands, while the coefficients in their Fourier expansion are unimportant in the proof,  
\begin{equation}
F^{(n)}_{\vec{r}_{\ell}}=\diamond _{p=1}^{n}\int_{0}^{\theta^{(p+1)}} d\theta^{\,(p)}\, \Psi(\,\prod_{i=1}^{\ell}f_{i}(\theta^{\,(p)})^{r_{i}^{\,(p)}}\,) \,\diamond \,\Psi(\,G_{1}(\theta^{\,(p)})
\label{FourierF} 
\end{equation}

As a matter of fact, {due to}
\begin{eqnarray}
&&\Psi(\,\prod_{i=1}^{\ell}f_{i}(\theta )^{r_{i}^{\,(1)}}\,)\diamond \Psi(\,\prod_{i=1}^{\ell}f_{i}(\theta )^{r_{i}^{\,(2)}}\,)\label{homomorphism}\\
&&=\Psi(\,\prod_{i=1}^{\ell}f_{i}(\theta )^{r_{i}^{\,(1)}}*\prod_{i=1}^{\ell}f_{i}(\theta\, )^{r_{i}^{\,(2)}}\,)=\Psi(\,\prod_{i=1}^{\ell}f_{i}(\theta )^{\oplus_{p=1}^{2}r_{i}^{\,(p)}}\,),\notag
\end{eqnarray}
$\Psi$ is a homomorphism from  the $Z_{2}^{\otimes \ell}$ group $\{\prod_{i=1}^{\ell}f_{i}(\theta )^{r_{i}}\}$ to either the $Z_{2}$ group $\{c^{N_{\ell},0}_{\rm even},\zeta^{N_{\ell},0}_{\rm odd}\}$ or the $Z_{2}\times Z_{2}$ group $\{c^{N_{\ell},0}_{\rm even},\zeta^{N_{\ell},0}_{\rm odd},\zeta^{N_{\ell},0}_{\rm even},c^{N_{\ell},0}_{\rm odd}\}$.


\subsection{Integrating out the first $n-1$ integrals of $F^{(n)}_{\vec{r}_{\ell}}$}
\label{sec: evaluate F}

We shall prove the following lemma,
\begin{mylemma}
For all $n\leq \Lambda+1$, the form of $F^{(n)}_{\vec{r}_{\ell}}$ after $n-1$ integrations becomes
\begin{equation}
F^{(n)}_{\vec{r}_{\ell}}=\int_{0}^{\pi} d\theta^{\,(n)}\,\Psi(\,\prod_{i=1}^{\ell}f_{i}(\theta^{\,(n)})^{r_{i}}\,)\,\diamond\,\Psi(\,G_{n}(\theta^{\,(n)})\,)
\label{integrateF}
\end{equation}
up to different unimportant coefficients in the Fourier expansion, where $\Psi(\,G_{n}(\theta^{\,(n)})\,)\equiv \zeta^{N_{\ell},n}_{\rm even}$. With $\Psi(\,\prod_{i=1}^{\ell}f_{i}(\theta^{\,(n)})^{r_{i}}\,)$ given explicitly in Table \ref{table:modulation type}, all the possible function types of the resulting integrands of $F^{(n)}_{\vec{r}_{\ell}}$ $\Psi(\,\prod_{i=1}^{\ell}f_{i}(\theta^{\,(n)})^{r_{i}}\,)\,\diamond\,\Psi(\,G_{n}(\theta^{\,(n)})\,)$ in Eq.  \eqref{integrateF} are listed in Table \ref{table:Ftype}:
\begin{table}[htbp]
\begin{tabular}{|c|c|c|}
\hline
error type $\vec{r}_{\ell}$ & the resulting integrand of $F^{(n)}_{\vec{r}_{\ell}}$ & $\Lambda_{\vec{r}_{\ell}}$\\  \hline 
$(1)$ $(\vec{v}_{\ell-1}^{\,0},0)$ :&  \,$ \zeta^{N_{\ell},n}_{\rm even}$ & $\infty$ \\ 
$(2)$ $(\vec{v}_{\ell-1}^{\,0},1)$: &  \,$c^{N_{\ell},n}_{\rm odd}$  & $N_{\ell}$\\  
$(3)$ $(\vec{v}_{\ell-1}^{\,1},0)$: &  \,$ c^{N_{\ell},n}_{\rm even}$  & $\check{N}_{\vec{r}_{\ell-1}}$\\   
$(4)$ $(\vec{v}_{\ell-1}^{\,1},1)$: &  \,$\zeta^{N_{\ell},n}_{\rm odd}$ & $\infty$\\  \hline 
\end{tabular}
\caption{The second column shows the  function type of the resulting integrand of $F^{(n)}_{\vec{r}_{\ell}}$ after $n-1$ integrations. The third column shows the maximum order $\Lambda_{\vec{r}_{\ell}}$ up to which the function type of the resulting integrand of $F^{(n)}_{\vec{r}_{\ell}}$  does not contain a constant term.}
\label{table:Ftype}
\centering
\end{table}

$\Lambda_{\vec{r}_{\ell}}$ is defined as the maximum order such that the function type of the resulting integrand of $F^{(n)}_{\vec{r}_{\ell}}$  does not contain any constant term, and 
\begin{equation}
\Lambda\equiv\min[\{\Lambda_{\vec{r}_{\ell}}\}].
\label{order of no constant}
\end{equation}
\label{lem:integrateF}
\end{mylemma}

The proof of Lemma \ref{lem:integrateF} is done by induction. It is trivial that Lemma \ref{lem:integrateF} is true for the first order $F^{(1)}_{\vec{r}_{\ell}}$  based on the form of Eq.~\eqref{FourierF}. Suppose Lemma \ref{lem:integrateF} holds for all the $(n-1)^{\textrm{th}}$ order NUDD coefficients where $n-1\leq \Lambda$, i.e.,
\begin{align}
\label{F n-1}
&F^{(n-1)}_{\vec{r}_{\ell}}=\\
&\quad
\int_{0}^{\pi} d\theta^{\,(n-1)}\,\Psi(\,\prod_{i=1}^{\ell}f_{i}(\theta^{\,(n-1)})^{r_{i}}\,)\,\diamond\,\Psi(\,G_{n-1}(\theta^{\,(n-1)})\,), \notag
\end{align}
where $\Psi(\,\prod_{i=1}^{\ell}f_{i}(\theta^{\,(n-1)})^{r_{i}}\,)\,\diamond\,\Psi(\,G_{n-1}(\theta^{\,(n-1)})\,)$ belongs to one of $ \{c^{N_{\ell},n-1}_{\rm even},\zeta^{N_{\ell},n-1}_{\rm odd},\zeta^{N_{\ell},n-1}_{\rm even},c^{N_{\ell},n-1}_{\rm odd}\}$.

To proceed to the next order, first compare the forms of  Eq.~\eqref{FourierF} between the $n^{\textrm{th}}$ order and the $(n-1)^{\textrm{th}}$ order NUDD coefficients. One can see that the $n^{\textrm{th}}$ order NUDD coefficients can actually be viewed as  one integral nested with one order lower [$(n-1)^{\textrm{th}}$-order]  NUDD coefficients, i.e., 
\begin{align}
\label{F n to n-1}
&F^{(n)}_{\vec{r}_{\ell}}=\\
&\ \int_{0}^{\pi} d\theta^{(n)} \Psi(\,\prod_{i=1}^{\ell}f_{i}(\theta^{\,(n)})^{r_{i}^{\,(n)}}\,) \diamond \Psi(\,G_{1}(\theta^{\,(n)}))\diamond 
F^{(n-1),\theta^{(n)}}_{\vec{r}^{\,'}_{\ell}}\notag
\end{align}
where $\vec{r}_{\ell}=\vec{r}_{\ell}^{\,(n)}\oplus\vec{r}^{\,'}_{\ell}$ implies that $\vec{r}^{\,'}_{\ell}$ could be a different vector from  $\vec{r}_{\ell}$, and the extra superscript $\theta^{(n)}$ of $F^{(n-1)}_{\vec{r}^{\,'}_{\ell}}$ indicates that the upper integration limit, $\pi$, of the last integral of $F^{(n-1)}_{\vec{r}^{\,'}_{\ell}}$ is replaced by $\theta^{(n)}$. Therefore,  we can just substitute the results of $F^{(n-1)}_{\vec{r'}_{\ell}}$ [Eq.~\eqref{F n-1}] directly into Eq.~\eqref{F n to n-1} to evaluate $F^{(n)}_{\vec{r}_{\ell}}$.

From the definition of $\Lambda$,  the  resulting integrands of $F^{(n-1)}_{\vec{r}^{\,'}_{\ell}}$ for all possible $\vec{r}^{\,'}_{\ell}$ are guaranteed to  contain no  constant term for $n-1\leq \Lambda$. Accordingly, it is straightforward to check that  the operation 
\begin{equation}
\Psi(\,G_{1}(\theta^{\,(n)}))\,\diamond \,\int_{0}^{\theta^{(n)}} d\theta^{\,(n-1)}
\label{int operation}
\end{equation}
 maps $\{c^{N_{\ell},n-1}_{\rm even},\zeta^{N_{\ell},n-1}_{\rm odd},\zeta^{N_{\ell},n-1}_{\rm even},c^{N_{\ell},n-1}_{\rm odd}\}$ to its corresponding $\{c^{N_{\ell},n}_{\rm even},\zeta^{N_{\ell},n}_{\rm odd},\zeta^{N_{\ell},n}_{\rm even},c^{N_{\ell},n}_{\rm odd}\}$ regardless of the unimportant change of the coefficients in the linear combination. Therefore, up to different coefficients in the linear combination, the operation Eq.~\eqref{int operation}  gives rise to the following mapping,
\begin{eqnarray}
&&\Psi(\,\prod_{i=1}^{\ell}f_{i}(\theta^{\,(n-1)})^{r'_{i}}\,)\,\diamond\,\Psi(\,G_{n-1}(\theta^{\,(n-1)})\,) \xrightarrow \notag\\
&&\Psi(\,\prod_{i=1}^{\ell}f_{i}(\theta^{\,(n)})^{r'_{i}}\,)\,\diamond\,\Psi(\,G_{n}(\theta^{\,(n)})\,). 
\label{operation: int}
\end{eqnarray}

By substituting Eq.~\eqref{operation: int} into Eq \eqref{F n to n-1}, the resulting integrand of $F^{(n)}_{\vec{r}_{\ell}}$ becomes
\begin{equation}
\Psi(\,\prod_{i=1}^{\ell}f_{i}(\theta^{\,(n)})^{r_{i}^{\,(n)}}\,) \diamond\Psi(\,\prod_{i=1}^{\ell}f_{i}(\theta^{\,(n)})^{r'_{i}}\,)\,\diamond\,\Psi(\,G_{n}(\theta^{\,(n)})\,).
\end{equation}
Applying the homomorphism property Eq.~\eqref{homomorphism} to the above equation, we obtain Eq.~\eqref{integrateF} for the  $n^{\rm th}$ order where $n \leq \Lambda +1$.

For the order $n=\Lambda +1$, the resulting integrands of  $F^{(n)}_{ \vec{r}_{\ell}}$ for some of the errors  begin to contain a constant term. Then  the operation $\Psi(\,G_{1}(\theta^{\,(n+1)})\,\diamond \,\int_{0}^{\theta^{(n+1)}} d\theta^{\,(n)}$ will map them to different functions other than a purely cosine series or a purely sine series. Hence, it follows that Eq.~\eqref{integrateF} or the {second} column of Table \ref{table:Ftype} are no longer true  for  the $\Lambda +2$ or higher order   NUDD coefficients.

Now {let us} prove the claimed values of $\Lambda_{\vec{r}_{\ell}}$  shown in the last column of Table \ref{table:Ftype}. By Definition \ref{def:type}, $c^{N_{\ell},n}_{\rm odd}$, which is the function type of the resulting integrand of the $(\vec{v}_{\ell-1}^{\,0},1)$-type error NUDD coefficients in Table \ref{table:Ftype},  does not contain a cosine function with zero argument (a constant 1 term)  when $n \leq N_{\ell}$. Therefore, $\Lambda_{(\vec{v}_{\ell-1}^{\,0},1)}=N_{\ell}$. Neither $\zeta^{N_{\ell},n}_{\rm even}$ nor $\zeta^{N_{\ell},n}_{\rm odd}$ ever have constant terms, which indicates that  $\Lambda_{(\vec{v}_{\ell-1}^{\,0},0)}=\Lambda_{(\vec{v}_{\ell-1}^{\,1},1)}=\infty$.   Only for the $(\vec{v}_{\ell-1}^{\,1},0)$-type error,  $c^{N_{\ell},n}_{\rm even}$ is in general (by definition) allowed to have a cosine function with zero argument, a constant term. However, in Sec.~\ref{sec:inner performance}, we already showed that $F^{(n)}_{\vec{r}_{\ell}}=0$ for all $n \leq \check{N}_{\vec{r}_{\ell-1}}$ [Eq.~\eqref{Finner2}].  Hence, with  Eq.~\eqref{Finner2}, 
\begin{equation}
F^{(n)}_{ \vec{r}_{\ell}\in (\vec{v}_{\ell-1}^{\,1},0)}=0=\int_{0}^{\pi} c^{N_{\ell},n}_{\rm even}\,d\theta \quad \forall n \leq \check{N}_{\vec{r}_{\ell-1}\in \vec{v}_{\ell-1}^{\,1}}
\label{F10}
\end{equation}
suggests that $c^{N_{\ell},n}_{\rm even}$  in Eq.~\eqref{F10} has no constant term for the  first $\check{N}_{\vec{r}_{\ell-1}\in \vec{v}_{\ell-1}^{\,1}}$ orders. Accordingly, 
\begin{equation}
\Lambda_{\vec{r}_{\ell}\in (\vec{v}_{\ell-1}^{\,1},0)}=\check{N}_{\vec{r}_{\ell-1}\in \vec{v}_{\ell-1}^{\,1} }.
\label{Nv1}
\end{equation}


\subsection{The vanishing of the $(\vec{v}_{\ell-1}^{\,0},1)$-type error}
\label{sec:outermost error}

According to Lemma \ref{lem:integrateF}, for the $(\vec{v}_{\ell-1}^{\,0},1)$-type error,  $F^{(n)}_{(\vec{v}_{\ell-1}^{\,0},1)}=\int_{0}^{\pi}\, c^{N_{\ell},n}_{\rm odd}\, d\theta$, where $c^{N_{\ell},n}_{\rm odd}$ does not contain any constant term, for all $n\leq\min[\Lambda+1,N_{\ell}]$, where  $n\leq\Lambda+1$ ensures that the function type of the resulting integrand of $F^{(n)}_{(\vec{v}_{\ell-1}^{\,0},1)}$ is $c^{N_{\ell},n}_{\rm odd}$, and $n\leq \Lambda_{(\vec{v}_{\ell-1}^{\,0},1)}=N_{\ell}$  ensures that there is no constant term in it.  Then as shown in Eq.~\eqref{outer error}, it follows that $F^{(n)}_{(\vec{v}_{\ell-1}^{\,0},1)}=0$  for
\begin{equation}
n\leq\min[\Lambda+1,N_{\ell}].
\label{v01even}
\end{equation}
where the value of $\Lambda$ expressed explicitly in terms of the sequence orders shall be determined as follows.

For the  NUDD with $[N_{i}]_{2}=0$ for all $i<\ell$, $\oplus_{k=1}^{\ell-1}r_{k}[N_{k}]_{2}=0$ is always true for all errors, so there are only two error types, $(\vec{v}_{\ell-1}^{\,0},0)$- and  $(\vec{v}_{\ell-1}^{\,0},1)$-type errors in this case.  Therefore, from the definition of $\Lambda$ [Eq.~\eqref{order of no constant}], we have
\begin{equation}
\Lambda=\min[ \Lambda_{(\vec{v}_{\ell-1}^{\,0},0)}, \Lambda_{(\vec{v}_{\ell-1}^{\,0},1)}]=\min[\infty, N_{\ell}]=N_{\ell}.
\label{validorder: even}
\end{equation}
which  leads to $\min[N_{\ell}+1,N_{\ell}]=N_{\ell}$ in Eq.~\eqref{v01even}. Therefore, we have {proven} that  
\begin{equation}
F^{(n)}_{(\vec{v}_{\ell-1}^{\,0},1)}=0 \qquad \forall\,n\leq N_{\ell}
\label{Feven}
\end{equation} 
for the  NUDD with all even inner sequence orders.

For NUDD with at least one odd sequence order in the first $\ell-1$ UDD layers,  from the definition of $\Lambda$ Eq.~\eqref{order of no constant},
\begin{eqnarray}
&&\Lambda=\min[\Lambda_{(\vec{v}_{\ell-1}^{\,0},0)},\Lambda_{(\vec{v}_{\ell-1}^{\,0},1)}, \Lambda_{\vec{r}_{\ell}\in (\vec{v}_{\ell-1}^{\,1},0)}, \Lambda_{(\vec{v}_{\ell-1}^{\,1},1)}]\notag \\
&&=\min[ \infty, N_{\ell},\{\check{N}_{\vec{r}_{\ell-1}\in \vec{v}_{\ell-1}^{\,1} }\},\infty]\notag \\
&&=\min[N_{\ell}, \{\check{N}_{\vec{r}_{\ell-1}\in \vec{v}_{\ell-1}^{\,1} }\} ] .
\label{validorder}
\end{eqnarray}
Due to 
\begin{mylemma}
$\min[\{\check{N}_{\vec{r}_{\ell-1}\in \vec{v}_{\ell-1}^{\,1} }\}]=N^{k'<\ell}_{o_{\min}} $,
\label{lem: min Nv1}
\end{mylemma}
\noindent which we prove in Appendix \ref{app:min Nv1},  Eq.~\eqref{validorder} reads
\begin{equation}
\Lambda=\min[N_{\ell}, N^{k'<\ell}_{o_{\min}}] ,
\end{equation}
which gives rise to  
\begin{eqnarray}
\min[\Lambda+1,N_{\ell}]
&&=\min[\min[N_{\ell}, N^{k'<\ell}_{o_{\min}}]+1,N_{\ell}]\notag \\
&&=\min[ N^{k'<\ell}_{o_{\min}} +1,N_{\ell}]
\end{eqnarray}
in Eq.~\eqref{v01even}.
Therefore, we have {proven} that  
\begin{equation}
F^{(n)}_{(\vec{v}_{\ell-1}^{\,0},1)}=0 \qquad \forall\,n\leq \min[ N^{k'<\ell}_{o_{\min}} +1,N_{\ell}]
\label{Fodd}
\end{equation} 
for NUDD with  at least one odd inner sequence order.

Combining  Eqs.~\eqref{Feven} and \eqref{Fodd} with the fact that the vanishing of the other error types cannot be deduced from this approach due to their function types (shown in Table \ref{table:Ftype}), we obtain 
$F^{(n)}_{\vec{r}_{\ell}}=0$ for all $n \leq r_{\ell}[( p_\oplus(1,\ell-1)\oplus 1]\widetilde{N}_{\ell}$ [Eq.~\eqref{Fouter1}],  where $\widetilde{N}_{\ell}=N_{\ell}$ for NUDD with all even inner sequence orders, and $\widetilde{N}_{\ell}=\min[ N^{k'<\ell}_{o_{\min}}+1,N_{\ell}]$ for  NUDD with at least one odd inner sequence order, which is equivalent to the definition given in Eq.~\eqref{Ni}.  This completes the proof of the NUDD Theorem.

\section{Numerical Results: 4-layer NUDD}
\label{sec:numerical result}

In this section, numerical simulations are employed to examine the performance of a 4-layer NUDD scheme for two contrasting MOOS. Given
the UDD sequence order set $\{N_j\}^{4}_{j=1}$, the proof given above ensures a minimum suppression of all non-trivial error types to at least $N=\min[N_{1},N_{2},N_{3},N_{4}]$, resulting in $U(T)\sim\mathcal{O}(T^{N+1})$. We reconcile our analytical results with numerical simulations to confirm the predicted scaling of $U(T)$ and show that this scaling is indeed MOOS-independent. Furthermore, we analyze the scaling of all $2^{4}$ individual error types to convey the validity of the decoupling order formula Eq.~\eqref{Nrl}.


\subsection{Model}

We consider a $2$-qubit system  coupled to a generic quantum bath to analyze the performance of the $4$-layer NUDD. The total Hamiltonian, which includes a pure-bath term and system-bath interactions, is given by
\begin{equation}
H=\sum_{\lambda_{1}=0,x,y,z}\sum_{\lambda_{2}=0,x,y,z}J_{\lambda_{1}\lambda_{2}}\sigma_{\lambda_{1}}\otimes \sigma_{\lambda_{2}}\otimes B_{\lambda_{1}\lambda_{2}}
\label{eq:2H}
\end{equation}
where $\sigma_{\lambda_{j}}$ are the standard Pauli matrices with $\sigma_{0}\equiv I$ for the $j=1,2$ system qubits,  $B_{\lambda_{1}\lambda_{2}}$ are arbitrary bath-operators with $\| B_{\lambda_{1}\lambda_{2}}\|=1$, and $J_{\lambda_{1}\lambda_{2}}$ are bounded coupling coefficients between the qubits and the bath. 

Modeling the environment as a spin bath consisting of four spin-1/2 particles with randomized couplings between them, the operator $B_{\lambda_{1}\lambda_{2}}$ is given by
\begin{equation}
B_{\lambda_{1}\lambda_{2}}=\sum_{\lambda_{3},\lambda_{4},\lambda_{5},\lambda_{6}}c_{\lambda_{1}\lambda_{2}}^{\lambda_{3}\lambda_{4}\lambda_{5}\lambda_{6}}\sigma_{\lambda_{3}}\otimes \sigma_{\lambda_{4}}\otimes \sigma_{\lambda_{5}}\otimes \sigma_{\lambda_{6}},
\label{eq:bathOps}
\end{equation}
where the interactions are non-restrictive, i.e. $1$- to $4$-local interactions are permitted. The index $\lambda_{j}=0,x,y,z$ where $j=3,4,5,6$ stands for the bath qubit and coefficients $c_{\lambda_{1}\lambda_{2}}^{\lambda_{3}\lambda_{4}\lambda_{5}\lambda_{6}}\in[0,1]$ are chosen randomly from a uniform distribution.

We focus on the case of uniform coupling, $J_{\lambda_{1}\lambda_{2}}=J$ $\forall \lambda_j$ except $\lambda_1=\lambda_2=0$, where we discriminate between the strength of the system-bath interactions and the pure bath dynamics $J_{00}$. For all simulations, $J=1$MHz and $J_{00}=1$kHz are considered so that $J_{00}\ll J$. In this particular regime, DD has been shown to be the most beneficial since the environment dynamics are effectively static with respect to system-environment interactions \cite{WestFongLidar:10,QuirozLidar:11}.


\subsection{Overall decoupling order}

The overall performance of NUDD is quantified with respect to the state-independent distance measure
\begin{equation}
D=\frac{1}{\sqrt{d_S d_B}}\,\,\min_{\Phi}\|U(T)-I\otimes\Phi\|_{F},
\label{eq:Dist}
\end{equation}
where $d_S$ is the system Hilbert space dimension, $d_B$ is the environment dimension,  and $\Phi$ is a bath operator \cite{Grace:10}. The minimum order scaling of NUDD is expected to be $D\sim\mathcal{O}[(JT)^{N+1}]$ for a total sequence duration $T$. Therefore, the overall numerical decoupling order of the NUDD scheme, $n_{\min}$, is obtained by
\begin{equation}
n_{\min}=(\log_{10} D )-1
\end{equation}
In the subsequent simulations, the scaling of $D$ is extracted by varying the minimum pulse delay $\tau \equiv Ts_{j_{\ell}=1}s_{j_{\ell-1}=1}\dots s_{j_{2}=1}s_{j_{1}=1}$ instead of the total time $T$. We utilize $\tau$ since this quantity is usually the most relevant experimental time constraint.


\subsection{$\vec{r}_{4}=(r_{1},r_{2},r_{3},r_{4})$-type error}
\label{sec:errorType}

Given a MOOS $\{\Omega_{j}\}^{4}_{j=1}$, the general 2-qubit Hamiltonian {\eqref{eq:2H}} can be partitioned into $2^{4}$ error types $H_{\vec{r}_{4}}=H_{(r_{1},r_{2},r_{3},r_{4})}$ [Eq.~\eqref{Hr}] via the procedure discussed in Appendix \ref{app: step H}. We also note that it is possible to generate $\{H_{(r_{1},r_{2},r_{3},r_{4})}\}$ using $H_{(1,0,0,0)}$, $H_{(0,1,0,0)}$, $H_{(0,0,1,0)}$, and $H_{(0,0,0,1)}$ as the generators of the $Z_{2}^{\otimes 4}$ group, where the complete set of error types can be obtained by all possible products of the error generators: $H_{\vec{r}^{\,(1)}_{4}}H_{\vec{r}^{\,(2)}_{4}}=H_{\vec{r}^{\,(1)}_{4}\oplus\vec{r}^{\,(2)}_{4}}$. We explicitly display the $16$ error types in Table \ref{table:Htype}, where each error type is generated by a product of the corresponding outer-most column and row elements.

\begin{table}[h]
\begin{tabular}{c|c|c|c|c}
	*						& $H_{(0,0,0,0)}$ & $H_{(0,0,1,0)}$ &  $H_{(0,0,0,1)}$ & $H_{(0,0,1,1)}$\\
							\hline
$H_{(0,0,0,0)}$	&         $H_{(0,0,0,0)}$		&			$H_{(0,0,1,0)}$			&			$H_{(0,0,0,1)}$			&			$H_{(0,0,1,1)}$		 \\
$H_{(1,0,0,0)}$ &		  $H_{(1,0,0,0)}$	    &			$H_{(1,0,1,0)}$			&			$H_{(1,0,0,1)}$			&			$H_{(1,0,1,1)}$				 \\
$H_{(0,1,0,0)}$	&		  $H_{(0,1,0,0)}$		&			$H_{(0,1,1,0)}$			&			$H_{(0,1,0,1)}$			&			$H_{(0,1,1,1)}$			   \\
$H_{(1,1,0,0)}$	&		  $H_{(1,1,0,0)}$		&			$H_{(1,1,1,0)}$			&			$H_{(1,1,0,1)}$			&			$H_{(1,1,1,1)}$				 
\end{tabular}
\caption{16 error types $\vec{r}_{4}$ for the $4$-layer NUDD scheme.}
\label{table:Htype}
\end{table}

In order to illustrate the procedure by which the error types are generated, we consider the two MOOS sets utilized to convey the unbiased nature of the minimum scaling for NUDD. The first,
\begin{equation}
\{I \otimes \sigma_{z},I\otimes \sigma_{x},\sigma_{z}\otimes I,\sigma_{x}\otimes I\}
\label{eq:MOOS1}
\end{equation} 
is composed of single-qubit operators with the corresponding error generators 
\begin{eqnarray}
H_{(1,0,0,0)}=I\otimes \sigma_{x}, \quad H_{(0,1,0,0)}=I \otimes \sigma_{z},\notag\\
H_{(0,0,1,0)}=\sigma_{x}\otimes I, \quad H_{(0,0,0,1)}=\sigma_{z}\otimes I ,
\label{generator1}
\end{eqnarray} 
up to arbitrary bath operators which are omitted in the above equations. Substituting Eq.~\eqref{generator1} into Table \ref{table:Htype}, we obtain the explicit forms for all error types, as shown as shown in Table \ref{table:Htype1}.
\begin{table}[h]
\begin{tabular}{c|c|c|c|c}
	*						& $I\otimes I$ & $\sigma_{x}\otimes I$ &  $\sigma_{z}\otimes I$ & $\sigma_{y}\otimes I$\\
							\hline
$I\otimes I$	&              $I\otimes I$	       	    &		$\sigma_{x}\otimes I$		      &			$\sigma_{z}\otimes I$		&	$\sigma_{y}\otimes I$	 \\
$I\otimes \sigma_{x}$ &		   $I\otimes \sigma_{x}$    &		$\sigma_{x}\otimes \sigma_{x}$	 &	$\sigma_{z}\otimes\sigma_{x}$		&	$\sigma_{y}\otimes \sigma_{x}$			 \\
$I \otimes \sigma_{z}$	&	   $I \otimes \sigma_{z}$	&		$\sigma_{x}\otimes \sigma_{z}$	 &	$\sigma_{z}\otimes\sigma_{z}$		&	$\sigma_{y}\otimes \sigma_{z}$			   \\
$I \otimes \sigma_{y}$ &	   $I \otimes \sigma_{y}$ 	&		$\sigma_{x}\otimes \sigma_{y}$	&	$\sigma_{z}\otimes \sigma_{y}$		&	$\sigma_{y}\otimes \sigma_{y}$				 
\end{tabular}
\caption{The $16$ error types $\vec{r}_{4}$ for the $4$-layer NUDD scheme with the MOOS $\{I \otimes \sigma_{z},I\otimes \sigma_{x},\sigma_{z}\otimes I,\sigma_{x}\otimes I\}$.}
\label{table:Htype1}
\end{table}

The second MOOS is chosen similarly as 
\begin{equation}
\{\sigma_{z} \otimes \sigma_{z}, I \otimes \sigma_{x}, \sigma_{z}\otimes I, \sigma_{x}\otimes I\},
\label{eq:MOOS2}
\end{equation}
{where} we replace the $\sigma_{z}\otimes I$ operator with a two-qubit $\sigma_z\otimes\sigma_z$ operator. The error generators for this MOOS, up to arbitrary bath operators, are
\begin{eqnarray}
H_{(1,0,0,0)}=I\otimes \sigma_{x}, \quad H_{(0,1,0,0)}=I \otimes \sigma_{z},\notag\\
H_{(0,0,1,0)}=\sigma_{x}\otimes \sigma_{x}, \quad H_{(0,0,0,1)}=\sigma_{z}\otimes I.
\label{generator2}
\end{eqnarray}
and the error types are given in Table \ref{table:Htype2}.
\begin{table}[h]
\begin{tabular}{c|c|c|c|c}
	*						& $I\otimes I$               &      $\sigma_{x}\otimes \sigma_{x}$      &  $\sigma_{z}\otimes I$             & $\sigma_{y}\otimes \sigma_{x}$\\
							\hline  
$I\otimes I$	&              $I\otimes I$	       	    &		$\sigma_{x}\otimes \sigma_{x}$	   &	$\sigma_{z}\otimes I$		&	$\sigma_{y}\otimes \sigma_{x}$	 \\
$I\otimes \sigma_{x}$ &		   $I\otimes \sigma_{x}$    &		$\sigma_{x}\otimes I$	       &	$\sigma_{z}\otimes\sigma_{x}$		&	$\sigma_{y}\otimes I$			 \\
$I \otimes \sigma_{z}$	&	   $I \otimes \sigma_{z}$	&		$\sigma_{x}\otimes \sigma_{y}$	 &	$\sigma_{z}\otimes\sigma_{z}$		&	$\sigma_{y}\otimes \sigma_{y}$			   \\
$I \otimes \sigma_{y}$ &	   $I \otimes \sigma_{y}$ 	&		$\sigma_{x}\otimes \sigma_{z}$	&	$\sigma_{z}\otimes \sigma_{y}$		&	$\sigma_{y}\otimes \sigma_{z}$				 
\end{tabular}.
\caption{The $16$ error types $\vec{r}_{4}$ for the $4$-layer NUDD scheme with the MOOS $\{\sigma_{z} \otimes \sigma_{z}, I \otimes \sigma_{x}, \sigma_{z}\otimes I, \sigma_{x}\otimes I\}$}
\label{table:Htype2}
\end{table}


\subsection{Decoupling order of each error type}

The decoupling order for each error type is extracted using an equivalent procedure to Ref.~\cite{QuirozLidar:11}, where
\begin{equation}
E_{\vec{r}_{4}}=\|\textrm{Tr}_{S}\left(U(T)H_{\vec{r}_{4}}\right)\|_F,  
\label{E_mu}
\end{equation}
represents the effective error for error type $H_{\vec{r}_{4}}$ present at the end of the DD evolution $U(T)$. Here, $\|A\|_F$ is the Frobenius norm of $A$, 
\begin{equation}
\|A\|_F = \mathrm{Tr}\sqrt{A^\dagger A},
\end{equation}
the sum of singular values of $A$. (The choice of norm is somewhat arbitrary; we could have used any other unitarily invariant norm as well, e.g., the trace norm.). 

The expected scaling for each effective error type can be shown to follow
\begin{equation}
E_{\vec{r}_{4}}\sim \mathcal{O}[(JT)^{\check{N}_{\vec{r}_{4}}+1}]\sim \mathcal{O}[(J\tau)^{\check{N}_{\vec{r}_{4}}+1}] ,
\end{equation} 
where $\check{N}_{\vec{r}_{4}}$ is given by the Eq.~\eqref{Nrl}. 
Therefore, the numerical decoupling order of $\vec{r}_{4}$-type error,  $n_{\vec{r}_{4}}$, is obtained by
\begin{equation}
n_{\vec{r}_{4}}=(\log_{10} E_{\vec{r}_{4}})-1 
\end{equation}


\subsection{Comparison of analytical predictions with numerical results}

\begin{figure*}[t]
\centering
\subfigure[]{\includegraphics[width=\columnwidth]{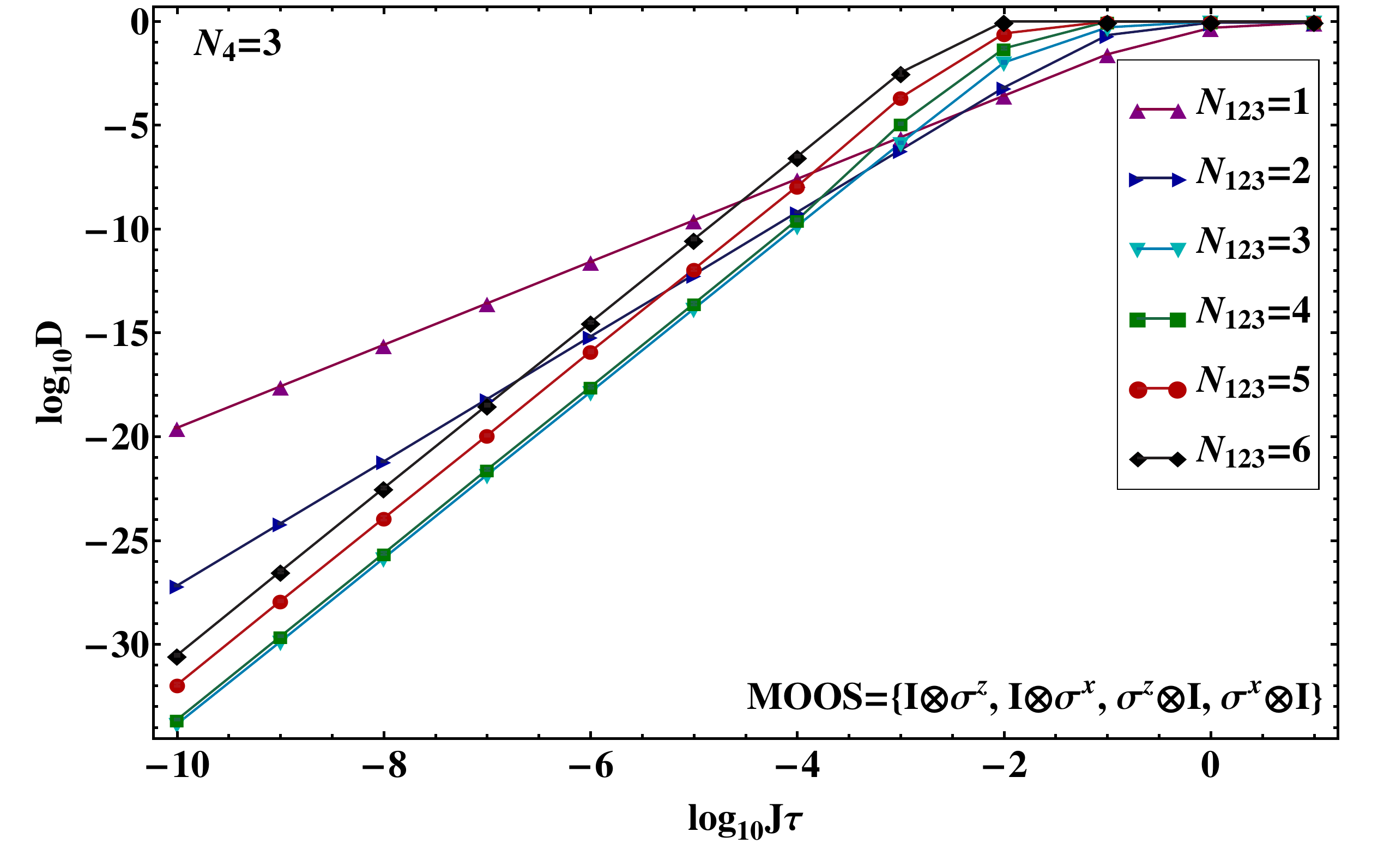}}
\subfigure[]{\includegraphics[width=\columnwidth]{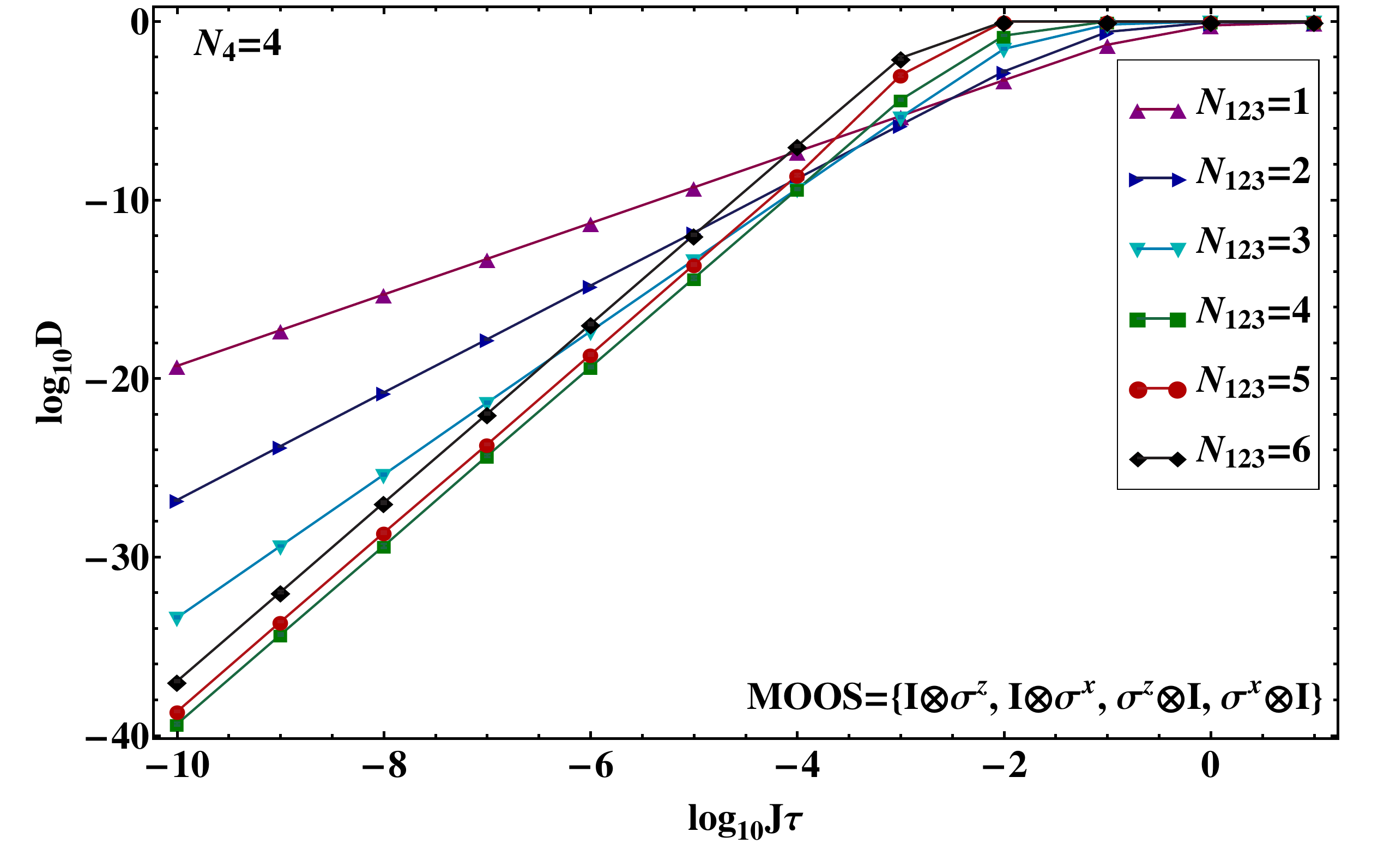}}\\
\subfigure[]{\includegraphics[width=\columnwidth]{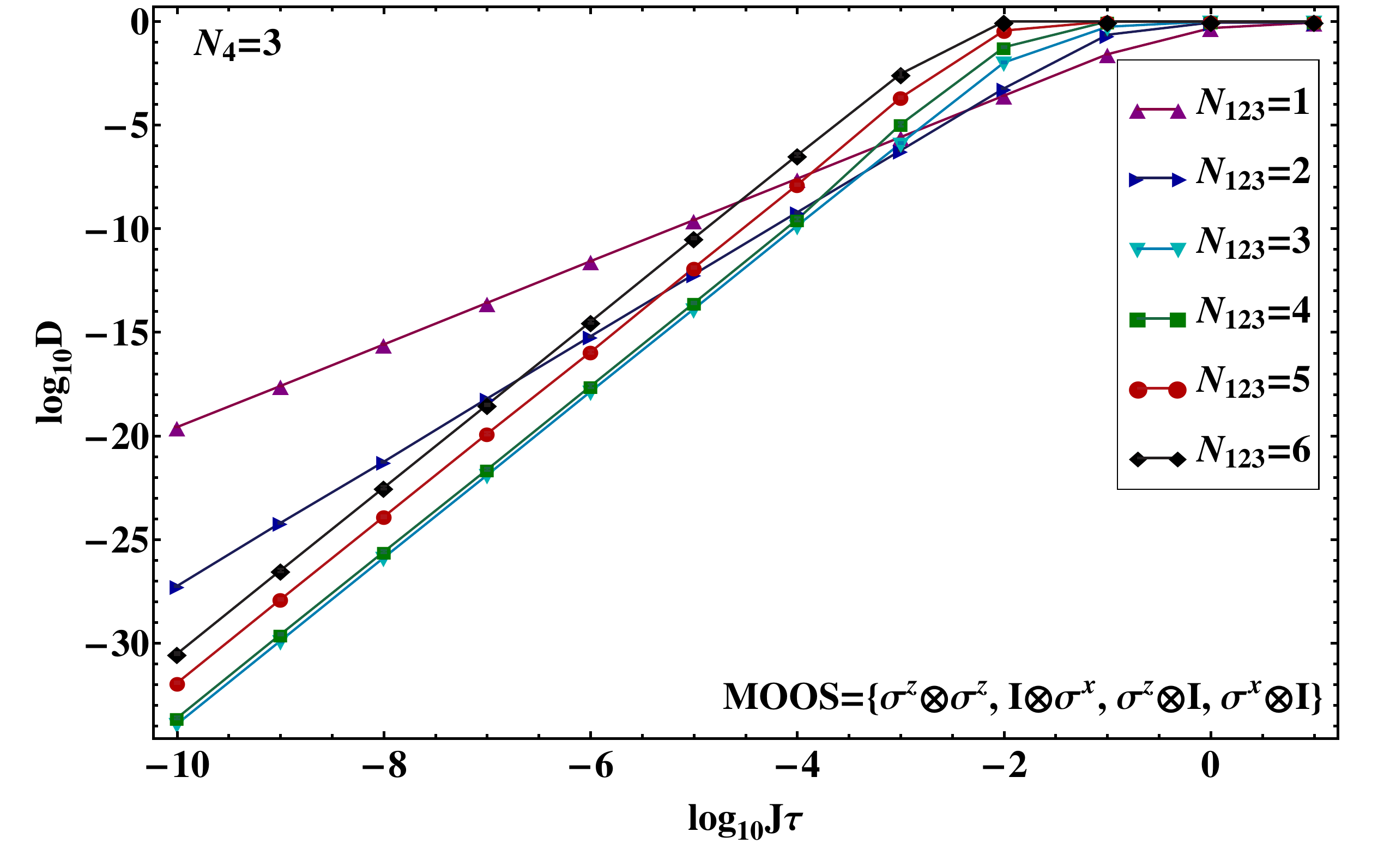}}
\subfigure[]{\includegraphics[width=\columnwidth]{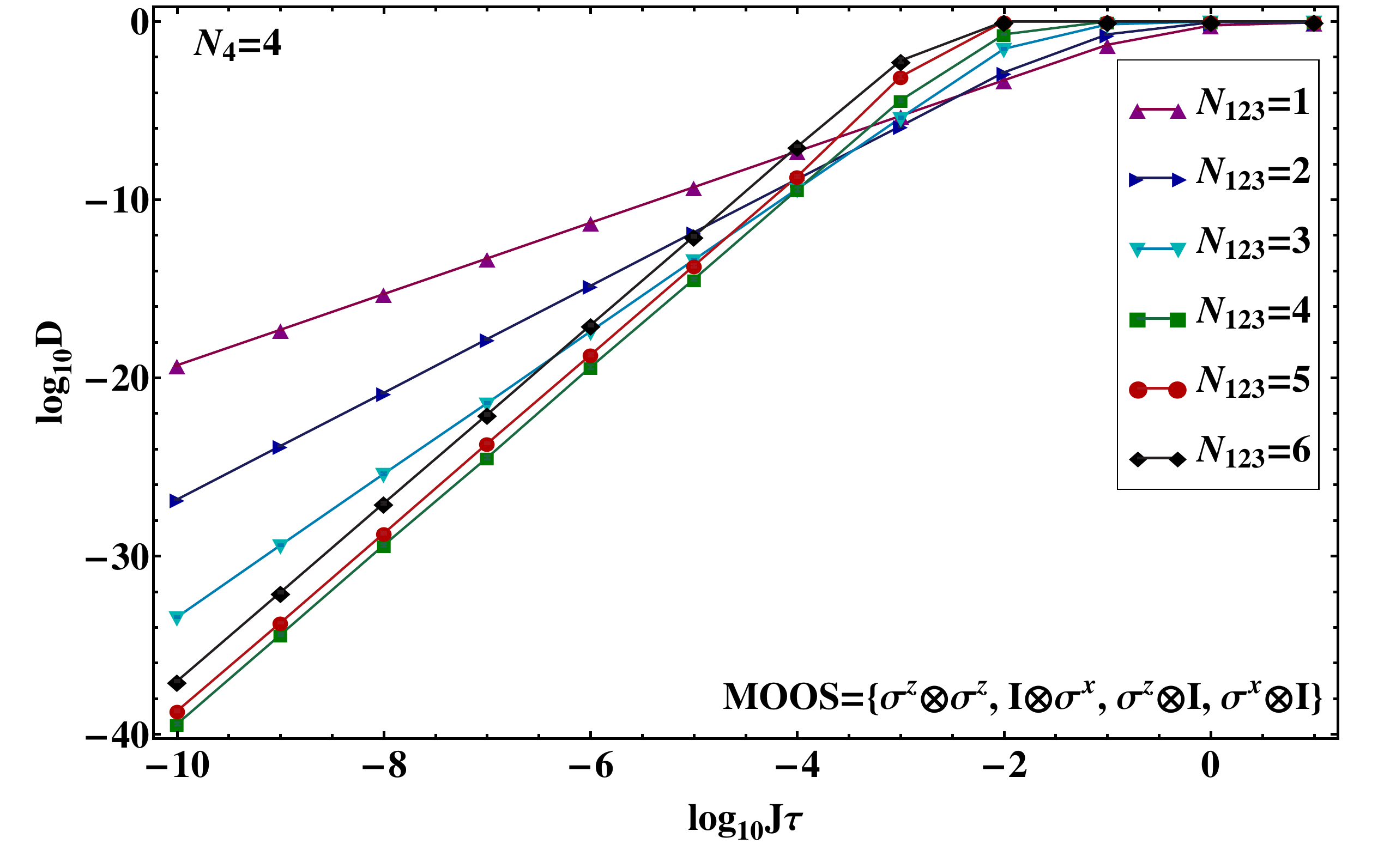}}
\caption{Performance of NUDD as a function of the minimum pulse interval $\tau$ for the two-qubit system specified by Eqs.~\eqref{eq:2H} and \eqref{eq:bathOps}, averaged over $15$ random realizations of $B_{\lambda_1\lambda_2}$. Error bars are included, but are very small. A four-layer construction is utilized to address all system-bath interactions, where the MOOS sets are chosen as Eq.~\eqref{eq:MOOS1} in (a) and (b), and Eq.~\eqref{eq:MOOS2} in (c) and (d). {The orders of the first three layers are identical: $N_{123}=N_j$, $j=2,3,4$ and are varied as $N_{123}=1,2,\ldots,6$.}  In particular, the outer layer sequence {order is} chosen as {$N_{4}=3$} in (a) and (c), and {$N_{4}=4$} in (b) and (d). Linear interpolation from $\log_{10} J\tau=-10$ to approximately $\log_{10} J\tau=-2$ confirms the MOOS-independent scaling $D\sim\mathcal{O}[(J\tau)^{n_{\min}+1}]$, $n_{\min}=\min \{N_j\}^{4}_{j=1}$, for all simulations.}
\label{fig:performance}
\end{figure*}

In order to convey the MOOS-independent scaling of NUDD performance, we consider the case where the three {inner-most} layers contain the same number of pulses, i.e. {$N_{123}=N_j$, $j=1,2,3$,} for both MOOS sets given above; see Eqs.~(\ref{eq:MOOS1}) and (\ref{eq:MOOS2}). We choose the NUDD layering such that the right-most control pulse operator corresponds to the {outer-most} UDD sequence with sequence order {$N_4$} and the left-most operator generates the {inner-most} UDD sequence with sequence order {$N_1$}. Fixing { the sequence order for the outer-most layer, $N_{4}$, and varying $N_{123}$}, a qualitative and quantitative similarity can be expected between the QDD analysis given in Ref.\cite{QuirozLidar:11} and the 4-layer NUDD construction considered here. Hence, the expected scaling of the distance measure described by Eq.~(\ref{eq:Dist}) is $D\sim\mathcal{O}[(J\tau)^{n_{\min}+1}]$, where {$n_{\min}=\min[N_4,N_{123}]$.}


\begin{table}[t]
\begin{tabular}{|c|c|c|c|}
\hline
$\check{N}_{(0,0,0,0)}$            &  $\check{N}_{(0,0,1,0)}$     &     $\check{N}_{(0,0,0,1)}$       &      $\check{N}_{(0,0,1,1)}$      \\
$0$            &  $6$     &     $8$       &      $8$      \\\hline
$\check{N}_{(1,0,0,0)}$            &  $\check{N}_{(1,0,1,0)}$     &     $\check{N}_{(1,0,0,1)}$       &      $\check{N}_{(1,0,1,1)}$       	 \\
$2$            &  $6$     &     $8$       &      $8$       	 \\\hline
$\check{N}_{(0,1,0,0)}$            &  $\check{N}_{(0,1,1,0)} $     &     $\check{N}_{(0,1,0,1)}$       &      $\check{N}_{(0,1,1,1)}$       	 \\
$4$            &  $6$     &     $8$       &      $8$       	 \\\hline
$\check{N}_{(1,1,0,0)}$            &  $\check{N}_{(1,1,1,0)}$     &     $\check{N}_{(1,1,0,1)}$       &      $\check{N}_{(1,1,1,1)}$       	 \\
$4$            &  $6$     &     $8$       &      $8$       	 \\\hline
\end{tabular}
\caption{For the case with $N_1=2$, $N_2=4$, $N_3=6$, $N_4=8$, analytical and numerical decoupling orders for all error types are in complete agreement. The analytical overall decoupling order $\check{N}_{\min}=\min[2,4,6,8]=2$ also agrees with the actual overall decoupling order. }
\label{ex:even}
\end{table}
In Fig. \ref{fig:performance}, the performance is shown as a function of the minimum pulse interval $\tau$ for both MOOS sets in the case of {$N_{4}=3$} in (a) and (c), and {$N_{4}=4$} in (b) and (d). The inner sequence orders are $N_1=1,2,\ldots,6$ for all simulations, which are averaged over $15$ random realizations of $B_{\lambda_1\lambda_2}$ using $120$ digits of precision.

Comparing Figs. \ref{fig:performance}(a) and (c), we find that variations in the MOOS do not change the qualitative features of NUDD performance. Furthermore, we can confirm by linear interpolation that the slope of each curve indeed satisfies {$n_{\min}=\min[N_{123},N_4]$}. Similar results are obtained for {$N_{4}=4$} in Figs.~\ref{fig:performance}(b) and (d) as well. By this comparison, we show that NUDD is a MOOS-independent construction and confirm this result for two specific even and odd parity values of {$N_{4}$}. We expect the scaling to remain consistent for general $4$-layer NUDD, where the scaling becomes $n_{\min}=\min[N_1,N_2,N_3,N_4]$, and for more general multi-layer NUDD as well.

The scaling for each error type is characterized as a function of the minimum pulse interval using Eq.~\eqref{E_mu} specifically for the MOOS given by Eq.~\eqref{eq:MOOS1}. For each error type, we display the decoupling order in tabular form, following the format described in Sec.~\ref{sec:errorType}. The numerically obtained decoupling order is denoted in black, analytical prediction from Eq.~\eqref{Nrl} in blue and the "naive" decoupling order in red.  If a blue (red) number is absent, it means our analytical decoupling order (naive decoupling order ) is exactly the same as the actual decoupling order.   We refer to
\begin{equation}
\check{N}_{\vec{r}_{\ell}}=\max[\{\,r_{i}N_{i}\}_{i=1}^{\ell}].
\label{naiveformula:even}
\end{equation}
as the ``naive" decoupling order since it assumes that each nested layer acts independently.  A previous analysis for QDD has shown that inter-layer interference may alter the expected UDD efficiency even when the layers do not provide decoupling for the same error type \cite{QuirozLidar:11}. We expect a similar characteristic for NUDD as well and seek to identify the conditions when the naive decoupling order is incorrect.


Tables \ref{ex:even}, \ref{ex: inner even outer odd}, and \ref{ex:odd>} contain typical examples where our analytically predicted decoupling {orders} for {all} error {types} exactly match those obtained numerically. Note that in these particular NUDD schemes, the suppression {order} of each UDD layer is equal to its sequence order, i.e., $\widetilde{N}_{i}=N_{i}$. In Table \ref{ex:even}, all sequence orders are even parity, i.e. $[N_{i}]_{2}=0$ $\forall$ $i$. Specifically, we consider $N_1=2$, $N_2=4$, $N_3=6$, and $N_4=8$, although we expect equivalent results for arbitrary even parity $N_j$. According to Eq.~\eqref{Nrl}, the decoupling order formula for such cases coincides with the naive decoupling order formula Eq.~\eqref{naiveformula:even}; interference between layers is not observed. The analytical overall decoupling order $N=\min[2,4,6,8]=2$ is found to agree with the actual overall decoupling order as well.

\begin{table}[t]
\begin{tabular}{|c|c|c|c|}
\hline
$\check{N}_{(0,0,0,0)}$            &  $\check{N}_{(0,0,1,0)}$     &     $\check{N}_{(0,0,0,1)}$       &      $\check{N}_{(0,0,1,1)}$      \\
$0$            &  $6$     &     $3$       &      ${\color{red}6},7$        	 \\\hline
$\check{N}_{(1,0,0,0)}$            &  $\check{N}_{(1,0,1,0)}$     &     $\check{N}_{(1,0,0,1)}$       &      $\check{N}_{(1,0,1,1)}$       	 \\
$2$            &  $6$     &     $3$       &      ${\color{red}6},7$       	 \\\hline
$\check{N}_{(0,1,0,0)}$            &  $\check{N}_{(0,1,1,0)} $     &     $\check{N}_{(0,1,0,1)}$       &      $\check{N}_{(0,1,1,1)}$       	 \\
$4$            &  $6$     &     ${\color{red}4},5$       &      ${\color{red}6},7$       	 \\\hline
$\check{N}_{(1,1,0,0)}$            &  $\check{N}_{(1,1,1,0)}$     &     $\check{N}_{(1,1,0,1)}$       &      $\check{N}_{(1,1,1,1)}$       	 \\
$4$            &  $6$     &     ${\color{red}4},5$       &      ${\color{red}6},7$       	 \\\hline
\end{tabular}
\caption{For the case with $N_1=2$, $N_2=4$, $N_3=6$, $N_4=3$, analytical and numerical decoupling order (marked as black) are in complete agreement for all error types. The analytical overall decoupling order $N=\min[2,4,6,3]=2$ also agrees with the actual overall decoupling order. The difference between the naive (marked as red) and the actual decoupling orders comes from outer-odd-UDD suppression effect.}
\label{ex: inner even outer odd}
\end{table}

\begin{table}[t]
\begin{tabular}{|c|c|c|c|}
\hline
$\check{N}_{(0,0,0,0)}$            &  $\check{N}_{(0,0,1,0)}$     &     $\check{N}_{(0,0,0,1)}$       &      $\check{N}_{(0,0,1,1)}$      \\
$0$  &  $3$  & $1$ &  ${\color{red}3},4$  \\\hline
$\check{N}_{(1,0,0,0)}$            &  $\check{N}_{(1,0,1,0)}$     &     $\check{N}_{(1,0,0,1)}$       &      $\check{N}_{(1,0,1,1)}$       	 \\
$7$ & ${\color{red}7},8$ & ${\color{red}7},8$  &  ${\color{red}7},9$ \\\hline
$\check{N}_{(0,1,0,0)}$            &  $\check{N}_{(0,1,1,0)} $     &     $\check{N}_{(0,1,0,1)}$       &      $\check{N}_{(0,1,1,1)}$       	 \\
$5$ & ${\color{red}5},6$ & ${\color{red}5},6$  &  ${\color{red}5},7$  \\\hline
$\check{N}_{(1,1,0,0)}$            &  $\check{N}_{(1,1,1,0)}$     &     $\check{N}_{(1,1,0,1)}$       &      $\check{N}_{(1,1,1,1)}$       	 \\
${\color{red}7},8$ & ${\color{red}7},9$ & ${\color{red}7},9$ &  ${\color{red}7},10$   \\\hline
\end{tabular}
\caption{For the case with $N_1=7$, $N_2=5$, $N_3=3$, $N_4=1$ analytical and numerical decoupling order for all error types are in complete agreement. The analytical overall decoupling order $\check{N}_{\min}=\min[7,5,3,1]=1$ also agrees with the actual overall decoupling order. The difference between the naive (marked as red) and the actual decoupling orders (marked as black) comes from outer-odd-UDD suppression effect. }
\label{ex:odd>}
\end{table}

The second example we consider, Table \ref{ex: inner even outer odd}, is one of the cases with all even sequence orders except the outer-most UDD layer. Analytical and numerical decoupling orders (black) are in complete agreement for all error types. There is one order difference between the naive decoupling order (red) and the actual decoupling order (black) in the last column and the last two entries in the third column. From the decoupling order formula  Eq.~\eqref{Nreven}, the difference comes from the fact that the outermost ($4^{\rm th}$) UDD layer with odd order $3$ boosts the decoupling order for error types that are also addressed by one of the inner layers. Clearly, the deviation from the standard UDD scaling can be attributed to the asymmetry of the outer layer.

The third example, shown in Table \ref{ex:odd>}, is for all odd sequence orders such that $N_1>N_2>N_3>N_4$ with $[N_{i}]_{2}=1$ for all $i$. The analytical decoupling order formula is given in Eq.~\eqref{formula:odd2}. The second term in Eq.~\eqref{formula:odd2} is called the outer-odd-UDD suppression effect, which implies a direct relationship between the number of odd parity layers which address the error, and the additional orders of error suppression achievable for a particular error type. Comparing the analytical/numerical results with the naive decoupling order, we find that the increase in decoupling order is attributed to this  outer-odd-UDD suppression effect.


Discrepancies between numerical results and analytical predictions occur when the inner UDD layers contain odd parity sequence orders that are smaller than the outer layers beyond the odd order. In Tables \ref{ex:2416} and \ref{ex:1357}, we consider two cases where deviations from the analytical predictions occur. It is important to note that although the numerical and analytical decoupling orders differ, we do not predict error suppression beyond what is {achieved}. The formula given by Eq.~(\ref{Nrl}) can be thought of as a lower bound for the actual decoupling order. Essentially, we are predicting the minimum order of error suppression for each error type.

Table \ref{ex:2416} is an example where the third layer is the only odd parity layer; the sequence orders are chosen specifically as $N_1=2$, $N_2=4$, $N_3=1$, $N_4=6$. Our analytical prediction from Eq.~\eqref{Nrl} is 
\begin{equation}
\check{N}_{\vec{r}_{4}}=\max[\{\,r_{i}N_{i}+r_{3}\}_{i=1}^{2}, r_{3}N_{3},r_{4}(r_{3}[N_{3}]_2\oplus 1)\widetilde{N}_{4}],
\label{N2416}
\end{equation}
where the suppression ability of the fourth UDD layer, $\widetilde{N}_{4}=\min[N_{4},N_{3}+1]$, is hindered by the third UDD layer if $N_{3} < N_{4}-1$. For the particular sequence orders we consider here, the analytical prediction is $\widetilde{N}_{4}=\min[6,1+1]=2$. Comparing the analytical results to numerical simulations, it appears that the analytical decoupling orders (marked in blue) are lower than the numerically calculated values for certain error types; see column 3 in Table \ref{ex:2416}. However, despite the discrepancy Eq.~(\ref{N2416}) captures the inhibiting characteristics of the odd parity inner sequence order on the decoupling order of the outer-most layer. As analytically predicted, the outer-odd-UDD suppression effects are observed on the error types, which anticommute with the third and its inner UDD layers, in the second column  by comparing their naive and the numerical decoupling orders. For the  error types  with $r_{3}=r_{4}=1$ in the fourth column of Table \ref{ex:2416}, due to $r_{3}[N_{3}]_2\oplus 1=0$, our analytical formula Eq.~(\ref{N2416}) shows that the fourth layer of UDD is totally ineffective due the odd sequence order of the third UDD layer. Indeed, the numerical decoupling order for the error types in the fourth column of Table \ref{ex:2416} is smaller than $6$, the sequence order of the fourth UDD layer.

\begin{table}[t]
\begin{tabular}{|c|c|c|c|}
\hline
$\check{N}_{(0,0,0,0)}$            &  $\check{N}_{(0,0,1,0)}$     &     $\check{N}_{(0,0,0,1)}$       &      $\check{N}_{(0,0,1,1)}$      \\
$0$ & $1$ & ${\color{red}6},3,{\color{blue}2}$  &  ${\color{red}6},1$    \\\hline
$\check{N}_{(1,0,0,0)}$            &  $\check{N}_{(1,0,1,0)}$     &     $\check{N}_{(1,0,0,1)}$       &      $\check{N}_{(1,0,1,1)}$       	 \\
$2$ & ${\color{red}2},3$ & ${\color{red}6},5,{\color{blue}2}$  &  ${\color{red}6},3$   \\\hline
$\check{N}_{(0,1,0,0)}$            &  $\check{N}_{(0,1,1,0)} $     &     $\check{N}_{(0,1,0,1)}$       &      $\check{N}_{(0,1,1,1)}$       	 \\
$4$ & ${\color{red}4},5$  &  $6,{\color{blue}4}$  &  ${\color{red}6},5$  \\\hline
$\check{N}_{(1,1,0,0)}$            &  $\check{N}_{(1,1,1,0)}$     &     $\check{N}_{(1,1,0,1)}$       &      $\check{N}_{(1,1,1,1)}$       	 \\
$4$ &  ${\color{red}4},5$  &  $6,{\color{blue}4}$ &  ${\color{red}6},5$  \\\hline
\end{tabular}
\caption{For the case with $N_1=2$, $N_2=4$, $N_3=1$, $N_4=6$, the analytical overall decoupling order $\check{N}_{\min}=\min[2,4,1,6]=1$ agrees with the actual overall decoupling order. However, the numerical decoupling order (marked in black) and our analytical predictions (marked in blue) have different values for error types in the third column. Naive decoupling order is marked in red.}
\label{ex:2416}
\end{table}
\begin{table}[t]
\begin{tabular}{|c|c|c|c|}
\hline
$\check{N}_{(0,0,0,0)}$            &  $\check{N}_{(0,0,1,0)}$     &     $\check{N}_{(0,0,0,1)}$       &      $\check{N}_{(0,0,1,1)}$      \\
$0$ & ${\color{red}5},3,{\color{blue}2}$  & ${\color{red}7},3,{\color{blue}2}$   &  ${\color{red}7},4,{\color{blue}3}$        	 \\\hline
$\check{N}_{(1,0,0,0)}$            &  $\check{N}_{(1,0,1,0)}$     &     $\check{N}_{(1,0,0,1)}$       &      $\check{N}_{(1,0,1,1)}$       	 \\
$1$ &  ${\color{red}5},2$     &     ${\color{red}7},2$       &      ${\color{red}7},5,{\color{blue}3}$       	 \\\hline
$\check{N}_{(0,1,0,0)}$            &  $\check{N}_{(0,1,1,0)} $     &     $\check{N}_{(0,1,0,1)}$       &      $\check{N}_{(0,1,1,1)}$       	 \\
$3,{\color{blue}2}$ &  ${\color{red}5},4,{\color{blue}3}$   &  ${\color{red}7},4,{\color{blue}3}$  &     $7,{\color{blue}4}$       	 \\\hline
$\check{N}_{(1,1,0,0)}$            &  $\check{N}_{(1,1,1,0)}$     &     $\check{N}_{(1,1,0,1)}$       &      $\check{N}_{(1,1,1,1)}$       	 \\
${\color{red}3},2$  &  $5,{\color{blue}3}$     &     ${\color{red}7},5,{\color{blue}3}$       &      ${\color{red}7},6,{\color{blue}4}$       	 \\\hline
\end{tabular}
\caption{For the case with  $N_1=1$, $N_2=3$, $N_3=5$, $N_4=7$, the analytical overall decoupling order $\check{N}_{\min}=\min[1,3,5,7]=1$ agrees with the actual overall decoupling order. However, the numerical decoupling order (marked in black) and our analytical predictions (marked in blue) have different values for some error types. Naive decoupling order is marked in red.}
\label{ex:1357}
\end{table}

As a final example, we consider the sequence orders $N_1=1$, $N_2=3$, $N_3=5$, and $N_4=7$. Since $\widetilde{N}_{j}=\min[N_{1}+1, N_{j}]=2< N_{j}$ for $j=2,3,4$, this indicates that the suppression abilities of the outer UDD layers are diminished due to the $N_1=1$ sequence order of the inner-most UDD layer. Indeed, the numerical decoupling orders are lower than the naive UDD scaling for a substantial fraction of error types, indicating interference between layers. However, we find that the analytically predicted decoupling orders still differ from the numerically obtained values. It appears that the actual decoupling order is not only determined by the sequence orders, but also by the error types in a complicated fashion not fully captured by the analytical decoupling formula, which, however, continues to provide a reliable lower bound on the decoupling order for each error type, as discussed earlier.


\section{Conclusions}
\label{sec:conclusion}

The NUDD scheme, which nests multiple layers of UDD sequences, is the most efficient scheme currently known for general multi-qubit decoherence suppression, assuming ideal pulses and a bath with a sharp cutoff or bounded operator norm.
In this work, we have given a rigorous analysis, along with a compact formulation, of the universality and performance of general NUDD sequence with a MOOS as {the} control pulse set.   We proved that the overall suppression order of NUDD with general control pulses is the minimum of  all the sequence orders of the individual UDD sequences comprising the NUDD scheme. Moreover, we also  obtained lower bounds on the decoupling order of each type of error, given a sequence order set. Our decoupling order formula Eq.~\eqref{Nrl} shows that only for NUDD sequences with all even sequence orders, all UDD layers work independently, i.e., the suppression ability of each UDD layer is unaffected by the presence of the other layers. For all other NUDD schemes, with at least one odd sequence order, the interference phenomenon between UDD layers appears and is summarized as follows,
\begin{enumerate}
\item For a given UDD layer, say the $i^{\rm th}$ UDD layer with sequence order $N_{i}$, if there are inner layers with odd sequence order and the lowest odd order is smaller than  $N_{i}-1$, then the suppression ability of the $i^{\rm th}$ UDD layer is hindered by this inner UDD layer with lowest odd order, i.e., it  cannot achieve $N_{i}^{\textrm{th}}$-order decoupling. [$\widetilde{N}_{i}=N^{k'<i}_{o_{\min}}+1<N_{i}$ in Eq.~\eqref{Nrl}.]
\item For the $i^{\textrm{th}}$ UDD layer to be effective against a given error type, this error  needs to not only anticommute with the  control pulses  of the $i^{\textrm{th}}$ UDD layer ($r_{i}=1$) but also anticommute with an even number of UDD layers with odd  sequence orders before the $i^{\textrm{th}}$ UDD layer ($p_\oplus(1,i-1)=0$).
\item For a given error type, if there is a total odd number of UDD layers with odd sequence orders before the $i^{\textrm{th}}$ UDD layer that the error anticommutes with ($p_\oplus(1,i-1)=1$), then the $i^{\textrm{th}}$ UDD layer is totally ineffective ($p_\oplus(1,i-1)\oplus 1=0$) irrespective of whether this error anticommutes with this layer or not.
\item For a given error type,  each odd order UDD layer that the error anticommutes with and is nested outside the $i^{\textrm{th}}$ UDD layer, can enhance the suppression ability of the $i^{\textrm{th}}$ UDD layer by one more order (outer-odd-UDD effect) on this error type. In other words, the outer-odd-UDD suppression effect is cumulative and is responsible for the $p_+(i+1,\ell)$ term in the decoupling order formula Eq.~\eqref{Nrl}.
\end{enumerate}
Since our analysis identifies the conditions under which the suppression ability of a given UDD layer is inhibited, or made totally ineffective, or rather enhanced by other UDD layers with odd sequence orders, one can use it to design an \emph{optimally ordered} NUDD scheme that exploits the full power of each UDD layer. To be more specific, suppose one would like to design an NUDD scheme from some UDD sequences whose  control pulse types and sequence orders are given. From the analysis presented here, in order to reach optimal efficiency, first one should nest all the UDD layers with even sequence orders together, where the nesting orders can be arbitrary, and denote this resulting sequence as NUDD$_{e}$; second, nest all the UDD layers with odd sequence orders together such that the sequence orders from the inner-most to the outer-most layers are decreasing, and denote this resulting sequence as NUDD$_{o}$; the final and optimal NUDD scheme is constructed by nesting NUDD$_{e}$ as the inner sequence with NUDD$_{o}$ as the outer sequence. 

An important challenge is to generalize the analysis we have presented here to the setting of non-ideal, finite width pulses.  

\begin{acknowledgments}
This research was supported by the ARO MURI grant W911NF-11-1-0268, by the Department of Defense, by the Intelligence Advanced Research Projects Activity (IARPA)
via Department of Interior National Business Center contract number
D11PC20165, and by
NSF grants No. CHE-924318 and CHE-1037992. 
The U.S. Government is authorized to reproduce and distribute
reprints for Governmental purposes notwithstanding any copyright annotation
thereon. The views and conclusions contained herein are those of
the authors and should not be interpreted as necessarily representing the
official policies or endorsements, either expressed or implied, of IARPA,
DoI/NBC, or the U.S. Government.
\end{acknowledgments}


\begin{thebibliography}{44}%
\makeatletter
\providecommand \@ifxundefined [1]{%
 \@ifx{#1\undefined}
}%
\providecommand \@ifnum [1]{%
 \ifnum #1\expandafter \@firstoftwo
 \else \expandafter \@secondoftwo
 \fi
}%
\providecommand \@ifx [1]{%
 \ifx #1\expandafter \@firstoftwo
 \else \expandafter \@secondoftwo
 \fi
}%
\providecommand \natexlab [1]{#1}%
\providecommand \enquote  [1]{``#1''}%
\providecommand \bibnamefont  [1]{#1}%
\providecommand \bibfnamefont [1]{#1}%
\providecommand \citenamefont [1]{#1}%
\providecommand \href@noop [0]{\@secondoftwo}%
\providecommand \href [0]{\begingroup \@sanitize@url \@href}%
\providecommand \@href[1]{\@@startlink{#1}\@@href}%
\providecommand \@@href[1]{\endgroup#1\@@endlink}%
\providecommand \@sanitize@url [0]{\catcode `\\12\catcode `\$12\catcode
  `\&12\catcode `\#12\catcode `\^12\catcode `\_12\catcode `\%12\relax}%
\providecommand \@@startlink[1]{}%
\providecommand \@@endlink[0]{}%
\providecommand \url  [0]{\begingroup\@sanitize@url \@url }%
\providecommand \@url [1]{\endgroup\@href {#1}{\urlprefix }}%
\providecommand \urlprefix  [0]{URL }%
\providecommand \Eprint [0]{\href }%
\@ifxundefined \urlstyle {%
  \providecommand \doi  [0]{\begingroup \@sanitize@url \@doi}%
  \providecommand \@doi [1]{\endgroup \@@startlink {\doibase
  #1}doi:\discretionary {}{}{}#1\@@endlink }%
}{%
  \providecommand \doi  [0]{doi:\discretionary{}{}{}\begingroup
  \urlstyle{rm}\Url }%
}%
\providecommand \doibase [0]{http://dx.doi.org/}%
\providecommand \Doi [0]{\begingroup \@sanitize@url \@Doi }%
\providecommand \@Doi  [1]{\endgroup\@@startlink{\doibase#1}\@@Doi}%
\providecommand \@@Doi [1]{#1\@@endlink}%
\providecommand \selectlanguage [0]{\@gobble}%
\providecommand \bibinfo  [0]{\@secondoftwo}%
\providecommand \bibfield  [0]{\@secondoftwo}%
\providecommand \translation [1]{[#1]}%
\providecommand \BibitemOpen [0]{}%
\providecommand \bibitemStop [0]{}%
\providecommand \bibitemNoStop [0]{.\EOS\space}%
\providecommand \EOS [0]{\spacefactor3000\relax}%
\providecommand \BibitemShut  [1]{\csname bibitem#1\endcsname}%
\bibitem [{\citenamefont {Breuer}\ and\ \citenamefont
  {Petruccione}(2002)}]{Breuer:book}%
  \BibitemOpen
  \bibfield  {author} {\bibinfo {author} {\bibfnamefont {H.-P.}\ \bibnamefont
  {Breuer}}\ and\ \bibinfo {author} {\bibfnamefont {F.}~\bibnamefont
  {Petruccione}},\ }\href@noop {} {\emph {\bibinfo {title} {The Theory of Open
  Quantum Systems}}}\ (\bibinfo  {publisher} {Oxford University Press},\
  \bibinfo {address} {Oxford},\ \bibinfo {year} {2002})\BibitemShut {NoStop}%
\bibitem [{\citenamefont {Nielsen}\ and\ \citenamefont
  {Chuang}(2000)}]{NielsenChuang:book}%
  \BibitemOpen
  \bibfield  {author} {\bibinfo {author} {\bibfnamefont {M.}~\bibnamefont
  {Nielsen}}\ and\ \bibinfo {author} {\bibfnamefont {I.}~\bibnamefont
  {Chuang}},\ }\href@noop {} {\emph {\bibinfo {title} {Quantum Computation and
  Quantum Information}}}\ (\bibinfo  {publisher} {Cambridge University Press},\
  \bibinfo {address} {Cambridge, England},\ \bibinfo {year} {2000})\BibitemShut
  {NoStop}%
\bibitem [{\citenamefont {Ladd}\ \emph {et~al.}(2010)\citenamefont {Ladd},
  \citenamefont {Jelezko}, \citenamefont {Laflamme}, \citenamefont {Nakamura},
  \citenamefont {Monroe},\ and\ \citenamefont {O'Brien}}]{Ladd:10}%
  \BibitemOpen
  \bibfield  {author} {\bibinfo {author} {\bibfnamefont {T.~D.}\ \bibnamefont
  {Ladd}}, \bibinfo {author} {\bibfnamefont {F.}~\bibnamefont {Jelezko}},
  \bibinfo {author} {\bibfnamefont {R.}~\bibnamefont {Laflamme}}, \bibinfo
  {author} {\bibfnamefont {Y.}~\bibnamefont {Nakamura}}, \bibinfo {author}
  {\bibfnamefont {C.}~\bibnamefont {Monroe}}, \ and\ \bibinfo {author}
  {\bibfnamefont {J.~L.}\ \bibnamefont {O'Brien}},\ }\Doi {10.1038/nature08812}
  {\bibfield  {journal} {\bibinfo  {journal} {Nature},\ }\textbf {\bibinfo
  {volume} {464}},\ \bibinfo {pages} {45} (\bibinfo {year} {2010})}\BibitemShut
  {NoStop}%
\bibitem [{\citenamefont {Hahn}(1950)}]{Hahn:50}%
  \BibitemOpen
  \bibfield  {author} {\bibinfo {author} {\bibfnamefont {E.~L.}\ \bibnamefont
  {Hahn}},\ }\Doi {10.1103/PhysRev.80.580} {\bibfield  {journal} {\bibinfo
  {journal} {Phys. Rev.},\ }\textbf {\bibinfo {volume} {80}},\ \bibinfo {pages}
  {580} (\bibinfo {year} {1950})}\BibitemShut {NoStop}%
\bibitem [{\citenamefont {Haeberlen}(1976)}]{Haeberlen:book}%
  \BibitemOpen
  \bibfield  {author} {\bibinfo {author} {\bibfnamefont {U.}~\bibnamefont
  {Haeberlen}},\ }\href@noop {} {\emph {\bibinfo {title} {High Resolution NMR
  in Solids}}},\ Advances in Magnetic Resonance Series, Supplement 1\ (\bibinfo
   {publisher} {Academic Press},\ \bibinfo {address} {New York},\ \bibinfo
  {year} {1976})\BibitemShut {NoStop}%
\bibitem [{\citenamefont {Yang}\ \emph {et~al.}(2011)\citenamefont {Yang},
  \citenamefont {Wang},\ and\ \citenamefont {Liu}}]{Yang-DD-review}%
  \BibitemOpen
  \bibfield  {author} {\bibinfo {author} {\bibfnamefont {W.}~\bibnamefont
  {Yang}}, \bibinfo {author} {\bibfnamefont {Z.-Y.}\ \bibnamefont {Wang}}, \
  and\ \bibinfo {author} {\bibfnamefont {R.-B.}\ \bibnamefont {Liu}},\ }\Doi
  {10.1007/s11467-010-0113-8} {\bibfield  {journal} {\bibinfo  {journal}
  {Front. Phys.},\ }\textbf {\bibinfo {volume} {6}},\ \bibinfo {pages} {2}
  (\bibinfo {year} {2011})}\BibitemShut {NoStop}%
\bibitem [{\citenamefont {Viola}\ and\ \citenamefont
  {Lloyd}(1998)}]{ViolaLloyd:98}%
  \BibitemOpen
  \bibfield  {author} {\bibinfo {author} {\bibfnamefont {L.}~\bibnamefont
  {Viola}}\ and\ \bibinfo {author} {\bibfnamefont {S.}~\bibnamefont {Lloyd}},\
  }\Doi {10.1103/PhysRevA.58.2733} {\bibfield  {journal} {\bibinfo  {journal}
  {Phys. Rev. A},\ }\textbf {\bibinfo {volume} {58}},\ \bibinfo {pages} {2733}
  (\bibinfo {year} {1998})}\BibitemShut {NoStop}%
\bibitem [{\citenamefont {Ban}(1998)}]{Ban:98}%
  \BibitemOpen
  \bibfield  {author} {\bibinfo {author} {\bibfnamefont {M.}~\bibnamefont
  {Ban}},\ }\Doi {10.1080/09500349808231241} {\bibfield  {journal} {\bibinfo
  {journal} {Journal of Modern Optics},\ }\textbf {\bibinfo {volume} {45}},\
  \bibinfo {pages} {2315 } (\bibinfo {year} {1998})}\BibitemShut {NoStop}%
\bibitem [{\citenamefont {Duan}\ and\ \citenamefont {Guo}(1999)}]{Duan:98e}%
  \BibitemOpen
  \bibfield  {author} {\bibinfo {author} {\bibfnamefont {L.-M.}\ \bibnamefont
  {Duan}}\ and\ \bibinfo {author} {\bibfnamefont {G.}~\bibnamefont {Guo}},\
  }\Doi {10.1016/S0375-9601(99)00592-7} {\bibfield  {journal} {\bibinfo
  {journal} {Phys. Lett. A},\ }\textbf {\bibinfo {volume} {261}},\ \bibinfo
  {pages} {139} (\bibinfo {year} {1999})}\BibitemShut {NoStop}%
\bibitem [{\citenamefont {Zanardi}(1999)}]{Zanardi:99}%
  \BibitemOpen
  \bibfield  {author} {\bibinfo {author} {\bibfnamefont {P.}~\bibnamefont
  {Zanardi}},\ }\Doi {10.1016/S0375-9601(99)00365-5} {\bibfield  {journal}
  {\bibinfo  {journal} {Phys. Lett. A},\ }\textbf {\bibinfo {volume} {258}},\
  \bibinfo {pages} {77} (\bibinfo {year} {1999})}\BibitemShut {NoStop}%
\bibitem [{\citenamefont {Viola}\ \emph {et~al.}(1999)\citenamefont {Viola},
  \citenamefont {Knill},\ and\ \citenamefont {Lloyd}}]{ViolaKnillLloyd:99}%
  \BibitemOpen
  \bibfield  {author} {\bibinfo {author} {\bibfnamefont {L.}~\bibnamefont
  {Viola}}, \bibinfo {author} {\bibfnamefont {E.}~\bibnamefont {Knill}}, \ and\
  \bibinfo {author} {\bibfnamefont {S.}~\bibnamefont {Lloyd}},\ }\Doi
  {10.1103/PhysRevLett.82.2417} {\bibfield  {journal} {\bibinfo  {journal}
  {Phys. Rev. Lett.},\ }\textbf {\bibinfo {volume} {82}},\ \bibinfo {pages}
  {2417} (\bibinfo {year} {1999})}\BibitemShut {NoStop}%
\bibitem [{\citenamefont {Khodjasteh}\ and\ \citenamefont
  {Lidar}(2005)}]{KhodjastehLidar:05}%
  \BibitemOpen
  \bibfield  {author} {\bibinfo {author} {\bibfnamefont {K.}~\bibnamefont
  {Khodjasteh}}\ and\ \bibinfo {author} {\bibfnamefont {D.~A.}\ \bibnamefont
  {Lidar}},\ }\Doi {10.1103/PhysRevLett.95.180501} {\bibfield  {journal}
  {\bibinfo  {journal} {Phys. Rev. Lett.},\ }\textbf {\bibinfo {volume} {95}},\
  \bibinfo {pages} {180501} (\bibinfo {year} {2005})}\BibitemShut {NoStop}%
\bibitem [{\citenamefont {Khodjasteh}\ and\ \citenamefont
  {Lidar}(2007)}]{KhodjastehLidar:07}%
  \BibitemOpen
  \bibfield  {author} {\bibinfo {author} {\bibfnamefont {K.}~\bibnamefont
  {Khodjasteh}}\ and\ \bibinfo {author} {\bibfnamefont {D.~A.}\ \bibnamefont
  {Lidar}},\ }\Doi {10.1103/PhysRevA.75.062310} {\bibfield  {journal} {\bibinfo
   {journal} {Phys. Rev. A},\ }\textbf {\bibinfo {volume} {75}},\ \bibinfo
  {pages} {062310} (\bibinfo {year} {2007})}\BibitemShut {NoStop}%
\bibitem [{\citenamefont {Peng}\ \emph {et~al.}(2011)\citenamefont {Peng},
  \citenamefont {Suter},\ and\ \citenamefont {Lidar}}]{Peng:11}%
  \BibitemOpen
  \bibfield  {author} {\bibinfo {author} {\bibfnamefont {X.}~\bibnamefont
  {Peng}}, \bibinfo {author} {\bibfnamefont {D.}~\bibnamefont {Suter}}, \ and\
  \bibinfo {author} {\bibfnamefont {D.}~\bibnamefont {Lidar}},\ }\Doi
  {10.1088/0953-4075/44/15/154003} {\bibfield  {journal} {\bibinfo  {journal}
  {J. Phys. B: At. Mol. Opt. Phys.},\ }\textbf {\bibinfo {volume} {44}},\
  \bibinfo {pages} {154003} (\bibinfo {year} {2011})}\BibitemShut {NoStop}%
\bibitem [{\citenamefont {\'Alvarez}\ \emph {et~al.}(2010)\citenamefont
  {\'Alvarez}, \citenamefont {Ajoy}, \citenamefont {Peng},\ and\ \citenamefont
  {Suter}}]{AlvarezAjoyPengSuter:10}%
  \BibitemOpen
  \bibfield  {author} {\bibinfo {author} {\bibfnamefont {G.~A.}\ \bibnamefont
  {\'Alvarez}}, \bibinfo {author} {\bibfnamefont {A.}~\bibnamefont {Ajoy}},
  \bibinfo {author} {\bibfnamefont {X.}~\bibnamefont {Peng}}, \ and\ \bibinfo
  {author} {\bibfnamefont {D.}~\bibnamefont {Suter}},\ }\Doi
  {10.1103/PhysRevA.82.042306} {\bibfield  {journal} {\bibinfo  {journal}
  {Phys. Rev. A},\ }\textbf {\bibinfo {volume} {82}},\ \bibinfo {pages}
  {042306} (\bibinfo {year} {2010})}\BibitemShut {NoStop}%
\bibitem [{\citenamefont {{Tyryshkin}}\ \emph {et~al.}(2010)\citenamefont
  {{Tyryshkin}}, \citenamefont {{Wang}}, \citenamefont {{Zhang}}, \citenamefont
  {{Haller}}, \citenamefont {{Ager}}, \citenamefont {{Dobrovitski}},\ and\
  \citenamefont {{Lyon}}}]{2010arXiv1011.1903T}%
  \BibitemOpen
  \bibfield  {author} {\bibinfo {author} {\bibfnamefont {A.~M.}\ \bibnamefont
  {{Tyryshkin}}}, \bibinfo {author} {\bibfnamefont {Z.}~\bibnamefont {{Wang}}},
  \bibinfo {author} {\bibfnamefont {W.}~\bibnamefont {{Zhang}}}, \bibinfo
  {author} {\bibfnamefont {E.~E.}\ \bibnamefont {{Haller}}}, \bibinfo {author}
  {\bibfnamefont {J.~W.}\ \bibnamefont {{Ager}}}, \bibinfo {author}
  {\bibfnamefont {V.~V.}\ \bibnamefont {{Dobrovitski}}}, \ and\ \bibinfo
  {author} {\bibfnamefont {S.~A.}\ \bibnamefont {{Lyon}}},\ }\href@noop {} {
  (\bibinfo {year} {2010})},\ \Eprint {http://arxiv.org/abs/arxiv:1011.1903}
  {arxiv:1011.1903} \BibitemShut {NoStop}%
\bibitem [{\citenamefont {{Wang}}\ \emph {et~al.}(2010)\citenamefont {{Wang}},
  \citenamefont {{Zhang}}, \citenamefont {{Tyryshkin}}, \citenamefont {{Lyon}},
  \citenamefont {{Ager}}, \citenamefont {{Haller}},\ and\ \citenamefont
  {{Dobrovitski}}}]{2010arXiv1011.6417W}%
  \BibitemOpen
  \bibfield  {author} {\bibinfo {author} {\bibfnamefont {Z.}~\bibnamefont
  {{Wang}}}, \bibinfo {author} {\bibfnamefont {W.}~\bibnamefont {{Zhang}}},
  \bibinfo {author} {\bibfnamefont {A.~M.}\ \bibnamefont {{Tyryshkin}}},
  \bibinfo {author} {\bibfnamefont {S.~A.}\ \bibnamefont {{Lyon}}}, \bibinfo
  {author} {\bibfnamefont {J.~W.}\ \bibnamefont {{Ager}}}, \bibinfo {author}
  {\bibfnamefont {E.~E.}\ \bibnamefont {{Haller}}}, \ and\ \bibinfo {author}
  {\bibfnamefont {V.~V.}\ \bibnamefont {{Dobrovitski}}},\ }\href@noop {} {
  (\bibinfo {year} {2010})},\ \Eprint {http://arxiv.org/abs/1011.6417}
  {arXiv:1011.6417 [quant-ph]} \BibitemShut {NoStop}%
\bibitem [{\citenamefont {Barthel}\ \emph {et~al.}(2010)\citenamefont
  {Barthel}, \citenamefont {Medford}, \citenamefont {Marcus}, \citenamefont
  {Hanson},\ and\ \citenamefont {Gossard}}]{Barthel:10}%
  \BibitemOpen
  \bibfield  {author} {\bibinfo {author} {\bibfnamefont {C.}~\bibnamefont
  {Barthel}}, \bibinfo {author} {\bibfnamefont {J.}~\bibnamefont {Medford}},
  \bibinfo {author} {\bibfnamefont {C.~M.}\ \bibnamefont {Marcus}}, \bibinfo
  {author} {\bibfnamefont {M.~P.}\ \bibnamefont {Hanson}}, \ and\ \bibinfo
  {author} {\bibfnamefont {A.~C.}\ \bibnamefont {Gossard}},\ }\Doi
  {10.1103/PhysRevLett.105.266808} {\bibfield  {journal} {\bibinfo  {journal}
  {Phys. Rev. Lett.},\ }\textbf {\bibinfo {volume} {105}},\ \bibinfo {pages}
  {266808} (\bibinfo {year} {2010})}\BibitemShut {NoStop}%
\bibitem [{\citenamefont {Uhrig}(2007)}]{Uhrig:07}%
  \BibitemOpen
  \bibfield  {author} {\bibinfo {author} {\bibfnamefont {G.}~\bibnamefont
  {Uhrig}},\ }\Doi {10.1103/PhysRevLett.98.100504} {\bibfield  {journal}
  {\bibinfo  {journal} {Phys. Rev. Lett.},\ }\textbf {\bibinfo {volume} {98}},\
  \bibinfo {pages} {100504} (\bibinfo {year} {2007})}\BibitemShut {NoStop}%
\bibitem [{\citenamefont {Uhrig}(2008)}]{Uhrig:08}%
  \BibitemOpen
  \bibfield  {author} {\bibinfo {author} {\bibfnamefont {G.~S.}\ \bibnamefont
  {Uhrig}},\ }\Doi {10.1088/1367-2630/13/5/059504} {\bibfield  {journal}
  {\bibinfo  {journal} {New. J. Phys},\ }\textbf {\bibinfo {volume} {10}},\
  \bibinfo {pages} {083024} (\bibinfo {year} {2008})}\BibitemShut {NoStop}%
\bibitem [{\citenamefont {Yang}\ and\ \citenamefont {Liu}(2008)}]{YangLiu:08}%
  \BibitemOpen
  \bibfield  {author} {\bibinfo {author} {\bibfnamefont {W.}~\bibnamefont
  {Yang}}\ and\ \bibinfo {author} {\bibfnamefont {R.-B.}\ \bibnamefont {Liu}},\
  }\Doi {10.1103/PhysRevLett.101.180403} {\bibfield  {journal} {\bibinfo
  {journal} {Phys. Rev. Lett.},\ }\textbf {\bibinfo {volume} {101}},\ \bibinfo
  {pages} {180403} (\bibinfo {year} {2008})}\BibitemShut {NoStop}%
\bibitem [{\citenamefont {Uhrig}\ and\ \citenamefont {Lidar}(2010)}]{UL:10}%
  \BibitemOpen
  \bibfield  {author} {\bibinfo {author} {\bibfnamefont {G.}~\bibnamefont
  {Uhrig}}\ and\ \bibinfo {author} {\bibfnamefont {D.}~\bibnamefont {Lidar}},\
  }\Doi {10.1103/PhysRevA.82.012301} {\bibfield  {journal} {\bibinfo  {journal}
  {Phys. Rev. A},\ }\textbf {\bibinfo {volume} {82}},\ \bibinfo {pages}
  {012301} (\bibinfo {year} {2010})}\BibitemShut {NoStop}%
\bibitem [{\citenamefont {Pasini}\ and\ \citenamefont
  {Uhrig}(2010)}]{PasiniUhrig:10}%
  \BibitemOpen
  \bibfield  {author} {\bibinfo {author} {\bibfnamefont {S.}~\bibnamefont
  {Pasini}}\ and\ \bibinfo {author} {\bibfnamefont {G.~S.}\ \bibnamefont
  {Uhrig}},\ }\Doi {10.1088/1751-8113/43/13/132001} {\bibfield  {journal}
  {\bibinfo  {journal} {J. Phys. A: Math. Theor.},\ }\textbf {\bibinfo {volume}
  {43}},\ \bibinfo {pages} {132001} (\bibinfo {year} {2010})}\BibitemShut
  {NoStop}%
\bibitem [{\citenamefont {Dhar}\ \emph {et~al.}(2006)\citenamefont {Dhar},
  \citenamefont {Grover},\ and\ \citenamefont {Roy}}]{DharGroverRoy:06}%
  \BibitemOpen
  \bibfield  {author} {\bibinfo {author} {\bibfnamefont {D.}~\bibnamefont
  {Dhar}}, \bibinfo {author} {\bibfnamefont {L.~K.}\ \bibnamefont {Grover}}, \
  and\ \bibinfo {author} {\bibfnamefont {S.~M.}\ \bibnamefont {Roy}},\ }\Doi
  {10.1103/PhysRevLett.96.100405} {\bibfield  {journal} {\bibinfo  {journal}
  {Phys. Rev. Lett.},\ }\textbf {\bibinfo {volume} {96}},\ \bibinfo {pages}
  {100405} (\bibinfo {year} {2006})}\BibitemShut {NoStop}%
\bibitem [{\citenamefont {Lee}\ \emph {et~al.}(2008)\citenamefont {Lee},
  \citenamefont {Witzel},\ and\ \citenamefont
  {Das~Sarma}}]{LeeWitzelDasSarma:08}%
  \BibitemOpen
  \bibfield  {author} {\bibinfo {author} {\bibfnamefont {B.}~\bibnamefont
  {Lee}}, \bibinfo {author} {\bibfnamefont {W.~M.}\ \bibnamefont {Witzel}}, \
  and\ \bibinfo {author} {\bibfnamefont {S.}~\bibnamefont {Das~Sarma}},\ }\Doi
  {10.1103/PhysRevLett.100.160505} {\bibfield  {journal} {\bibinfo  {journal}
  {Phys. Rev. Lett.},\ }\textbf {\bibinfo {volume} {100}},\ \bibinfo {pages}
  {160505} (\bibinfo {year} {2008})}\BibitemShut {NoStop}%
\bibitem [{\citenamefont {Cywi\ifmmode~\acute{n}\else \'{n}\fi{}ski}\ \emph
  {et~al.}(2008)\citenamefont {Cywi\ifmmode~\acute{n}\else \'{n}\fi{}ski},
  \citenamefont {Lutchyn}, \citenamefont {Nave},\ and\ \citenamefont
  {Das~Sarma}}]{Cywinskietal:08}%
  \BibitemOpen
  \bibfield  {author} {\bibinfo {author} {\bibfnamefont {L.}~\bibnamefont
  {Cywi\ifmmode~\acute{n}\else \'{n}\fi{}ski}}, \bibinfo {author}
  {\bibfnamefont {R.~M.}\ \bibnamefont {Lutchyn}}, \bibinfo {author}
  {\bibfnamefont {C.~P.}\ \bibnamefont {Nave}}, \ and\ \bibinfo {author}
  {\bibfnamefont {S.}~\bibnamefont {Das~Sarma}},\ }\Doi
  {10.1103/PhysRevB.77.174509} {\bibfield  {journal} {\bibinfo  {journal}
  {Phys. Rev. B},\ }\textbf {\bibinfo {volume} {77}},\ \bibinfo {pages}
  {174509} (\bibinfo {year} {2008})}\BibitemShut {NoStop}%
\bibitem [{\citenamefont {Biercuk}\ \emph
  {et~al.}(2009){\natexlab{a}}\citenamefont {Biercuk}, \citenamefont {Uys},
  \citenamefont {VanDevender}, \citenamefont {Shiga}, \citenamefont {Itano},\
  and\ \citenamefont {Bollinger}}]{Biercuk2:09}%
  \BibitemOpen
  \bibfield  {author} {\bibinfo {author} {\bibfnamefont {M.~J.}\ \bibnamefont
  {Biercuk}}, \bibinfo {author} {\bibfnamefont {H.}~\bibnamefont {Uys}},
  \bibinfo {author} {\bibfnamefont {A.~P.}\ \bibnamefont {VanDevender}},
  \bibinfo {author} {\bibfnamefont {N.}~\bibnamefont {Shiga}}, \bibinfo
  {author} {\bibfnamefont {W.~M.}\ \bibnamefont {Itano}}, \ and\ \bibinfo
  {author} {\bibfnamefont {J.~J.}\ \bibnamefont {Bollinger}},\ }\Doi
  {10.1103/PhysRevA.79.062324} {\bibfield  {journal} {\bibinfo  {journal}
  {Phys. Rev. A},\ }\textbf {\bibinfo {volume} {79}},\ \bibinfo {pages}
  {062324} (\bibinfo {year} {2009}{\natexlab{a}})}\BibitemShut {NoStop}%
\bibitem [{\citenamefont {Biercuk}\ \emph
  {et~al.}(2009){\natexlab{b}}\citenamefont {Biercuk}, \citenamefont {Uys},
  \citenamefont {VanDevender}, \citenamefont {Shiga}, \citenamefont {Itano},\
  and\ \citenamefont {Bollinger}}]{Biercuk:09}%
  \BibitemOpen
  \bibfield  {author} {\bibinfo {author} {\bibfnamefont {M.~J.}\ \bibnamefont
  {Biercuk}}, \bibinfo {author} {\bibfnamefont {H.}~\bibnamefont {Uys}},
  \bibinfo {author} {\bibfnamefont {A.~P.}\ \bibnamefont {VanDevender}},
  \bibinfo {author} {\bibfnamefont {N.}~\bibnamefont {Shiga}}, \bibinfo
  {author} {\bibfnamefont {W.~M.}\ \bibnamefont {Itano}}, \ and\ \bibinfo
  {author} {\bibfnamefont {J.~J.}\ \bibnamefont {Bollinger}},\ }\Doi
  {10.1038/nature07951} {\bibfield  {journal} {\bibinfo  {journal} {Nature},\
  }\textbf {\bibinfo {volume} {458}},\ \bibinfo {pages} {996} (\bibinfo {year}
  {2009}{\natexlab{b}})}\BibitemShut {NoStop}%
\bibitem [{\citenamefont {Du}\ \emph {et~al.}(2009)\citenamefont {Du},
  \citenamefont {Rong}, \citenamefont {Zhao}, \citenamefont {Wang},
  \citenamefont {Yang},\ and\ \citenamefont {Liu}}]{Du:09}%
  \BibitemOpen
  \bibfield  {author} {\bibinfo {author} {\bibfnamefont {J.}~\bibnamefont
  {Du}}, \bibinfo {author} {\bibfnamefont {X.}~\bibnamefont {Rong}}, \bibinfo
  {author} {\bibfnamefont {N.}~\bibnamefont {Zhao}}, \bibinfo {author}
  {\bibfnamefont {Y.}~\bibnamefont {Wang}}, \bibinfo {author} {\bibfnamefont
  {J.}~\bibnamefont {Yang}}, \ and\ \bibinfo {author} {\bibfnamefont {R.~B.}\
  \bibnamefont {Liu}},\ }\Doi {10.1038/nature08470} {\bibfield  {journal}
  {\bibinfo  {journal} {Nature},\ }\textbf {\bibinfo {volume} {461}},\ \bibinfo
  {pages} {1265} (\bibinfo {year} {2009})}\BibitemShut {NoStop}%
\bibitem [{\citenamefont {Jenista}\ \emph {et~al.}(2009)\citenamefont
  {Jenista}, \citenamefont {Stokes}, \citenamefont {Branca},\ and\
  \citenamefont {Warren}}]{Elizabeth:09}%
  \BibitemOpen
  \bibfield  {author} {\bibinfo {author} {\bibfnamefont {E.~R.}\ \bibnamefont
  {Jenista}}, \bibinfo {author} {\bibfnamefont {A.~M.}\ \bibnamefont {Stokes}},
  \bibinfo {author} {\bibfnamefont {R.~T.}\ \bibnamefont {Branca}}, \ and\
  \bibinfo {author} {\bibfnamefont {W.~S.}\ \bibnamefont {Warren}},\ }\Doi
  {10.1063/1.3263196} {\bibfield  {journal} {\bibinfo  {journal} {J. Chem.
  Phys},\ }\textbf {\bibinfo {volume} {31}},\ \bibinfo {pages} {204510}
  (\bibinfo {year} {2009})}\BibitemShut {NoStop}%
\bibitem [{\citenamefont {Wang}\ and\ \citenamefont
  {Dobrovitski}(2011)}]{WangDobrovitski:11}%
  \BibitemOpen
  \bibfield  {author} {\bibinfo {author} {\bibfnamefont {Z.-H.}\ \bibnamefont
  {Wang}}\ and\ \bibinfo {author} {\bibfnamefont {V.~V.}\ \bibnamefont
  {Dobrovitski}},\ }\Doi {10.1088/0953-4075/44/15/154004} {\bibfield  {journal}
  {\bibinfo  {journal} {J. Phys. B: At. Mol. Opt. Phys},\ }\textbf {\bibinfo
  {volume} {44}},\ \bibinfo {pages} {154004} (\bibinfo {year}
  {2011})}\BibitemShut {NoStop}%
\bibitem [{\citenamefont {Almog}\ \emph {et~al.}(2011)\citenamefont {Almog},
  \citenamefont {Sagi}, \citenamefont {Gordon}, \citenamefont {Bensky},
  \citenamefont {Kurizki},\ and\ \citenamefont {Davidson}}]{Almog1Davidson:11}%
  \BibitemOpen
  \bibfield  {author} {\bibinfo {author} {\bibfnamefont {I.}~\bibnamefont
  {Almog}}, \bibinfo {author} {\bibfnamefont {Y.}~\bibnamefont {Sagi}},
  \bibinfo {author} {\bibfnamefont {G.}~\bibnamefont {Gordon}}, \bibinfo
  {author} {\bibfnamefont {G.}~\bibnamefont {Bensky}}, \bibinfo {author}
  {\bibfnamefont {G.}~\bibnamefont {Kurizki}}, \ and\ \bibinfo {author}
  {\bibfnamefont {N.}~\bibnamefont {Davidson}},\ }\Doi
  {10.1088/0953-4075/44/15/154006} {\bibfield  {journal} {\bibinfo  {journal}
  {J. Phys. B: At. Mol. Opt. Phys},\ }\textbf {\bibinfo {volume} {44}},\
  \bibinfo {pages} {154006} (\bibinfo {year} {2011})}\BibitemShut {NoStop}%
\bibitem [{\citenamefont {Young}\ and\ \citenamefont
  {Whaley}(2011)}]{YoungWhaley:11}%
  \BibitemOpen
  \bibfield  {author} {\bibinfo {author} {\bibfnamefont {K.~C.}\ \bibnamefont
  {Young}}\ and\ \bibinfo {author} {\bibfnamefont {K.~B.}\ \bibnamefont
  {Whaley}},\ }\href {http://arxiv.org/abs/1102.5115} {\bibfield  {journal}
  {\bibinfo  {journal} {arXiv:1102.5115v1}} (\bibinfo {year}
  {2011})}\BibitemShut {NoStop}%
\bibitem [{\citenamefont {Xia}\ \emph {et~al.}(2011)\citenamefont {Xia},
  \citenamefont {Uhrig},\ and\ \citenamefont {Lidar}}]{Xia:11}%
  \BibitemOpen
  \bibfield  {author} {\bibinfo {author} {\bibfnamefont {Y.}~\bibnamefont
  {Xia}}, \bibinfo {author} {\bibfnamefont {G.~S.}\ \bibnamefont {Uhrig}}, \
  and\ \bibinfo {author} {\bibfnamefont {D.~A.}\ \bibnamefont {Lidar}},\ }\Doi
  {10.1103/PhysRevA.84.062332} {\bibfield  {journal} {\bibinfo  {journal}
  {Phys. Rev. A},\ }\textbf {\bibinfo {volume} {84}},\ \bibinfo {pages}
  {062332} (\bibinfo {year} {2011})}\BibitemShut {NoStop}%
\bibitem [{\citenamefont {West}\ \emph {et~al.}(2010)\citenamefont {West},
  \citenamefont {Fong},\ and\ \citenamefont {Lidar}}]{WestFongLidar:10}%
  \BibitemOpen
  \bibfield  {author} {\bibinfo {author} {\bibfnamefont {J.~R.}\ \bibnamefont
  {West}}, \bibinfo {author} {\bibfnamefont {B.~H.}\ \bibnamefont {Fong}}, \
  and\ \bibinfo {author} {\bibfnamefont {D.~A.}\ \bibnamefont {Lidar}},\ }\Doi
  {10.1103/PhysRevLett.104.130501} {\bibfield  {journal} {\bibinfo  {journal}
  {Phys. Rev. Lett.},\ }\textbf {\bibinfo {volume} {104}},\ \bibinfo {pages}
  {130501} (\bibinfo {year} {2010})}\BibitemShut {NoStop}%
\bibitem [{\citenamefont {Wang}\ and\ \citenamefont {Liu}(2011)}]{WangLiu:11}%
  \BibitemOpen
  \bibfield  {author} {\bibinfo {author} {\bibfnamefont {Z.-Y.}\ \bibnamefont
  {Wang}}\ and\ \bibinfo {author} {\bibfnamefont {R.-B.}\ \bibnamefont {Liu}},\
  }\Doi {10.1103/PhysRevA.83.022306} {\bibfield  {journal} {\bibinfo  {journal}
  {Phys. Rev. A},\ }\textbf {\bibinfo {volume} {83}},\ \bibinfo {pages}
  {022306} (\bibinfo {year} {2011})}\BibitemShut {NoStop}%
\bibitem [{\citenamefont {Uhrig}(2009)}]{Uhrig:09}%
  \BibitemOpen
  \bibfield  {author} {\bibinfo {author} {\bibfnamefont {G.~S.}\ \bibnamefont
  {Uhrig}},\ }\Doi {10.1103/PhysRevLett.102.120502} {\bibfield  {journal}
  {\bibinfo  {journal} {Phys. Rev. Lett.},\ }\textbf {\bibinfo {volume}
  {102}},\ \bibinfo {pages} {120502} (\bibinfo {year} {2009})}\BibitemShut
  {NoStop}%
\bibitem [{\citenamefont {Jiang}\ and\ \citenamefont
  {Imambekov}(2011)}]{JiangImambekov:11}%
  \BibitemOpen
  \bibfield  {author} {\bibinfo {author} {\bibfnamefont {L.}~\bibnamefont
  {Jiang}}\ and\ \bibinfo {author} {\bibfnamefont {A.}~\bibnamefont
  {Imambekov}},\ }\Doi {10.1103/PhysRevA.84.060302} {\bibfield  {journal}
  {\bibinfo  {journal} {Phys. Rev. A},\ }\textbf {\bibinfo {volume} {84}},\
  \bibinfo {pages} {060302} (\bibinfo {year} {2011})}\BibitemShut {NoStop}%
\bibitem [{\citenamefont {Kuo}\ and\ \citenamefont
  {Lidar}(2011)}]{WanLidar:11}%
  \BibitemOpen
  \bibfield  {author} {\bibinfo {author} {\bibfnamefont {W.-J.}\ \bibnamefont
  {Kuo}}\ and\ \bibinfo {author} {\bibfnamefont {D.}~\bibnamefont {Lidar}},\
  }\Doi {10.1103/PhysRevA.84.042329} {\bibfield  {journal} {\bibinfo  {journal}
  {Phys. Rev. A},\ }\textbf {\bibinfo {volume} {84}},\ \bibinfo {pages}
  {042329} (\bibinfo {year} {2011})}\BibitemShut {NoStop}%
\bibitem [{\citenamefont {Quiroz}\ and\ \citenamefont
  {Lidar}(2011)}]{QuirozLidar:11}%
  \BibitemOpen
  \bibfield  {author} {\bibinfo {author} {\bibfnamefont {G.}~\bibnamefont
  {Quiroz}}\ and\ \bibinfo {author} {\bibfnamefont {D.~A.}\ \bibnamefont
  {Lidar}},\ }\Doi {10.1103/PhysRevA.84.042328} {\bibfield  {journal} {\bibinfo
   {journal} {Phys. Rev. A},\ }\textbf {\bibinfo {volume} {84}},\ \bibinfo
  {pages} {042328} (\bibinfo {year} {2011})}\BibitemShut {NoStop}%
\bibitem [{\citenamefont {Khym}\ \emph {et~al.}(2006)\citenamefont {Khym},
  \citenamefont {Lee},\ and\ \citenamefont {Kang}}]{KHYMLEEKANG:06}%
  \BibitemOpen
  \bibfield  {author} {\bibinfo {author} {\bibfnamefont {G.~L.}\ \bibnamefont
  {Khym}}, \bibinfo {author} {\bibfnamefont {Y.}~\bibnamefont {Lee}}, \ and\
  \bibinfo {author} {\bibfnamefont {K.}~\bibnamefont {Kang}},\ }\Doi
  {10.1143/JPSJ.75.063707} {\bibfield  {journal} {\bibinfo  {journal} {Journal
  of the Physical Society of Japan},\ }\textbf {\bibinfo {volume} {75}},\
  \bibinfo {pages} {063707} (\bibinfo {year} {2006})}\BibitemShut {NoStop}%
\bibitem [{\citenamefont {Pan}\ \emph {et~al.}(2011)\citenamefont {Pan},
  \citenamefont {Xi},\ and\ \citenamefont {Gong}}]{PanXiGong:11}%
  \BibitemOpen
  \bibfield  {author} {\bibinfo {author} {\bibfnamefont {Y.}~\bibnamefont
  {Pan}}, \bibinfo {author} {\bibfnamefont {Z.-R.}\ \bibnamefont {Xi}}, \ and\
  \bibinfo {author} {\bibfnamefont {J.}~\bibnamefont {Gong}},\ }\Doi
  {10.1088/0953-4075/44/17/175501} {\bibfield  {journal} {\bibinfo  {journal}
  {J. Phys. B: At. Mol. Opt. Phys},\ }\textbf {\bibinfo {volume} {44}},\
  \bibinfo {pages} {175501} (\bibinfo {year} {2011})}\BibitemShut {NoStop}%
\bibitem [{\citenamefont {Mukhtar}\ \emph {et~al.}(2010)\citenamefont
  {Mukhtar}, \citenamefont {Soh}, \citenamefont {Saw},\ and\ \citenamefont
  {Gong}}]{Mukhtar2:10}%
  \BibitemOpen
  \bibfield  {author} {\bibinfo {author} {\bibfnamefont {M.}~\bibnamefont
  {Mukhtar}}, \bibinfo {author} {\bibfnamefont {W.~T.}\ \bibnamefont {Soh}},
  \bibinfo {author} {\bibfnamefont {T.~B.}\ \bibnamefont {Saw}}, \ and\
  \bibinfo {author} {\bibfnamefont {J.}~\bibnamefont {Gong}},\ }\Doi
  {10.1103/PhysRevA.82.052338} {\bibfield  {journal} {\bibinfo  {journal}
  {Phys. Rev. A},\ }\textbf {\bibinfo {volume} {82}},\ \bibinfo {pages}
  {052338} (\bibinfo {year} {2010})}\BibitemShut {NoStop}%
\bibitem [{\citenamefont {Grace}\ \emph {et~al.}(2010)\citenamefont {Grace},
  \citenamefont {Dominy}, \citenamefont {Kosut}, \citenamefont {Brif},\ and\
  \citenamefont {Rabitz}}]{Grace:10}%
  \BibitemOpen
  \bibfield  {author} {\bibinfo {author} {\bibfnamefont {M.~D.}\ \bibnamefont
  {Grace}}, \bibinfo {author} {\bibfnamefont {J.}~\bibnamefont {Dominy}},
  \bibinfo {author} {\bibfnamefont {R.~L.}\ \bibnamefont {Kosut}}, \bibinfo
  {author} {\bibfnamefont {C.}~\bibnamefont {Brif}}, \ and\ \bibinfo {author}
  {\bibfnamefont {H.}~\bibnamefont {Rabitz}},\ }\href
  {http://stacks.iop.org/1367-2630/12/i=1/a=015001} {\bibfield  {journal}
  {\bibinfo  {journal} {New Journal of Physics},\ }\textbf {\bibinfo {volume}
  {12}},\ \bibinfo {pages} {015001} (\bibinfo {year} {2010})}\BibitemShut
  {NoStop}%
\end{thebibliography}

%

\appendix

\section{Procedure to obtain $H_{(r_{1},r_{2},\dots, r_{\ell})}$ }
\label{app: step H}

Starting with the control pulse operator composing the first UDD layer, $\Omega_{1}$, the Hamiltonian can be partitioned into
\begin{equation}
H=\sum_{r_{1}=0,1}H_{(r_{1})}
\end{equation} 
where
\begin{equation}
H_{(r_{1})}\equiv \frac{H+(-1)^{r_{1}}\Omega_{1}H\Omega_{1}}{2}.
\label{H1}
\end{equation}
The Hamiltonian $H_{(0)}$ commutes with $\Omega_{1}$, while $H_{(1)}$ anticommutes with $\Omega_{1}$. Continuing the procedure for the next layer of NUDD, each $H_{(r_{1})}$, with $r_{1}=0\textrm{ or }1$, can be divided further into commuting and anticommuting parts for the control {pulse} operator $\Omega_{2}$ as 
\begin{equation}
H_{(r_{1}, r_{2})}\equiv \frac{H_{(r_{1})}+(-1)^{r_{2}}\Omega_{2}H_{(r_{1})}\Omega_{2}}{2}.
\end{equation}
Using Eq.~\eqref{H1} iteratively for the first $l$ layers, we have 
\begin{equation}
H_{(r_{1},\dots, r_{\ell})}\equiv [H_{(r_{1},\dots, r_{\ell-1})}+(-1)^{r_{\ell}}\Omega_{\ell}H_{(r_{1},\dots, r_{\ell-1})}\Omega_{\ell}]/2
\end{equation}
which commutes with $\Omega_{i}$ if $r_{i}=0$ and anticommutes with $\Omega_{i}$ if $r_{i}=1$, where $i\in\{1,\dots,\ell\}$.


\section{The outer-most UDD interval decomposition}
\label{app: outer decomposition}

We shall derive Eq.~\eqref{Fdecomposition} by splitting each integral of $F_{\oplus_{p=1}^{n}\vec{r}_{\ell}^{\,(p)}}$ [Eq.~\eqref{F}] into a sum of sub-integrals over the normalized outer-most layer intervals $s_{j_{\ell}}$. Since $F_{\oplus_{p=1}^{n}\vec{r}_{\ell}^{\,(p)}}$ comprises a series of time-ordered nested integrals, our procedure for decomposing $F_{\oplus_{p=1}^{n}\vec{r}_{\ell}^{\,(p)}}$ is to evaluate each nested integrals one by one, from $\eta^{(n)}$ to $\eta^{(1)}$.

We call the sub-integral over the $j_{\ell}^{\textrm{th}}$ outer-most interval ``sub-integral-$j_{\ell}$''. Suppose the integral of the integration variable $\eta^{(p)}$ follows the sub-integral-$j_{\ell}^{(p+1)}$ of the previous variable $\eta^{(p+1)}$. By splitting up the integral of $\eta^{(p)}$ with respect to the normalized outer-most intervals $s_{j_{\ell}}$, we have 
\begin{eqnarray}
\label{eq:B1}
&& \int_{0}^{\eta^{(p+1)}}\prod_{i=1}^{\ell}f_{i}(\eta^{\,(p)})^{r_{i}^{\,(p)}} d\eta^{\,(p)}\\
=&&\sum_{j_{\ell}^{(p)}=1}^{j_{\ell}^{(p+1)}-1}f_{\ell}(j_{\ell}^{(p)})^{r_{\ell}^{\,(p)}} \int_{\eta_{j^{(p)}_{\ell}-1}}^{\eta_{j^{(p)}_{\ell}}}\prod_{i=1}^{\ell-1}f_{i}(\eta^{\,(p)})^{r_{i}^{\,(p)}} d\eta^{\,(p)}\notag \\
+&&f_{\ell}(j_{\ell}^{(p+1)})^{r_{\ell}^{\,(p)}}\int_{\eta_{j^{(p+1)}_{\ell}-1}}^{\eta^{(p+1)}}\prod_{i=1}^{\ell-1}f_{i}(\eta^{\,(p)})^{r_{i}^{\,(p)}} d\eta^{\,(p)}\notag
\end{eqnarray}
where 
$\eta_{j^{(p)}_{\ell}}$ is UDD$_{N_{\ell}}$ pulse timing and $f_{\ell}(j_{\ell}^{(p)})=(-1)^{j_{\ell}^{(p)}-1}$.

For each sub-integral with the integration domain from $[\eta_{j^{(p)}_{\ell}-1},\eta_{j^{(p)}_{\ell}})$, with pulse interval $s_{j_{\ell}^{(p)}}$, we make the following linear change of variable,  
\begin{equation}
\tilde{\eta}^{(p)}=\frac{\eta^{(p)}-\eta_{j_{\ell}^{(p)}-1}}{s_{j_{\ell}^{(p)}}},
\end{equation} 
to normalize the  integration domain to $1$. Accordingly, each $f_{i}(\tilde{\eta}^{(p)})$ with $i \leq \ell-1$ becomes the normalized $i^{\textrm{th}}$-layer modulation function for $\ell-1$ layers of NUDD,
and is the same function for all of the outermost pulse intervals. Consequently {the summands in Eq.~\eqref{eq:B1} become}
\bes
\begin{align}
&\int_{0}^{1}d\tilde{\eta}^{(p)} \prod_{i=1}^{\ell-1}f_{i}(\tilde{\eta}^{(p)})^{r_{i}^{\,(p)}}\sum_{j_{\ell}^{(p)}=1}^{j_{\ell}^{(p+1)}-1}f_{\ell}(j_{\ell}^{(p)})^{r_{\ell}^{\,(p)}} s_{j_{\ell}^{(p)}}\label{breakinta}\\
&\ +\int_{0}^{\tilde{\eta}^{(p+1)}}d\tilde{\eta}^{(p)}\prod_{i=1}^{\ell-1}f_{i}(\tilde{\eta}^{(p)})^{r_{i}^{\,(p)}} f_{\ell}(j_{\ell}^{(p+1)})^{r_{\ell}^{\,(p)}}  s_{j_{\ell}^{(p+1)}}\label{breakintb}
\end{align}
\ees
where $\int_{0}^{1}\prod_{i=1}^{\ell-1}f_{i}(\tilde{\eta}^{(p)})^{r_{i}^{\,(p)}}\,d\tilde{\eta}^{(p)}$ is taken out of the summation.  

It is possible to rewrite Eq.~\eqref{breakinta} and Eq.~\eqref{breakintb} in similar form by introducing the configuration number $\xi_{p}$, where $\xi_{p}=0$ stands for Eq.~\eqref{breakinta} while $\xi_{p}=1$ stands for Eq.~\eqref{breakintb}. In this notation, Eq.~\eqref{breakinta} becomes
\begin{align}
&\int_{0}^{(\,\eta^{(p+1)}\,)^{\xi_{p}}}\prod_{i=1}^{\ell-1}f^{\prime}_{i}(\eta^{\,(p)})^{r_{i}^{\,(p)}} d\eta^{\,(p)} \times \notag \\
&\quad \sum_{j_{\ell}^{(p)}=1}^{j_{\ell}^{(p+1)}-(\xi_{p}\oplus 1)} f_{\ell}(j_{\ell}^{(p)})^{r_{\ell}^{\,(p)}}\,s_{j_{\ell}^{(p)}},
\label{inout1}
\end{align}
where $f^{\prime}_{i}(\eta^{\,(p)})\equiv f_{i}(\tilde{\eta}^{(p)})$,
and Eq.~\eqref{breakintb} with $\xi_{p}=1$  becomes
\begin{align}
&\int_{0}^{(\,\eta^{(p+1)}\,)^{\xi_{p}}}\prod_{i=1}^{\ell-1}f^{\prime}_{i}(\eta^{\,(p)})^{r_{i}^{\,(p)}} d\eta^{\,(p)} \times \notag \\
&\quad \sum_{j_{\ell}^{(p)}=\xi_{p}j_{\ell}^{(p+1)}}^{j_{\ell}^{(p+1)}-(\xi_{p}\oplus 1)} f_{\ell}(j_{\ell}^{(p)})^{r_{\ell}^{\,(p)}}\,s_{j_{\ell}^{(p)}},
\label{inout2}
\end{align}
where 
\begin{equation}
\sum_{j_{\ell}^{(p)}=\xi_{p}j_{\ell}^{(p+1)}}^{j_{\ell}^{(p+1)}-(\xi_{p}\oplus 1)} = \sum_{j_{\ell}^{(p)}=j_{\ell}^{(p+1)}}^{j_{\ell}^{(p+1)}}
\end{equation}
which means that $j_{\ell}^{(p)}=j_{\ell}^{(p+1)}$.

Note that both forms of Eqs.~\eqref{inout1} and ~\eqref{inout2} are naturally decomposed into an inner part, which is an integral of the first $\ell-1$ modulation functions for $(\ell-1)$-layer NUDD, and an outer part, which is sum over the outer-most ($\ell^{\textrm{th}}$) layer modulation function.

Each integral of  $F_{\oplus_{p=1}^{n}\vec{r}_{\ell}^{\,(p)}}$ can be split into  Eqs.~\eqref{inout1} and ~\eqref{inout2}, with one exception: if $j_{\ell}^{(p+1)}=1$ the subsequent sub-integrals from $\eta^{(p)}$ to $\eta^{(1)}$ will only contain Eq.~\eqref{inout2}, whose configuration number is $1$.

By substituting Eqs.~\eqref{inout1} and ~\eqref{inout2} into each integral of $F_{\oplus_{p=1}^{n}\vec{r}_{\ell}^{\,(p)}}$, in sequence from $\eta^{(n)}$ to $\eta^{(1)}$ (taking the exception into account) and collecting all the inner (outer) parts of all sub-integrals accordingly,  we obtain 
\begin{align}
\label{D1}
&F_{\oplus_{p=1}^{n}\vec{r}_{\ell}^{\,(p)}}= \\
&\quad \sum_{\{\xi_{k}=0,1\}_{k=1}^{n-1}}\Phi^{\rm in}_{\xi_{n}\dots \xi_{1}}(\{\vec{r}_{\ell-1}^{\,(p)}\}_{p=1}^{n})\,
\Phi^{\rm out}_{\xi_{n}\dots \xi_{1}}(\{r_{\ell}^{\,(p)}\}_{p=1}^{n}),\notag
\end{align}
where 
\begin{align}
&\Phi^{\rm in}_{\xi_{n}\dots \xi_{1}}(\{\vec{r}_{\ell-1}^{\,(p)}\}_{p=1}^{n})=\\
&\quad \prod_{p=1}^{n}\int_{0}^{(\,\eta^{(p+1)}\,)^{\xi_{p}}}\prod_{i=1}^{\ell-1}f^{\prime}_{i}(\eta^{\,(p)})^{r_{i}^{\,(p)}} d\eta^{\,(p)},\notag
\label{innerphi}
\end{align}
and
\begin{eqnarray}
&&\Phi^{\rm out}_{\xi_{n}\dots \xi_{1}}(\{r_{\ell}^{\,(p)}\}_{p=1}^{n})=\sum_{j_{\ell}^{(n)}=\sum_{k=1}^{n}(\xi_{k}\oplus 1)}^{N_{\ell}+ 1} f_{\ell}(j_{\ell}^{(n)})^{r_{\ell}^{\,(n)}}\,s_{j_{\ell}^{(n)}}\times\notag\\
&&\prod_{p=1}^{n-1}  \sum_{j_{\ell}^{(p)}=\xi_{p}j_{\ell}^{(p+1)}+(\xi_{p}\oplus 1)\sum_{k=1}^{p}(\xi_{k}\oplus 1)}^{j_{\ell}^{(p+1)}-(\xi_{p}\oplus 1)} f_{\ell}(j_{\ell}^{(p)})^{r_{\ell}^{\,(p)}}\,s_{j_{\ell}^{(p)}},
\label{outerphi}
\end{eqnarray} 
with $\sum_{\{\xi_{k}=0,1\}_{k=1}^{n-1}}$ summing over all possible nested integration (sum) configurations $\xi_{n}\dots\xi_{2}\,\xi_{1}$  with $\xi_{n}\equiv 0$  for the inner part $\Phi^{\rm in}_{\xi_{n}\dots \xi_{1}}$ (the outer part $\Phi^{\rm out}_{\xi_{n}\dots \xi_{1}}$).

Note that in Eq.~\eqref{outerphi} for $\xi_{p}=1$,
\begin{equation}
\sum_{j_{\ell}^{(p)}=j_{\ell}^{(p+1)}+(1\oplus 1)\sum_{k=1}^{p}(\xi_{k}\oplus 1)}^{j_{\ell}^{(p+1)}-(1\oplus 1)} = \sum_{j_{\ell}^{(p)}=j_{\ell}^{(p+1)}}^{j_{\ell}^{(p+1)}},
\label{xi1}
\end{equation}
which means that the variables $\eta^{(p)}$ and $\eta^{(p+1)}$ in Eq.~\eqref{F} are in the same  interval $j_{\ell}^{(p)}=j_{\ell}^{(p+1)}$, while for $\xi_{p}=0$
\begin{equation}
\sum_{j_{\ell}^{(p)}=0+(0\oplus 1)\sum_{k=1}^{p}(\xi_{k}\oplus 1)}^{j_{\ell}^{(p+1)}-(0\oplus 1)}= \sum_{j_{\ell}^{(p)}=\sum_{k=1}^{p}(\xi_{k}\oplus 1)}^{j_{\ell}^{(p+1)}-1}.
\label{xi0}
\end{equation}
In the latter case, the variable $\eta^{(p)}$ in Eq.~\eqref{F} is located in the earlier interval than the variable $\eta^{(p+1)}$, namely, $j_{\ell}^{(p)}<j_{\ell}^{(p+1)}$. $\sum_{k=1}^{p}(\xi_{k}\oplus 1)$ counts the number of $\xi_{k}=0$ from $k=1$ to $k=p$.

Furthermore, each $\Phi^{\rm in}_{\xi_{n}\dots \xi_{1}}(\{\vec{r}_{\ell-1}^{\,(p)}\}_{p=1}^{n})$ is a product of multiple nested integrals with the modulation functions of $\ell-1$ layers of NUDD as integrands. In fact, each nested integral contained in $\Phi^{\rm in}_{\xi_{n}\dots \xi_{1}}(\{\vec{r}_{\ell-1}^{\,(p)}\}_{p=1}^{n})$ is actually one of the $(\ell-1)$-layers NUDD coefficients. For example, an inner term appearing in the expansion of the $8^{\textrm{th}}$ order  $\ell$-layers NUDD coefficients $F_{\oplus_{p=1}^{8}\vec{r}_{\ell}^{\,(p)}}$ reads
\begin{eqnarray}
\Phi^{\rm in}_{01101011}(\,\{\vec{r}_{\ell-1}^{\,(p)}\}_{p=1}^{8}\,)&=&F_{\oplus_{p=6}^{8}\vec{r}_{\ell-1}^{\,(p)}}F_{\oplus_{p=4}^{5}\vec{r}_{\ell-1}^{\,(p)}}F_{\oplus_{p=1}^{3}\vec{r}_{\ell-1}^{\,(p)}}\notag\\
&=& F^{(3)}_{\vec{r}_{\ell-1}^{\,'''}}F^{(2)}_{\vec{r}_{\ell-1}^{\,''}}F^{(3)}_{\vec{r}_{\ell-1}^{\,'}}
\label{inner ex}
\end{eqnarray}
where $\vec{r}_{\ell-1}^{\,'''}=\oplus_{p=6}^{8}\vec{r}_{\ell-1}^{\,(p)}$, $\vec{r}_{\ell-1}^{\,''}=\oplus_{p=4}^{5}\vec{r}_{\ell-1}^{\,(p)}$, and $\vec{r}_{\ell-1}^{\,'}=\oplus_{p=1}^{3}\vec{r}_{\ell-1}^{\,(p)}$.
Its corresponding outer part reads
\begin{eqnarray}
&&\Phi^{\rm out}_{01101011}(\,\{r_{\ell}^{\,(p)}\}_{p=1}^{8}\,)
=\sum_{j_{\ell}^{(8)}=3}^{N_{\ell}+ 1} f_{\ell}(j_{\ell}^{(8)})^{\oplus_{p=6}^{8}r_{\ell}^{\,(p)}}\,s_{j_{\ell}^{(8)}}^3 \times \notag\\
&&\sum_{j_{\ell}^{(5)}=2}^{j_{\ell}^{(8)}-1} f_{\ell}(j_{\ell}^{(5)})^{\oplus_{p=4}^{5}r_{\ell}^{\,(p)}}\,s_{j_{\ell}^{(5)}}^2\times
\sum_{j_{\ell}^{(3)}=1}^{j_{\ell}^{(5)}-1} f_{\ell}(j_{\ell}^{(3)})^{\oplus_{p=1}^{3}r_{\ell}^{\,(p)}}\,s_{j_{\ell}^{(3)}}^3\notag\\
\label{outer ex}
\end{eqnarray}
where we used Eq.~\eqref{Z2}.

As suggested by Eqs.~\eqref{inner ex} and ~\eqref{outer ex}, one can see that each configuration $\xi_{n}\dots\xi_{1}$ defines a  way to separate the $n$ vectors $\{\vec{r}_{\ell}^{\,(p)}\}_{p=1}^{n}$ into several clusters with order configuration numbers and error vectors as $\xi_{n}\vec{r}_{\ell}^{\,(n)}\dots \xi_{1}\vec{r}_{\ell}^{\,(1)}$. A cluster of vectors is defined as a contiguous set of vectors only connected by configuration numbers whose value is 1. Different clusters are separated by configuration numbers whose value is 0. For  $\xi_{8}\dots\xi_{1}=01101011$ in Eqs.~\eqref{inner ex} and ~\eqref{outer ex} as an example, we have 
\begin{equation}
0\,\vec{r}_{\ell}^{\,(8)}1\vec{r}_{\ell}^{\,(7)}1\vec{r}_{\ell}^{\,(6)}\,0\,\vec{r}_{\ell}^{\,(5)}1\vec{r}_{\ell}^{\,(4)}\,0\,\vec{r}_{\ell}^{\,(3)}1\vec{r}_{\ell}^{\,(2)}1\vec{r}_{\ell}^{\,(1)} 
\end{equation}
which divides the $8$ vectors into three clusters. Suppose for a given configuration $\xi_{n}\dots\xi_{2}\,\xi_{1}$, which separates the $n$ vectors into $m$ clusters, the $a^{\textrm{th}}$ cluster (counting from right to left)  contains  $n_{a}$ vectors $\{\vec{r}_{\ell}^{\,(p)}\}_{p=\nu}^{\mu}$ with $n_{a}=\mu-\nu+1$. Then, according to the result of $\oplus_{p=\nu}^{\mu}\vec{r}_{\ell}^{\,(p)}$, its corresponding inner part is $F^{(n_{a})}_{\vec{r}_{\ell-1}^{\,\<a\>}}$  with $\vec{r}_{\ell-1}^{\,\<a\>}\equiv \oplus_{p=\nu}^{\mu}\vec{r}_{\ell-1}^{\,(p)}$ and the corresponding outer part is $\sum_{j_{\ell}^{(a)}=a}^{j_{\ell}^{(a+1)}-1}f_{\ell}(j_{\ell}^{(a)})^{r_{\ell}^{\,\<a\>}}\,s_{j_{\ell}^{(a)}}^{\,n_{a}}$
with $r_{\ell}^{\,\<a\>}\equiv \oplus_{p=\nu}^{\mu}r_{\ell}^{\,(p)}$. Note that for the last cluster, $a=m$, the upper bound of the sum $j_{\ell}^{(m+1)}-1$ is replaced by $N_{\ell}+1$.
Therefore, Eq.~\eqref{D1} can be re-expressed in the more compact form of Eq.~\eqref{Fdecomposition}.


\section{$F^{(n)}_{\vec{r}_{\ell}}$ in terms of $F^{(n')}_{\vec{r'}_{\ell-1}}$}
\label{app: decomplemma}

We prove Lemma \ref{lem:innersum N}: $\check{N}_{\vec{r}_{\ell-1}} \leq \sum_{a=1}^{m}\check{N}_{\vec{r}_{\ell-1}^{\,\<a\>}}$, where  $\vec{r}_{\ell-1}=\oplus_{a=1}^{m}\vec{r}_{\ell-1}^{\,\<a\>}$.

\begin{proof}[Proof of Lemma \ref{lem:innersum N}]

For a given $\vec{r}_{\ell-1}$-type error with $\vec{r}_{\ell-1}=\oplus_{a=1}^{m}\vec{r}_{\ell-1}^{\,\<a\>}$ [Eq.~\eqref{r l-1}], its decoupling order is [Eq.~\eqref{Nrl}] 
\begin{equation}
\check{N}_{\vec{r}_{\ell-1}}= \max_{i\in\{1,\dots,\ell-1\}}[r_{i}(p_\oplus(1,i-1)\oplus 1)\widetilde{N}_{i}+p_+(i+1,\ell-1)]
\end{equation}
Suppose the maximum  occurs at the $M^{\textrm{th}}$ UDD layer, i.e.,
\begin{equation}
\check{N}_{\vec{r}_{\ell-1}}=\widetilde{N}_{M}+  p_+(M+1,\ell-1)
\end{equation}
where the coefficient of $\widetilde{N}_{M}$ is equal to 1, implying that 
\bes
\begin{eqnarray}
&&r_{M}=1 \label{r M}\\
&& p_\oplus(1,M-1)=0 \label{r 1 to M-1}.
\end{eqnarray}
\ees
With  Eq.~\eqref{r l-1}, Eq.~\eqref{r M}, Eq.~\eqref{r 1 to M-1} , we have the following equalities and inequality.  First,
\begin{equation}
\oplus_{a=1}^{m}r_{M}^{\,\<a\>}=r_{M}=1, 
\label{sumra=rM}
\end{equation}
second,
\begin{eqnarray}
p_\oplus(1,M-1)&=&\oplus_{k=1}^{M-1}r_{k}[N_{k}]_{2}=\oplus_{k=1}^{M-1}(\oplus_{a=1}^{m}r_{k}^{\,\<a\>})[N_{k}]_{2}\notag\\
&=&\oplus_{a=1}^{m}(\oplus_{k=1}^{M-1}r_{k}^{\,\<a\>}[N_{k}]_{2})\notag\\
&=&\oplus_{a=1}^{m}|\nu_{M}^{\<a\>}|,
\end{eqnarray}
where to keep the notation more compact we defined
\beq
|\nu_{i}^{\<a\>}|\equiv\oplus_{k=1}^{i-1}r^{\<a\>}_{k}[N_{k}]_{2} .
\eeq
Thus, using Eq.~\eqref{r 1 to M-1}
\begin{equation}
\oplus_{a=1}^{m}|\nu_{M}^{\<a\>}|=0 ,
\label{v}
\end{equation}
and third,
\begin{eqnarray}
p_+(M+1,\ell-1)&=&\sum_{k=M+1}^{\ell-1}r_{k}[N_{k}]_{2}\notag \\
&=&\sum_{k=M+1}^{\ell-1}(\oplus_{a=1}^{m}r_{k}^{\,\<a\>})[N_{k}]_{2}\notag\\
&\leq& \sum_{k=M+1}^{\ell-1}\sum_{a=1}^{m}r_{k}^{\,\<a\>}[N_{k}]_{2}\notag\\
&\leq& \sum_{a=1}^{m}\sum_{k=M+1}^{\ell-1}r_{k}^{\,\<a\>}[N_{k}]_{2}\label{r M to l}.
\end{eqnarray}
By using the properties
$\max[A,B,C]\geq \max[A,B]\label{max1}$ and 
$\max[A+c, B+c]=\max[A, B]+c$
we have
\begin{eqnarray}
&&\sum_{a=1}^{m}\check{N}_{\vec{r}_{\ell-1}^{\,\<a\>}}=\notag\\
&&\sum_{a=1}^{m}\max[\{\,r_{i}^{\,\<a\>}(|\nu^{\<a\>}_{i}|\oplus 1 )\widetilde{N}_{i}+\sum_{k=i+1}^{\ell-1}r_{k}^{\,\<a\>}[N_{k}]_{2}\,\}_{i=1}^{\ell-1}]\notag\\
&&\geq \sum_{a=1}^{m}\max[\{\,r_{i}^{\,\<a\>}(|\nu^{\<a\>}_{i}|\oplus 1 )\widetilde{N}_{i}+\sum_{k=i+1}^{\ell-1}r_{k}^{\,\<a\>}[N_{k}]_{2}\,\}_{i=1}^{M}]\notag\\
&&\geq \sum_{a=1}^{m}\max[\{\,r_{i}^{\,\<a\>}(|\nu^{\<a\>}_{i}|\oplus 1 )\widetilde{N}_{i}+\sum_{k=i+1}^{M}r_{k}^{\,\<a\>}[N_{k}]_{2}\,\}_{i=1}^{M}] \notag\\
&&+\sum_{a=1}^{m}\sum_{k=M+1}^{\ell-1}r_{k}^{\,\<a\>}[N_{k}]_{2}\notag\\
\end{eqnarray}
Equivalently, 
\begin{equation}
\sum_{a=1}^{m}\check{N}_{\vec{r}_{\ell-1}^{\,\<a\>}}\geq\sum_{a=1}^{m}\check{N}_{\vec{r}_{M}^{\,\<a\>}} +\sum_{a=1}^{m}\sum_{k=M+1}^{\ell-1}r_{k}^{\,\<a\>}[N_{k}]_{2}.
\label{sumNa}
\end{equation}
In light of the inequality~\eqref{r M to l}, the remaining task is to show 
\begin{equation}
\sum_{a=1}^{m}\check{N}_{\vec{r}_{M}^{\,\<a\>}} \geq \widetilde{N}_{M}
\label{sumNa>NM}
\end{equation}
 which will immediately lead to Lemma \ref{lem:innersum N}: $\sum_{a=1}^{m}\check{N}_{\vec{r}_{\ell-1}^{\,\<a\>}}\geq \check{N}_{\vec{r}_{\ell-1}}$ .

The condition $\oplus_{a=1}^{m}r_{M}^{\,\<a\>}=r_{M}=1$ [Eq.~\eqref{sumra=rM}] implies that  there must exist at least one $a'$ such that $r_{M}^{\<a'\>}=1$. For this $\vec{r}_{M}^{\<a'\>}$ vector with $r_{M}^{\<a'\>}=1$, its first $M-1$ components either satisfy  $|\nu^{\<a'\>}_{M}|=0$ or  $|\nu^{\<a'\>}_{M}|=1$.

For the $\vec{r}_{M}^{\,\<a\>}$-type  error with $r_{M}^{\<a'\>}=1$ and $|\nu^{\<a'\>}_{M}|=0$,  the coefficient $r_{M}^{\<a'\>}(|\nu^{\<a'\>}_{M}|\oplus 1)$ of $\widetilde{N}_{M}$ in $\check{N}_{\vec{r}_{M}^{\<a'\>}}$ [Eq.~\eqref{Nrl}] is not zero. Therefore, 
\begin{equation}
\check{N}_{\vec{r}_{M}^{\<a'\>}}= \max[\{\cdots\}_{i=1}^{M-1},\widetilde{N}_{M}]\geq \widetilde{N}_{M}
\end{equation}
which leads to Eq.~\eqref{sumNa>NM}.

However, for  the $\vec{r}_{M}^{\,\<a\>}$-type  error with $r_{M}^{\<a'\>}=1$ and $v^{\<a'\>}_{M}=1$, the coefficient $r_{M}^{\<a'\>}(|\nu^{\<a'\>}_{M}|\oplus 1)$ of $\widetilde{N}_{M}$ in Eq.~\eqref{Nrl} ends up vanishing, leading to \begin{equation}
\check{N}_{\vec{r}_{M}^{\<a'\>}}= \max[\{\cdots\}_{i=1}^{M-1},0]
\end{equation} 
which cannot determine whether $\check{N}_{\vec{r}_{M}^{\<a'\>}}\geq \widetilde{N}_{M}$. This notwithstanding, $v^{\<a'\>}_{M}\equiv\oplus_{k=1}^{M-1}r_{k}^{\<a'\>}[N_{k}]_{2}=1$ indicates that among the $M-1$ components of the vector $\vec{r}_{M}^{\<a'\>}$, there is a total odd number of components each of which is equal to 1, and associates with an odd sequence order.  Suppose the $o'$ component is the outermost nonzero component $r_{o'}=1$ with $[N_{o'}]_{2}=1$ of  $\vec{r}_{M}^{\<a'\>}$. Accordingly, 
\bes
\begin{eqnarray}
&&|\nu^{\<a'\>}_{o'}|=\oplus_{k=1}^{o'-1}r_{k}^{\<a'\>}[N_{k}]_{2}=0\label{o'1}\\
&&\sum_{k=o'+1}^{M}r_{k}^{\<a'\>}[N_{k}]_{2}=0
\label{o'2}
\end{eqnarray}
\ees
With  Eq.~\eqref{o'1} and Eq.~\eqref{o'2}, it follows that
\begin{equation}
\check{N}_{\vec{r}_{M}^{\<a'\>}}=\max[\{\cdots\}_{i=1}^{o'-1},\widetilde{N}_{o'}, \{\cdots\}_{i=o'+1}^{M-1}]\geq \widetilde{N}_{o'}\geq \widetilde{N}_{M}-1,
\label{No1}
\end{equation}
where the last inequality is derived as follows:
\begin{eqnarray}
\widetilde{N}_{M}&&=\min[N_{M}, N^{k<M}_{o_{\min}}+1]=\min[N_{M}, N^{k<o'}_{o_{\min}}+1]\notag\\
&&\leq N^{k<o'}_{o_{\min}}+1=\min[N^{k<o'-1}_{o_{\min}}+1, N_{o'}+1]\notag\\
&&\leq\widetilde{N}_{o'}+1
\label{NM<No+1}
\end{eqnarray}

Moreover, if  $v^{\<a'\>}_{M}=1$, Eq.~\eqref{v} implies that there is another vector $\vec{r}_{M}^{<a''>}$ such that  $|\nu^{\<a''\>}_{M}|\equiv\oplus_{k=1}^{M-1}r_{k}^{\<a''\>}[N_{k}]_{2}{=1}$. Using the same argument, there must exist a $r_{o''}=1$ with $v^{\<a"\>}_{o''}=0$ with an odd sequence order $N_{o''}$, $o''<M$. It follows that 
\begin{equation}
\check{N}_{\vec{r}_{M}^{\<a''\>}}=\max[\{\cdots\}_{i=1}^{o''-1},\widetilde{N}_{o''}, \{\cdots\}_{i=o''+1}^{M}]\geq \widetilde{N}_{o''}\geq 1.
\label{No2}
\end{equation}

Therefore, with  Eqs.~\eqref{No1} and \eqref{No2},
\begin{equation}
\sum_{a=1}^{m}\check{N}_{\vec{r}_{M}^{\,\<a\>}} \geq \check{N}_{\vec{r}_{M}^{\<a'\>}}+\check{N}_{\vec{r}_{M}^{\<a''\>}} \geq \widetilde{N}_{M}
\end{equation}
which proves Eq.~\eqref{sumNa>NM}.

\end{proof}

\section{Fourier expansions after linear change of variable}
\label{app:Fourier}

The $\ell^{\textrm{th}}$-layer modulation function $f_{\ell}(\theta )$, which alternates in sign with the new outer-most-layer pulse timing $\theta_{j_{\ell}}=\frac{j_{\ell}\pi}{N_{\ell}+1}$, has a period of $\frac{2\pi }{^{N_{\ell}+1}}$ since the outer-most-layer pulses are now  equally spaced.  Therefore,  $f_{\ell}(\theta )$ has a simple Fourier expansion
\begin{equation}
f_{\ell}(\theta )=\sum_{k=0}^{\infty }d_{k}^{\ell} \sin [(2k+1)(N_{\ell}+1)\theta ],
\label{Fourier fl}
\end{equation}
with $d_{k}^{\ell} = \frac{4}{(2k+1)\pi}$. The remaining modulation functions such as $f_{i}(\theta)$ with $i<\ell$, which switch sign with the inner-layer UDD pulse timing $\theta_{\,j_{\ell},\dots,j_{i}}$, all have the same period $\frac{\pi }{^{N_{\ell}+1}}$. This is because the inner ($\ell-1$)-layers NUDD pulse timing structure inside each $[\theta_{j_{\ell}-1}, \theta_{j_{\ell}})$ is still preserved and each are actually identical to each other, in that Eq.~\eqref{linear} just rescales  the outer-most layer interval $[\theta_{j_{\ell}-1}, \theta_{j_{\ell}})$ linearly.  In addition, due to the time-symmetric structure of NUDD pulse timings, it follows that inside each $[\theta_{j_{\ell}-1}, \theta_{j_{\ell}})$, $f_{i}(\theta)$ for $i<\ell$ is an even function when $N_{i}$ is even, and odd when $N_{i}$ is odd. Therefore, it follows that  the Fourier expansion of $f_{i}(\theta )$  has the following form:
\begin{eqnarray}
f_{i}(\theta )=\begin{dcases}
               \sum_{k=0}^{\infty}d_{k}^{i}\cos[2k(N_{\ell}+1)\theta ], & N_{i} \textrm{ even} \\
               \sum_{k=1}^{\infty}d_{k}^{i}\sin[2k(N_{\ell}+1)\theta ], & N_{i} \textrm{ odd} 
               \end{dcases}
\label{Fourier fi}
\end{eqnarray}
for $i<\ell$. Note that the Fourier expansion coefficients $d_{k}^{i}$ in the even and odd cases are, in fact, different.  However, we use the same notation for both since the exact values of these coefficients are irrelevant for our proof. Finally, the Fourier expansion of $G_{1}(\theta)$ is
\begin{equation}
G_{1}(\theta)=\sum_{k=0}^{\infty }\sum_{q=-1,1}g_{k,q}\sin [2k(N_{\ell}+1)\theta +q\theta ]
\label{G1}
\end{equation}
whose detailed calculations can be found in Ref. \cite{WanLidar:11}. In accordance with Definition \ref{def:type}, the function types of $f_{\ell}(\theta )$, $f_{\Omega_{i<\ell}}(\theta )$, and $G_{1}(\theta)$ [Eq.~\eqref{Fourier fl}-~\eqref{G1}] are identified as Eqs ~\eqref{psifl}-~\eqref{psiG}, respectively.

\vspace{1cm}
\section{$\min[\{\check{N}_{\vec{r}_{\ell-1}\in \vec{v}_{\ell-1}^{\,1}}\}]=N^{k'<\ell}_{o_{\min}} $}
\label{app:min Nv1}

We prove Lemma \ref{lem: min Nv1}: $\min[\{\check{N}_{\vec{r}_{\ell-1}\in \vec{v}_{\ell-1}^{\,1} }\}]=N^{k'<\ell}_{o_{\min}} $.

\begin{proof}[Proof of Lemma \ref{lem: min Nv1}]

For all $\vec{r}_{\ell-1}\in \vec{v}_{\ell-1}^{\,1} $, their components satisfy, by definition, the condition $p_\oplus(1,\ell-1)=1$, which implies that the total number of non-zero components with odd sequence orders in $\vec{r}_{\ell-1}$ is odd. Suppose that before the $\ell$th-layer, the $o_{1}^{\rm th}$, $o_{2}^{\rm th}$, $\dots$, and $o_{j}^{\rm th}$ UDD layers where $0<o_{1}<o_{2}<\dots<o_{j}<\ell$ are the layers with odd sequence orders, namely, $[N_{o_{i}}]_{2}=1$ for $i=1,2,\dots, j$. 

Define  $\vec{e}_{o_{i}}$ as the $\vec{r}_{\ell-1}$-type errors with  $r_{o_{i}}=1$ and all $r_{j\neq o_{i}}=0$. For $\vec{e}_{o_{i}}$-type errors which obviously belong to the $\vec{v}_{\ell-1}^{\,1}$-type error,   we have $\check{N}_{\vec{e}_{o_{i}} }=\widetilde{N}_{o_{i}}$ according to  the decoupling order formula Eq.~\eqref{Nrl}.

For the remaining $\vec{v}_{\ell-1}^{\,1} $ vectors, they have  3 or more non-zero components associated with odd orders. For a given $\vec{r}_{\ell-1}\neq \vec{e}_{o_{i}}\in \vec{v}_{\ell-1}^{\,1} $, suppose the inner-most component with odd order is $r_{o_{i}}=1$.  Then its decoupling order follows 
\begin{eqnarray}
\check{N}_{\vec{r}_{\ell-1}\neq \vec{e}_{o_{i}}\in \vec{v}_{\ell-1}^{\,1} }&&=\max[\dots,\widetilde{N}_{o_{i}}+\sum_{k=o_{i}+1}^{\ell-1}r_{k}[N_{k}]_{2}, \dots]\notag\\
&&\geq \widetilde{N}_{o_{i}}=\check{N}_{\vec{e}_{o_{i}} }
\end{eqnarray}

Therefore, $\min[\{\check{N}_{\vec{r}_{\ell-1}\in \vec{v}_{1} }\}]$ will occur among the decoupling orders of the errors which has only one non-zero component with odd sequence order, i.e.,
\begin{eqnarray}
\min[\{\check{N}_{\vec{r}_{\ell-1}\in \vec{v}_{\ell-1}^{\,1}}\}]&=&\min[\check{N}_{\vec{e}_{o_{1}}},\check{N}_{\vec{e}_{o_{2}}},\dots,\check{N}_{\vec{e}_{o_{j}}}]\notag\\
&=&\min[\widetilde{N}_{o_{1}},\widetilde{N}_{o_{2}},\dots,\widetilde{N}_{o_{j}}]\notag \label{minNodd1}\\
&=& \min[N_{o_{1}},N_{o_{2}},\dots,N_{o_{j}}] \label{minNodd2}
\end{eqnarray}
Eq.~\eqref{minNodd2} is the same expression as {$N^{k'<i}_{o_{\min}}\equiv\min\{N_{k'}\ |\ [N_{k'}]_{2}=1\}$} [Eq.~\eqref{Nioddmin}], which completes the proof of Lemma \ref{lem: min Nv1}.

\end{proof}

\end{document}